\documentclass[a4paper,11pt, fleqn]{article}

\usepackage[a4paper, total={6.5in, 10in}]{geometry}
\usepackage{setspace} 

\usepackage{enumitem}
\newcommand{\E}{\mathbb{E}}
\newcommand{\abs}[1]{\left\lvert#1\right\rvert}

\usepackage[dvipsnames,svgnames]{xcolor}

\usepackage{amsmath,amsfonts,amssymb,amsthm}
\usepackage{dsfont}
\usepackage{booktabs}
\DeclareMathOperator*{\argmin}{arg\,min}
\usepackage{algorithm}
\usepackage{algpseudocode}
\usepackage{witharrows}
\usepackage{tikz}
\usetikzlibrary{arrows.meta, positioning, fit, backgrounds}
\usepackage{float}
\usepackage{authblk}
\usepackage{comment}
\usepackage{natbib}
\usepackage{longtable}
\usepackage{threeparttable}
\usepackage{tabularx}
\usepackage{ragged2e}
\usepackage{endnotes}
\usepackage{array}
\usepackage[hyphens]{url}
\usepackage[
	breaklinks=true,	
	colorlinks=true,	
	citecolor=blue, 
	linkcolor=blue,
	urlcolor=blue,	
	extension=pdf,		
	ngerman]{hyperref}
\usepackage{graphicx}
\usepackage{subcaption}

\theoremstyle{definition}
\newtheorem{example}{Example}
\theoremstyle{remark}
\newtheorem{remark}{Remark}
\theoremstyle{plain}
\newtheorem{theorem}{Theorem}
\newtheorem{lemma}{Lemma}
\newtheorem{corollary}{Corollary}
\newtheorem{assumption}{Assumption}

\title{Sparse Tree-Based Aggregation for Time Series Regressions}

\author[1]{Marie Corillon}
\author[1]{Stephan Smeekes}
\author[1]{Ines Wilms}
\affil[1]{Department of Quantitative Economics, Maastricht University, The Netherlands}

\date{July 1, 2026}

\begin{document}

\maketitle

\onehalfspacing

\begin{abstract}
High-dimensional time series regressions are often regularized to produce sparse coefficients. We show that temporal aggregation provides a powerful alternative to reduce dimensionality in high-order autoregressions and mixed-frequency regressions. To this end, we propose StarTime (Sparse Tree-based Aggregation for Time Series), a convex penalization method that uses a temporal tree to arrange lags hierarchically from high to low frequency. StarTime then flexibly selects coefficients to be aggregated at possibly varying frequencies, sparse or a combination thereof. We provide new  error bounds for StarTime, demonstrate improved estimation accuracy and recovery of aggregation and sparsity in simulations relative to benchmarks, and illustrate StarTime's relevance for financial and macroeconomic applications.
\end{abstract} 


\paragraph{\textbf{JEL-Classification.}} C22, C53, C55

\doublespacing
\section{Introduction}

Forecasting financial and macroeconomic variables increasingly requires modeling dynamics across multiple time periods and temporal resolutions.
Daily realized volatility, for example, displays pronounced serial correlation; current outcomes depend on a long history of past observations (see, e.g., \citealp{Corsi2009}). Besides, empirical research increasingly requires mixed-frequency settings where higher-frequency predictors forecast a lower-frequency response; the MIxed DAta Sampling (MIDAS) framework \citep{MIDAStouch} is a prominent example. In both examples, the persistence and frequency mismatch naturally lead to high-dimensional models with many highly correlated regressors. In this paper, we propose the Sparse Tree-Based Aggregation for Time Series Regressions (\textit{StarTime}) penalized estimator that addresses the dimensionality in such models through data-driven temporal aggregation. It determines the sampling frequency at which lagged predictors enter the model, allowing for coefficients to be sparse, aggregated at varying lower frequencies, or both. 

Time series are nowadays recorded at relatively fine temporal resolutions (e.g., daily financial or weekly macroeconomic series). Modeling such series with high-order autoregressions or mixed-frequency regressions quickly gives rise to the curse of dimensionality. First, strong serial dependence implies that predictive information may be distributed across many, potentially noisy lags. Second, in mixed-frequency regressions, the problem is further amplified by the frequency mismatches: A single lower-frequency observation corresponds to numerous higher-frequency predictors, each entering as a separate regressor. As the frequency mismatch widens (e.g., daily-to-quarterly versus monthly-to-quarterly set-up), the dimensionality can increase substantially. 

Penalized estimators such as the Lasso and extensions have become popular to address dimensionality in time series regressions (see e.g.; \citealp{luo2025midas, babii2024high, masini2023machine, TernesHierarchical, Babii2022, smeekes2021automated, mogliani2021bayesian}). These are particularly effective when the underlying model is sparse. Persistent and mixed-frequency environments, however, differ in important ways. When dynamics are smooth and distributed over time, the true lag structure may be dense rather than sparse with highly correlated adjacent lags. As more lags are included, the design matrix becomes increasingly collinear, undermining variable selection methods that operate on individual coefficients. Penalized methods may then overshrink lags or select unstable subsets of lags, failing to capture the underlying temporal structure (see, e.g., \citealp{Zou2005}).

Temporal aggregation offers an appealing alternative when dynamics are smooth and densely distributed over time. Temporal aggregation in time series has long been studied (e.g., \citealp{Silvestrini2008, Rossana1995, Brewer1973, Amemiya1972,Tiao1972}). In finance, the Heterogeneous Autoregressive (HAR) model \citep{Corsi2009} offers a prominent, popular example: It aggregates daily realized volatility into weekly and monthly components, motivated by investor heterogeneity across trading horizons. It demonstrates strong forecasting performance, often rivaling or outperforming more complex machine learning methods such as Lasso \citep{HARdToBeat}. Similarly, aggregation-based measures also successfully  mitigate microstructure noise in high-frequency financial data (e.g., \citealp{Bollerslev2006, Todorov2012} ). 

In macroeconomics, mixed-frequency models also commonly rely on aggregation schemes to align predictors sampled at different rates. MIDAS regressions \citep{MIDAStouch}  employ polynomial distributed lags, while restricted MIDAS \citep{Andreou2016} and bridge equations \citep{Schumacher2016}  impose shape restrictions or fixed weighting schemes to simplify estimation. They are widely adopted in macroeconomic and financial research \citep{ghysels2007midas}. While these methods effectively reduce noise and dimensionality, they typically require the aggregation scheme to be specified a priori. In practice, however, the appropriate temporal resolution may not be known in advance. This motivates the need for a flexible method that determines the temporal aggregation structure in a data-driven way; we fill this gap.

In this paper, we propose \textit{StarTime} (Sparse Tree-Based Aggregation for Time Series Regressions), a penalized estimation procedure designed to reduce dimensionality in autoregressive and mixed-frequency regressions through data-driven temporal aggregation. StarTime allows adjacent lagged covariates to share a common coefficient, thereby selecting the sampling frequency at which lagged effects enter the model. It can produce temporally aggregated coefficients at varying lower frequencies, sparse coefficients, or a combination of both, enabling fully data-driven selection of the effective temporal resolution. To implement this idea, we leverage a temporal tree that reflects natural temporal aggregation. For each predictor, lags are arranged hierarchically from its highest available frequency (leaf nodes) toward progressively lower sampling frequencies (internal  and root nodes). We embed the tree  into a convex penalized optimization problem, encouraging nearby lags to fuse when supported by the data. When the underlying lag pattern varies smoothly, the procedure can aggregate adjacent lags and estimate identical coefficients. When finer resolution is warranted, it retains disaggregated lag-specific effects. Tree-based parameterizations have proven successful for high-dimensional regressions \citep{Yan2021, Fu2025}  and graphical models \citep{Wilms2022}; we adapt and extend these ideas to high-order autoregressions and mixed-frequency time-series regressions.

\begin{figure}[t]
    \centering
\includegraphics[width=0.85\textwidth]{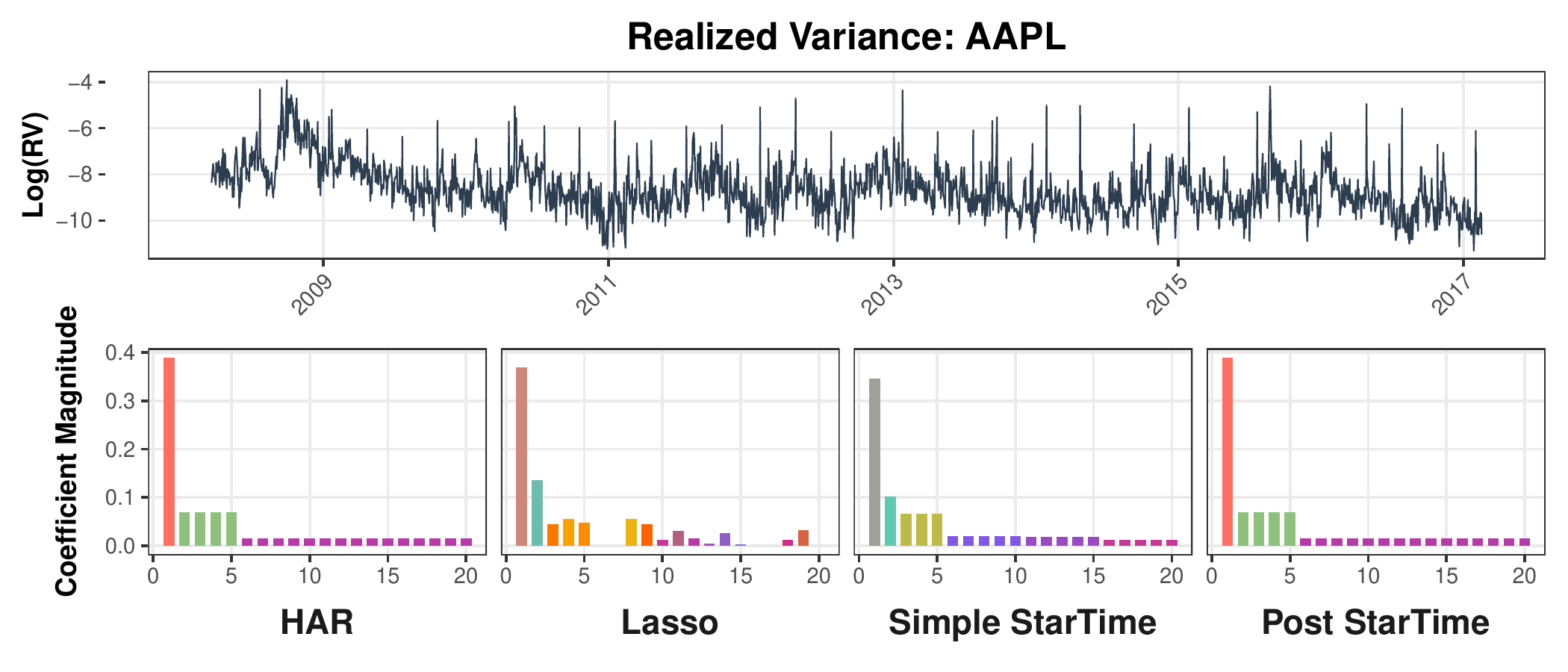}{}
    \caption{Top: Log-transformed realized variance for Apple Inc.\ (AAPL) from March 20, 2008 to February 17, 2017. Bottom: Estimated coefficients for the HAR, Lasso, and StarTime. Identical coefficient values are displayed in the same color.}
    \label{fig:intro_graph}
\end{figure}

Figure \ref{fig:intro_graph} illustrates  StarTime's ability to recover temporally aggregated lag structures in a data-driven manner. We consider an autoregressive model of order 20 for the log-transformed realized variance of Apple Inc.\ (AAPL) and  compare the estimated coefficients returned by the HAR, Lasso, and StarTime. The HAR returns a parsimonious coefficient structure by \textit{a priori} specifying aggregation at the daily, weekly and monthly lag frequency. Lasso selects individual lags in a scattered and difficult-to-interpret manner; its  performance relative to the HAR model has been mixed (e.g., \citealp{zhang2024volatility, christensen2023machine, audrino2020impact, audrino2016Lassoing}). StarTime uncovers an aggregation structure that closely aligns with the economically-motivated HAR specification, particularly when combined with a `post-selection' de-biasing step (Post versus Simple StarTime; see Section \ref{sec:startime}). In contrast to the HAR, StarTime uncovers this temporal aggregation structure fully data-driven, making it broadly applicable across diverse empirical settings.

We derive new error bounds for StarTime to guarantee prediction and estimation consistency  under near-epoch dependence assumptions. We extend the results for independent and identically distributed (IID) data in \cite{ Yan2021} to typical time series settings in  econometric applications characterized by non-Gaussianity, serial dependence, heteroskedasticity, and heavy tails. We hereby leverage the theoretical framework developed in \cite{Adamek2023}. Simulations demonstrate good performance of StarTime in  high-order autoregressive and mixed-frequency settings, particularly when the true lag structure is dense but smooth. We also illustrate StarTime's empirical relevance  in financial and macroeconomic applications. 

The remainder of the paper is organized as follows. Section \ref{sec:methodology} introduces the model and the tree-based aggregation mechanism. Section \ref{sec:startime} presents our penalized estimator and its algorithmic implementation. Section \ref{sec:performance}  establishes the theoretical properties of our estimator. Section \ref{sec:simulations} reports simulation results for autoregressive and mixed-frequency settings, and Section \ref{sec:application}  provides empirical applications. Section \ref{sec:conclusion} concludes.

\section{Temporal Aggregation in Time Series Regressions}
\label{sec:methodology}
We introduce our model set-up for mixed-frequency regressions and then propose a tree-based parametrization structure for temporal aggregation in such regressions.

\subsection{The Mixed-Frequency Regression Model} \label{sec:model}

We consider a general linear regression model for mixed-frequency time series with covariates observed at higher (or equal) frequencies than the response. 
Let $y_{mt}$ ($t=1, \ldots, T$) be the response at time $mt$ and $x^{[i]}_\tau$ the $i$th ($i=1,\ldots, D$) covariate at time $\tau$. For convenience, the first variable is observed at the highest frequency with $m$ recording its frequency mismatch with the response;  high-frequency variable $i=1$  is observed $m$ times per low-frequency
period of the response (e.g., for quarterly/monthly data $m = 3$). Similarly, $m_i \geq 1$  denotes the frequency mismatch of  variable $i$ relative to the first (by construction $m_1=1)$.
The model is given by
\begin{equation}
    y_{mt} = \sum_{i = 1}^{D} \sum_{j = 1}^{P_i} \beta_{j}^{[i]} x_{m(t-1) - m_i(j-1)}^{[i]} + \varepsilon_{mt}, \quad t = 1, \dots, T,    \label{eq:BaseEquation}
\end{equation}
where for each variable $i$, we include $P_i$ lags with parameter $\beta^{[i]}_j$ at lag $j$. The variables are assumed
to be mean-centered such that no intercept is included and $\varepsilon_{mt}$ denotes an error term. The total number of parameters to estimate is  $N := \sum_{i=1}^{D}P_i$. Model \eqref{eq:BaseEquation} includes autoregressive and mixed-frequency regressions, see Examples \ref{ex:AR} and \ref{ex:lowhigh}. 

\begin{example}[AR Models]
The autoregressive model of order $P$, AR($P$), given by
\begin{equation} \label{eq:AR}
y_{t} = \sum_{j = 1}^{P} \beta_{j} y_{t-j} + \varepsilon_{t}, \quad t = 1, \ldots, T,
\end{equation}
is a special case of model \eqref{eq:BaseEquation}, with $D=1, m=m_1=1$, $P_1=P, x^{[1]}_t = y_t$. 
We focus on high-order autoregressive models where adequately capturing the dynamics requires the inclusion of many lags.  The model dimensionality then increases with 
$P$.  
\label{ex:AR}
\end{example}

\begin{example}[Mixed-Frequency Models]
Model \eqref{eq:BaseEquation} accommodates unrestricted MIDAS regressions \citep{foroni2015umidas}. For simplicity, consider the model
\begin{equation} 
    y_{mt} = \sum_{j = 1}^{P_1} \beta^{[1]}_{j}\, x_{m(t-1) - (j - 1)} + \sum_{j = 1}^{P_2} \beta^{[2]}_{j}\, y_{m(t-1) - m(j - 1)} +\varepsilon_{mt}, \quad t = 1, \ldots, T,  \label{eq:MIDAS}
\end{equation}
with  low-frequency response,  high-frequency covariate $x^{[1]}_t = x_t$  with $P_1$ lags and $P_2$ lags of the response ($x^{[2]}_t = y_t$). A typical MIDAS set-up would consider a quarterly-monthly mismatch with $m=3$. While equation \eqref{eq:MIDAS}  permits forecasting, nowcasting can be achieved by including 
$x_{mt - (j - 1)}$
 rather than $x_{m(t-1) - (j - 1)}$.\endnote{Model \eqref{eq:MIDAS} uses different notation than the standard mixed-frequency literature; this is done for convenience; to facilitate the use of temporal aggregation in our modeling set-up.}

Model \eqref{eq:MIDAS} can easily be extended to multiple covariates with their own sampling frequency and lag structure. We focus on modeling mixed-frequency regressions where the curse of dimensionality arises due to two sources. First, the number of variables $D$ may become large relative to the sample size $T$. Secondly, and specific to mixed-frequency models, we allow for the inclusion of high-frequency covariates where the difference in sampling frequencies between the highest- and lowest-frequency ones can be substantial. In Section \ref{sec:macro}, for example, we model GDP growth using daily, weekly, monthly and quarterly variables, the frequency mismatches then range from $m=60$ (assuming 60 trading days in a quarter) over 20 (for day/month) to 5 (for day/week).  Large frequency mismatches further increase the dimensionality of the model, since the higher frequency variables typically require including more lags.
\label{ex:lowhigh}
\end{example}

Examples \ref{ex:AR} and \ref{ex:lowhigh} illustrate that the key drivers of the dimensionality in model \eqref{eq:BaseEquation} are the number of covariates ($D$), the number of lags ($P_i$s), and the frequency mismatches ($m, m_i$s).
We address this curse of dimensionality through penalized estimation. A typical choice would consist of inducing sparsity, for instance via the Lasso \citep{TibshiraniLasso}. In many practical applications with mixed-frequency time series, the lags are often highly correlated, but this is typically not handled well by the Lasso (e.g., Chapter 4 in \citealp{hastie2015sparsity}). Moreover, it is likely that  no single past observation affects the present but rather longer-run patterns, such as sums (or averages) across lags as in the HAR model.
We therefore address the curse of dimensionality in model \eqref{eq:BaseEquation} through temporal lag aggregation (in addition to sparsity). Our penalized estimation routine determines in a data-driven way how to aggregate over past observations in a longer period, thereby reducing the dimensionality by estimating adjacent parameters identically.  We do not impose aggregation structures on the lags a priori. Instead, we employ temporal trees to guide the aggregation in a natural, data-driven way, as discussed next.

\subsection{Temporal Aggregation with Trees}
\label{sec:NotationTreeAggregation}
We first introduce temporal aggregation for autoregressive models, as in equation \eqref{eq:AR}, using a single temporal tree.
We then discuss how  the idea of temporal aggregation can be extended for the general mixed-frequency model \eqref{eq:BaseEquation} with multiple trees.

\subsubsection{A Temporal Tree for AutoRegressive Models} \label{sec:tree-AR}
Consider AR models where the data are observed at high frequency (for instance daily). It is not always suitable to include each lag individually at that finest scale in model \eqref{eq:AR}, as high-frequency measurements can be noisy. Capturing effects at lower temporal resolutions may help reduce noise and amplify the underlying signal by revealing longer-run patterns. We leverage a temporal tree as side information to guide such lag aggregation.

\textit{Temporal Tree.} Time series naturally lend themselves to a tree structure, since finer time units can be grouped into coarser ones: Seconds form minutes, minutes form hours, hours form days, and so on. While our focus is on AR processes observed at relatively high sampling frequency, we illustrate the main idea of temporal aggregation using a smaller-scale toy example to keep the visualizations concise. In particular, consider an AR process where data are observed monthly. Figure \ref{fig:TreeParam} displays a temporal tree where monthly lags (six leaves) are aggregated naturally into two quarterly lags (internal nodes) and then into a single semester lag (root node). Moving up the tree thus corresponds to a natural transition from higher- to lower-frequency observations.

\begin{figure}[t]
   \centering
    \begin{subfigure}[t]{0.42\textwidth}
        \centering
        \begin{tikzpicture}
          [baseline=(root.base),
           scale=1, transform shape,
           level distance=14mm,
           every node/.style={circle,draw,inner sep=1pt},
           level 1/.style={sibling distance=30mm},
           level 2/.style={sibling distance=10mm}]
          \node[fill=cyan!90] (root) {$\gamma_{1,3}$} 
             child {node[fill=cyan!90] {$\gamma_{1,2}$}
               child {node (gamma11) {$\gamma_{1,1}$}}
               child {node (gamma21) {$\gamma_{2,1}$}}
               child {node[fill=cyan!90] (gamma31) {$\gamma_{3,1}$}
               edge from parent
          node[pos=0.35,right=2pt,draw=none] {\textcolor{cyan!90}{\(\mathbf{+}\)}}} edge from parent
          node[pos=0.65,right=2pt,draw=none] {\textcolor{cyan!90}{\(\mathbf{+}\)}}
               }
             child {node {$\gamma_{2,2}$}
               child {node (gamma41) {$\gamma_{4,1}$}}
               child {node (gamma51) {$\gamma_{5,1}$}}
               child {node (gamma61) {$\gamma_{6,1}$}}
               };
\node[draw=none, below=0.4cm of gamma11] (beta1) {$\beta_1$};
\node[draw=none, below=0.4cm of gamma21] (beta2) {$\beta_2$};
\node[draw=none, below=0.15cm of gamma31, rotate = 90, anchor = center, xshift = -0.3ex] (beta3) {\textcolor{cyan!90}{$=$}};
\node[fill=cyan!90, draw=none, below=0.4cm of gamma31] (beta3) {$\beta_3$};
\node[draw=none, below=0.4cm of gamma41] (beta4) {$\beta_4$};
\node[draw=none, below=0.4cm of gamma51] (beta5) {$\beta_5$};
\node[draw=none, below=0.4cm of gamma61] (beta6) {$\beta_6$};

        \end{tikzpicture}
    \end{subfigure}
    \hfill
    \begin{subfigure}[t]{0.12\textwidth}
        \centering
        \begin{tikzpicture}[
            baseline=(root.base),
            scale=1, transform shape,
            level distance=14mm,
            every node/.style={inner sep=1pt, draw=none, font=\small}
        ]
        \node (root) {Semester} 
            child[edge from parent/.style={draw=none}] {node {Quarters}
                child[edge from parent/.style={draw=none}] {node {Months}}
            };
        \end{tikzpicture}
    \end{subfigure}
    \hfill
    \begin{subfigure}[t]{0.42\textwidth}
        \centering
        \begin{tikzpicture}
          [baseline=(root.base),
           scale=1, transform shape,
           level distance=14mm,
           every node/.style={circle,draw,inner sep=1pt},
           level 1/.style={sibling distance=30mm},
           level 2/.style={sibling distance=10mm}]
          \node (root) {$\gamma_{1,3}$} 
             child {node {$\gamma_{1,2}$}
               child {node (gamma11) {$\gamma_{1,1}$}}
               child {node[fill=gray] (gamma21) {$\gamma_{2,1}$}}
               child {node[fill=gray] (gamma31) {$\gamma_{3,1}$}}
               }
             child {node {$\gamma_{2,2}$}
               child {node[fill=gray] (gamma41) {$\gamma_{4,1}$}}
               child {node[fill=gray] (gamma51) {$\gamma_{5,1}$}}
               child {node[fill=gray] (gamma61) {$\gamma_{6,1}$}}
               };
\node[draw=none, below=0.4cm of gamma11] (beta1) {$\beta_1$};
\node[draw=none, below=0.4cm of gamma21] (beta2) {$\beta_2$};
\node[draw=none, below=0.4cm of gamma31] (beta3) {$\beta_3$};
\node[draw=none, below=0.4cm of gamma41] (beta4) {$\beta_4$};
\node[draw=none, below=0.4cm of gamma51] (beta5) {$\beta_5$};
\node[draw=none, below=0.4cm of gamma61] (beta6) {$\beta_6$};

\begin{scope}
    \draw[densely dashed, rounded corners, thin]
      (beta2.south west) rectangle (beta3.north east)
      (beta4.south west) rectangle (beta6.north east);
\end{scope}
        \end{tikzpicture}
    \end{subfigure}
    \caption{Toy example of a temporal tree aggregating months into quarters and a semester. Left panel: Re-parametrization of the original parameters $\beta_j$ in terms of the $\gamma$s. The path from leaf 3 to the root, 
   highlighted in blue, presents the re-parametrization of its corresponding coefficient: 
   $\beta_3 = \gamma_{3,1} + \gamma_{1,2} + \gamma_{1,3}$. 
   Similarly, the other descendants of $\gamma_{1,2}$ are expressed with $\beta_i = \gamma_{i,1} + \gamma_{1,2} + \gamma_{1,3}$ 
   for $i = 1,2,3$, while the descendants of $\gamma_{2,2}$ are $\beta_j = \gamma_{j,1} + \gamma_{2,2} + \gamma_{1,3}$ 
   for $j = 4,5,6$. Right panel: Effect of penalization where the nodes in gray are zeroed out, thereby giving 
   $\beta_1 = \gamma_{1,1} + \gamma_{1,2}+\gamma_{1,3}$, $\beta_i = \gamma_{1,2}+\gamma_{1,3}$ for $i=2,3$ and  $\beta_j = \gamma_{2,2} + \gamma_{1,3}$ for $j=4,5,6$.} 
   \label{fig:TreeParam}
\end{figure}
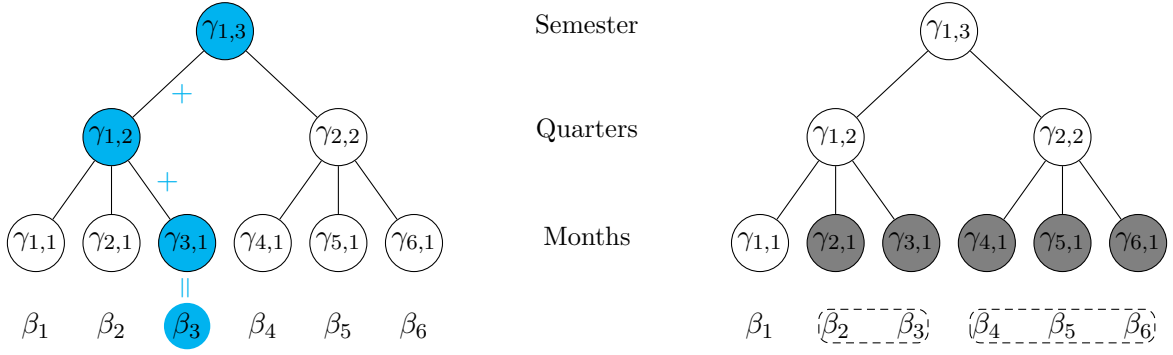

While the  tree provides an intuitive way to perform temporal aggregation, in practice, the most suitable temporal aggregation level for the lags in model \eqref{eq:AR}  is often unknown a priori. In our toy example, lags may be incorporated at monthly, quarterly, semiannual frequencies, or even a  combination thereof. Our objective is  to let the data  determine the appropriate temporal resolution at which lags should enter the model. To this end, we employ a reparametrization strategy on the parameters $\beta_j$ of model \eqref{eq:AR}, extending the approach of \cite{Yan2021} to our time series setting. The temporal tree is central to help select, in a data-driven way, an appropriate temporal frequency at which lags should be included  in model \eqref{eq:AR}.

\textit{Tree-Based Model Parameterization.} Let $\mathcal{T}$ denote
the tree  that encodes the temporal aggregation structure across $L$ different frequency levels. At each level $l$, there are $p_l$ nodes, with $p_1$ 
denoting the number of leaves at the bottom level and $p_L = 1$ (the single root node). At each level $l \in \{1, \dots, L-1\}$, the $p_l$ nodes are grouped into $p_{l+1}$ parent nodes. Each parent node aggregates a fixed number $K_l$ of consecutive child nodes, so that $p_{l+1} = p_l/K_l$. This defines the branching structure of the tree, with each group at level $l$ aggregating $K_l$ nodes from the previous level.

Consider the toy example in Figure \ref{fig:TreeParam}. The tree has $L=3$ levels.  There are  $p_1 = 6$ leaf nodes, corresponding to the monthly lags. Those are grouped in triplets ($K_1=3$) leading to the $p_2 = 2$ quarterly internal nodes on level 2 which are, in turn, grouped into a single root node  ($K_2 = 2$) representing a  full semester of monthly observations. 

The tree  can now be used to naturally guide temporal aggregation of the lags in the AR model. We hereby aim to encourage aggregation of lags within a branch to form new, lower-frequency lags.
To implement such tree-guided aggregation, we associate a coefficient $\beta_j$ with every leaf of the tree ($p_1=P$ by construction) and encourage fusion of lags within a branch so that the descendant lags of that branch are aggregated when their effects are similar.

To enable fusion, each node in the tree (leaf, internal, root) gets assigned a new auxiliary parameter $\gamma_{k,l}$, where $k = 1, \dots, p_l$ indexes the nodes on level $l=1,\ldots,L$. We represent each original $\beta$ parameter as the sum of the $\gamma$ parameters along the path from its leaf to the root. The relation between the $\beta_j$  and $\gamma$s is expressed as 
\begin{equation*}\beta_{j} = \sum_{l=1}^L \gamma_{\theta_{\mathcal{T}}(j,l)}, 
\end{equation*}
where the mapping function  $\theta: (\mathcal{T},j,l) \rightarrow \theta_{\mathcal{T}}(j,l):= \big( k, l \big)$ uses the tree $\mathcal{T}$, the position of leaf $j$, and the level $l$ to uniquely identify 
each node in the tree. The left panel of Figure \ref{fig:TreeParam} highlights such a path in blue, namely for parameter $\beta_3$ corresponding to leaf $j=3$. We then have $\beta_3 = \gamma_{3,1} + \gamma_{1,2} + \gamma_{1,3}$, using that at level 1, $\theta_{\mathcal{T}}(3,1) = (3,1)$, at level 2, $\theta_{\mathcal{T}}(3,2) = (1,2)$, and at level 3, $\theta_{\mathcal{T}}(3,3) = (1,3)$.

Finally, stacking all $N$ parameters in $\boldsymbol\beta$ (for the AR model $N=P)$, we can write $\boldsymbol{\beta} = \boldsymbol{A} \boldsymbol{\gamma}$, where $\boldsymbol{A} \in \{0,1\}^{N \times |\mathcal{T}|}$ encodes the tree structure and $\boldsymbol{\gamma} \in \mathbb{R}^{|\mathcal{T}|\times 1}$ stacks the parameters from each level of the tree sequentially, starting with all level-1 (leaf) nodes ordered from left to right, followed by the level-2 nodes in the same order, and proceeding upward until the root node at level $L$. The entries of $\boldsymbol{A}$ are given by 
\begin{equation*} 
    a_{j,p} = 
    \begin{cases}
        1 & \text{if the node at position $p$ lies on the path from $\gamma_{j,1}$ to $\gamma_{1,L}$}, \\
        0 & \text{otherwise.}
    \end{cases}
\end{equation*}

\textit{Temporal Aggregation Through Penalization.}
We reparametrize the $\beta_j$s  in terms of the $\gamma$s since aggregation can now be achieved by zeroing out $\gamma$s. Each branch in the tree reflects the aggregation of its descendants.
Zeroing out $\gamma$s for all descendants of a branch fuses the coefficients in this branch, leading to  temporal aggregation. 

In the toy example, zeroing out all monthly leaves results in the inclusion of its corresponding quarterly effect. In Figure \ref{fig:TreeParam} (right panel), we illustrate the effect of zeroing out $\gamma$s (their nodes are grayed out). For example, months 4 to 6 are no longer included individually in the model. Instead, their coefficients are equal ($\beta_4 = \beta_5 = \beta_6 = \gamma_{2,2} + \gamma_{1,3}$), so that their aggregate quarterly effect is captured in the model.

By allowing each $\gamma$ in the tree to be set to zero, we allow the model to flexibly determine the appropriate level of aggregation for each part of the lag space. 
In the same figure, the first three monthly lags display distinct behavior either remaining separate (not fused) in the model like the first monthly lag $\beta_1$ or entering partially fused like the second and third monthly lags ($\beta_2 = \beta_3 = \gamma_{1,2}+\gamma_{1,3}$). As a result, we reduce the dimensionality of the AR model \eqref{eq:AR} by fusing certain lags, hence performing temporal aggregation, while still allowing for meaningful differences in the temporal dynamics when economically relevant and supported by the data. 

\begin{remark}\label{rem:non-unique-gamma}
Multiple configurations of the zeroed out $\boldsymbol\gamma$s may induce the same aggregation in the target parameter $\boldsymbol{\beta}$. In the toy example, zeroing out $\gamma_{2,2}$ in addition to the other zeroed out $\gamma$s leads to the same fusion as before (then $\beta_j = \gamma_{1,3}$ for $j=4,5,6$). The  non-uniqueness in terms of the $\gamma$s is not problematic, since the $\beta$s are the parameters of  interest. From a computational perspective, this non-uniqueness does not hinder estimation: The optimization problem, introduced in Section \ref{sec:startime}, is convex in $\boldsymbol\beta$ and $\boldsymbol\gamma$ and both can be estimated in a computationally efficient way. 
\end{remark}

\begin{example}[HAR Model]
The HAR model by \cite{Corsi2009} can be recovered through our temporal aggregation framework. Consider AR model \eqref{eq:AR} with $P=20$ daily lags. The temporal tree in Figure \ref{fig:HARTree} encourages aggregation of the daily lags (20 leaves) into four weekly lags (internal nodes) and one single monthly lag (root). The sparsity structure on the $\gamma$s in Figure \ref{fig:HARTree} induces the fusion constraints  $\beta_1 = \gamma_{1,1} + \gamma_{1,2} + \gamma_{1,3}, \beta_{i} =\gamma_{1,2} + \gamma_{1,3}$ for $i=2, \dots, 5$ and $\beta_{j} = \gamma_{1,3}$ for $j=6, \dots, 20$. Incorporating these constraints into the AR(20) leads to the HAR model, namely 
\begin{align*}
    y_{t} &= \sum_{j=1}^{20} \beta_j y_{t-j} + \varepsilon_{t} \\
        &= \left (\gamma_{1,1} + \gamma_{1,2} + \gamma_{1,3} \right ) y_{t-1} + \left ( \gamma_{1,2} + \gamma_{1,3} \right ) \sum_{j=2}^{5} y_{t-j} + \gamma_{1,3}\sum_{j=6}^{20} y_{t-j} + \varepsilon_{t} \\
        &= \gamma_{1,1} y_{t-1} + \gamma_{1,2} \sum_{j=1}^{5} y_{t-j} + \gamma_{1,3} \sum_{j=1}^{20} y_{t-j} + \varepsilon_{t}, 
\end{align*}
where line 2 follows from substituting the sparsity constraints on $\gamma$ depicted in Figure \ref{fig:HARTree} and line 3 rearranges terms to express the HAR model according to its daily ($\beta^{(d)} = \gamma_{1,1}$), weekly ($\beta^{(w)} = 5\gamma_{1,2}$), and monthly ($\beta^{(m)} = 20\gamma_{1,3}$) effects. 
\end{example}

Unlike the HAR model, we do not impose the aggregation structure a priori. Instead, the temporal tree serves as guidance for potential temporal aggregations. The latter are determined in a data-driven manner through tree-based penalized estimation (Section \ref{sec:startime}). In particular, we add an $\ell_1$ penalty on $\boldsymbol\gamma$ to the objective to enable selection of a suitable temporal aggregation structure for the problem at hand.

 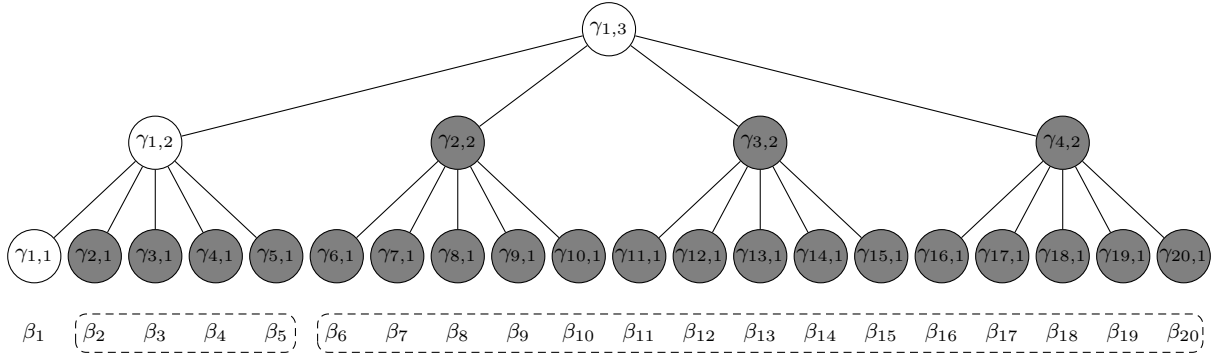
\begin{figure}[t]
    \centering
        \begin{tikzpicture}
          [scale=1, transform shape,
          level distance=15mm,
            every node/.style={
            circle,
            draw,
            inner sep=0pt,
            minimum size=20pt,
            align=center,
            font=\scriptsize},
           level 1/.style={sibling distance=40mm},
           level 2/.style={sibling distance=8mm}]
          \node {$\gamma_{1,3}$}
             child {node {$\gamma_{1,2}$}
               child {node (gamma11) {$\gamma_{1,1}$}}
               child {node[fill=gray] (gamma21) {$\gamma_{2,1}$}}
               child {node[fill=gray] (gamma31) {$\gamma_{3,1}$}}
               child {node[fill=gray] (gamma41) {$\gamma_{4,1}$}}
               child {node[fill=gray] (gamma51) {$\gamma_{5,1}$}}
             }
              child {node[fill=gray] {$\gamma_{2,2}$}
               child {node[fill=gray] (gamma61) {$\gamma_{6,1}$}}
               child {node[fill=gray] (gamma71) {$\gamma_{7,1}$}}
               child {node[fill=gray] (gamma81){$\gamma_{8,1}$}}
               child {node[fill=gray] (gamma91){$\gamma_{9,1}$}}
               child {node[fill=gray] (gamma101){$\gamma_{10,1}$}}
             }
             child {node[fill=gray] {$\gamma_{3,2}$}
               child {node[fill=gray] (gamma111){$\gamma_{11,1}$}}
               child {node[fill=gray] (gamma121){$\gamma_{12,1}$}}
               child {node[fill=gray] (gamma131){$\gamma_{13,1}$}}
               child {node[fill=gray] (gamma141){$\gamma_{14,1}$}}
               child {node[fill=gray] (gamma151){$\gamma_{15,1}$}}
             }
              child {node[fill=gray] {$\gamma_{4,2}$}
               child {node[fill=gray] (gamma161){$\gamma_{16,1}$}}
               child {node[fill=gray] (gamma171){$\gamma_{17,1}$}}
               child {node[fill=gray] (gamma181){$\gamma_{18,1}$}}
               child {node[fill=gray] (gamma191){$\gamma_{19,1}$}}
               child {node[fill=gray] (gamma201){$\gamma_{20,1}$}}
             };

\node[draw=none, below=0.3cm of gamma11] (beta1) {$\beta_1 $};
\node[draw=none, below=0.3cm of gamma21] (beta2) {$\beta_2 $};
\node[draw=none, below=0.3cm of gamma31] (beta3) {$\beta_3 $};
\node[draw=none, below=0.3cm of gamma41] (beta4) {$\beta_4 $};
\node[draw=none, below=0.3cm of gamma51] (beta5) {$\beta_5 $};
\node[draw=none, below=0.3cm of gamma61] (beta6) {$\beta_6 $};
\node[draw=none, below=0.3cm of gamma71] (beta7) {$\beta_7 $};
\node[draw=none, below=0.3cm of gamma81] (beta8) {$\beta_8 $};
\node[draw=none, below=0.3cm of gamma91] (beta9) {$\beta_9 $};
\node[draw=none, below=0.3cm of gamma101] (beta10) {$\beta_{10} $};
\node[draw=none, below=0.3cm of gamma111] (beta11) {$\beta_{11} $};
\node[draw=none, below=0.3cm of gamma121] (beta12) {$\beta_{12} $};
\node[draw=none, below=0.3cm of gamma131] (beta13) {$\beta_{13} $};
\node[draw=none, below=0.3cm of gamma141] (beta14) {$\beta_{14} $};
\node[draw=none, below=0.3cm of gamma151] (beta15) {$\beta_{15} $};
\node[draw=none, below=0.3cm of gamma161] (beta16) {$\beta_{16} $};
\node[draw=none, below=0.3cm of gamma171] (beta17) {$\beta_{17} $};
\node[draw=none, below=0.3cm of gamma181] (beta18) {$\beta_{18} $};
\node[draw=none, below=0.3cm of gamma191] (beta19) {$\beta_{19} $};
\node[draw=none, below=0.3cm of gamma201] (beta20) {$\beta_{20} $};

\begin{scope}
    \draw[densely dashed, rounded corners, thin]
      (beta2.south west) rectangle (beta5.north east)
      (beta6.south west) rectangle (beta20.north east);
\end{scope}
        \end{tikzpicture}
        \caption{Temporal tree aggregating 20 days into four weeks and one month. The sparsity pattern on the $\gamma$s (gray nodes are zeroed out) leads to recovery of the HAR model.}
        \label{fig:HARTree}
\end{figure}

\subsubsection{Temporal Trees for Mixed-Frequency Models} \label{sec:tree-general}
Temporal aggregation introduced in Section \ref{sec:tree-AR} for AR models can easily be extended to the general mixed-frequency model \eqref{eq:BaseEquation}. 
We assign a separate temporal tree $\mathcal{T}^{[i]}$ for each covariate $i$ to govern the temporal aggregation of its  lags. 
The leaves represent the highest recorded frequency for that variable, the root represents the lowest temporal frequency of interest. Each variable $i$ has its auxiliary parameter vector $\boldsymbol{\gamma}^{[i]}$ and matrix $\boldsymbol{A}^{[i]}$ that relates the parameters $\beta^{[i]}_j$ to the $\boldsymbol\gamma^{[i]}$.  The vector $\boldsymbol{\gamma}$ stacks the individual vectors $\boldsymbol{\gamma}^{[i]}$ vertically, ordered from variable $1$ to $D$. Similarly, the matrix $\boldsymbol{A}$ is constructed by placing the individual matrices $\boldsymbol{A}^{[i]}$ on its diagonal, so that 
$\boldsymbol{A} \in \{0,1\}^{N \times M}$, with $M = \sum_{i=1}^{D} |\mathcal{T}^{[i]}|$.
Aggregation and fusion for variable $i$ are performed  within its tree,  independently of all other variables. Temporal aggregation of one variable does not constrain that of another, allowing full flexibility.

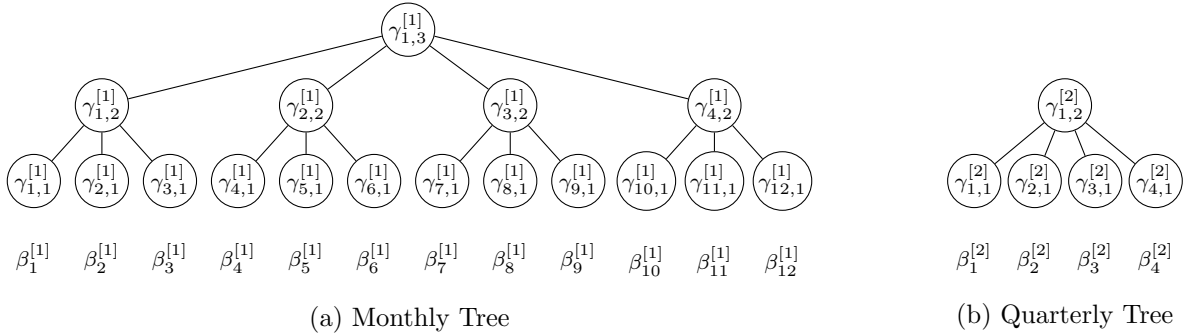
\begin{figure}[t]
    \centering
    \begin{subfigure}[t]{0.65\textwidth}
        \centering
        \begin{tikzpicture}
          [scale=1, transform shape,
          baseline=(beta1.base),
          level distance=10mm,
            every node/.style={
            circle,
            draw,
            inner sep=0pt,
            minimum size=20pt,
            align=center,
            font=\scriptsize},
           level 1/.style={sibling distance=27mm},
           level 2/.style={sibling distance=9mm}]
          \node (gamma13) {$\gamma_{1,3}^{[1]}$}
               child {node (gamma12) {$\gamma_{1,2}^{[1]}$}
                 child {node (gamma11) {$\gamma_{1,1}^{[1]}$}}
                 child {node (gamma21) {$\gamma_{2,1}^{[1]}$}}
                 child {node (gamma31) {$\gamma_{3,1}^{[1]}$}}
                }
               child {node (gamma22) {$\gamma_{2,2}^{[1]}$}
                 child {node (gamma41) {$\gamma_{4,1}^{[1]}$}}
                 child {node (gamma51) {$\gamma_{5,1}^{[1]}$}}
                 child {node (gamma61) {$\gamma_{6,1}^{[1]}$}}
               }
               child {node (gamma32) {$\gamma_{3,2}^{[1]}$}
                 child {node  (gamma71) {$\gamma_{7,1}^{[1]}$}}
                 child {node  (gamma81) {$\gamma_{8,1}^{[1]}$}}
                 child {node  (gamma91) {$\gamma_{9,1}^{[1]}$}}
                }
               child {node (gamma42) {$\gamma_{4,2}^{[1]}$}
                 child {node (gamma101) {$\gamma_{10,1}^{[1]}$}}
                 child {node (gamma111) {$\gamma_{11,1}^{[1]}$}}
                 child {node (gamma121) {$\gamma_{12,1}^{[1]}$}}
               };
               \node[draw=none, below=0.3cm of gamma11] (beta1) {$\beta^{[1]}_1$};
\node[draw=none, below=0.3cm of gamma21] (beta2) {$\beta^{[1]}_2$};
\node[draw=none, below=0.3cm of gamma31] (beta3) {$\beta^{[1]}_3$};
\node[draw=none, below=0.3cm of gamma41] (beta4) {$\beta^{[1]}_4$};
\node[draw=none, below=0.3cm of gamma51] (beta5) {$\beta^{[1]}_5$};
\node[draw=none, below=0.3cm of gamma61] (beta6) {$\beta^{[1]}_6$};
\node[draw=none, below=0.3cm of gamma71] (beta7) {$\beta^{[1]}_7$};
\node[draw=none, below=0.3cm of gamma81] (beta8) {$\beta^{[1]}_8$};
\node[draw=none, below=0.3cm of gamma91] (beta9) {$\beta^{[1]}_9$};
\node[draw=none, below=0.3cm of gamma101] (beta10) {$\beta^{[1]}_{10}$};
\node[draw=none, below=0.3cm of gamma111] (beta11) {$\beta^{[1]}_{11}$};
\node[draw=none, below=0.3cm of gamma121] (beta12) {$\beta^{[1]}_{12}$};

        \end{tikzpicture}
        \caption{Monthly Tree}
        \label{fig:multiexample_tree1}
    \end{subfigure}
    \hspace*{\fill}
    \begin{subfigure}[t]{0.3\textwidth}
        \centering
        \begin{tikzpicture}
          [scale=1, transform shape,
          baseline=(beta21.base),
          level distance=10mm,
            every node/.style={
            circle,
            draw,
            inner sep=0pt,
            minimum size=20pt,
            align=center,
            font=\scriptsize},
           level 1/.style={sibling distance=8mm}]
          \node (gamma13) {$\gamma_{1,2}^{[2]}$}
               child {node (gamma11) {$\gamma_{1,1}^{[2]}$}
                }
               child {node(gamma21) {$\gamma_{2,1}^{[2]}$}
               }
               child {node (gamma31) {$\gamma_{3,1}^{[2]}$}
                }
               child {node (gamma41) {$\gamma_{4,1}^{[2]}$}
               };
\node[draw=none, below=0.3cm of gamma11] (beta21) {$\beta^{[2]}_1$};
\node[draw=none, below=0.3cm of gamma21] (beta22) {$\beta^{[2]}_2$};
\node[draw=none, below=0.3cm of gamma31] (beta23) {$\beta^{[2]}_3$};
\node[draw=none, below=0.3cm of gamma41] (beta24) {$\beta^{[2]}_4$};
        \end{tikzpicture}
        \caption{Quarterly Tree}
        \label{fig:multiexample_tree2}
    \end{subfigure}
    \caption{Trees for a mixed-frequency regression with a quarterly response and monthly covariate.}
  \label{fig:MultiTree}  
\end{figure}

\begin{example}[Mixed-Frequency Regressions with Multiple Trees]
\label{ex:MultiTrees}
Consider model \eqref{eq:MIDAS} with a quarterly response and monthly covariate (hence $m=3$). When including $P_1 = 12$ monthly lags and $P_2=4$ quarterly lags, 
the model takes the form
\begin{equation*} 
    y_{3t} = \sum_{j=1}^{12} \beta_{j}^{[1]} x^{[1]}_{3(t-1) - (j - 1)} + \sum_{j=1}^{4} \beta_{j}^{[2]} y_{3(t-1) - 3j} + \varepsilon_{3t}, \quad t = 1, \ldots, T.
\end{equation*} 
Figure \ref{fig:MultiTree} displays the  temporal trees for the monthly variable and the quarterly variable, each aggregating to a yearly frequency at the root. 
The relationship between the $N=16$-dimensional $\boldsymbol{\beta}$-vector and the $M=22$-dimensional $\boldsymbol{\gamma}$-vector is given by
\begin{equation*}
\boldsymbol{\beta} = 
    \begin{bmatrix}
        \boldsymbol{\beta}^{[1]} \\
        \boldsymbol{\beta}^{[2]}
    \end{bmatrix}
    = \begin{bmatrix}
        \boldsymbol{A}^{[1]} & \boldsymbol{O}_{12 
        \times 5} \\
        \boldsymbol{O}_{4
        \times 17} & \boldsymbol{A}^{[2]}
    \end{bmatrix}
    \begin{bmatrix}
        \boldsymbol{\gamma}^{[1]} \\
        \boldsymbol{\gamma}^{[2]}
    \end{bmatrix} 
    = \boldsymbol{A} \boldsymbol{\gamma},
\end{equation*}
where $\boldsymbol{\beta}^{[1]} \in \mathbb{R}^{12\times 1}$ 
collects the monthly  and $\boldsymbol{\beta}^{[2]} \in \mathbb{R}^{4\times 1}$ the quarterly parameters, similarly for $\boldsymbol{\gamma}^{[1]} \in \mathbb{R}^{17\times 1}$ 
and $\boldsymbol{\gamma}^{[2]} \in \mathbb{R}^{5\times 1}$, 
$\boldsymbol{A}^{[1]} \in \{0,1\}^{12 \times 17}$, $\boldsymbol{A}^{[2]} \in \{0,1\}^{4 \times 5}$, $\boldsymbol{A} \in \{0,1\}^{16 \times 22}$ and $\boldsymbol{O}$ is a conformable matrix of zeros.
\end{example}

\section{Sparse Tree-Based Aggregation through Penalized Estimation}\label{sec:startime}
We first introduce our penalized estimation procedure. Details on our algorithmic implementation and the hyperparameter tuning are discussed subsequently. 

\subsection{Estimator} \label{sec:estimators}
Consider model \eqref{eq:BaseEquation} in compact matrix form $\boldsymbol{y} = \boldsymbol{X}\boldsymbol{\beta} + \boldsymbol{\varepsilon},$ where $\boldsymbol{y} \in \mathbb{R}^{T\times 1}$ is the response vector, $\boldsymbol{X} \in \mathbb{R}^{T \times N}$ is the matrix of covariates, containing the lags of $x_t^{[i]}$ in its columns for $i = 1, \dots, D$, $\boldsymbol{\beta} \in \mathbb{R}^{N \times 1}$ is the parameter vector, and $\boldsymbol{\varepsilon} \in \mathbb{R}^{T\times 1}$ contains the errors. 
To estimate the model parameters, we propose the \textit{StarTime} estimator, which is defined as the solution to the following objective
\begin{equation}
\boldsymbol{\Tilde{\beta}}_{\lambda_1,\lambda_2},\boldsymbol{\Tilde{\gamma}}_{\lambda_1, \lambda_2} = \argmin_{\boldsymbol{\beta}, \boldsymbol{\gamma}} \bigg\{ \frac{1}{2T} \| \boldsymbol{y} - \boldsymbol{X} \boldsymbol{\beta} \|^2_2 
+ \lambda_1  \|\boldsymbol{\gamma}\|_1 
+ \lambda_2 \|\boldsymbol{\beta}\|_1 \; \; \text{s.t.} \; \boldsymbol{\beta} = \boldsymbol{A \gamma}  
\bigg\},
\label{eq:objectivefunction}
\end{equation}
where the first part in the objective is the usual sum of squared residuals, the second part consists of the penalty terms on $\boldsymbol\gamma$ and $\boldsymbol\beta$, and the constraint $\boldsymbol{\beta} = \boldsymbol{A \gamma}$ imposes the tree correspondence between the $\gamma$'s and $\beta$'s, as explained in Section \ref{sec:NotationTreeAggregation}. The penalty $\lambda_1 \|\boldsymbol{\gamma}\|_1$ induces sparsity on the tree nodes, encouraging temporal aggregation of the lags where appropriate.
The  $\lambda_2 \|\boldsymbol{\beta}\|_1$ encourages direct sparsity on the lagged coefficients. The tuning parameter $\lambda_1$ controls the level of aggregation: Larger values of $\lambda_1$ result in a higher degree of aggregation, whereas $\lambda_2$ controls the degree of sparsity. StarTime can thus perform joint estimation, temporal aggregation and lag selection.
When $\lambda_1 = 0$, StarTime reduces to the classical Lasso. 
We refer to the StarTime estimator in equation \eqref{eq:objectivefunction} as `Simple StarTime' from here on.

We also consider a `post-selection' 
variant referred to as `Post StarTime'. Given the sparsity and aggregation pattern in $\boldsymbol{\Tilde{\beta}}$, we refit the model using OLS on the reduced set of selected, potentially temporally aggregated variables. 
This two-stage procedure can improve estimation accuracy by removing the bias induced by the penalties. 
Formally, let $\boldsymbol{\Tilde{\beta}} \in \mathbb{R}^{N\times 1}$ denote the vector encoding the selected and aggregated variables. 
Denote the vector collecting the unique, non-zero entries of $\boldsymbol{\Tilde{\beta}}$  by $\boldsymbol{\Tilde{\phi}} \in \mathbb{R}^{Q\times 1}$, where $Q$ denotes the number of unique, non-zero values in $\boldsymbol{\Tilde{\beta}}$. The aggregation matrix $\Tilde{\boldsymbol{W}} \in \mathbb{R}^{N \times Q}$ is given by
\begin{equation*}
    \Tilde{w}_{j,\ell} = 
\begin{cases}
1 & \text{if } \Tilde{\beta}_j = \Tilde{\phi}_\ell \\
0 & \text{otherwise}
\end{cases}
\qquad\text{for } j = 1, \ldots, N, \ell = 1, \ldots, Q.
\end{equation*}
We can then write $\boldsymbol{\Tilde{\beta}} = \boldsymbol{\Tilde{W}} \boldsymbol{\Tilde{\phi}}$ and  
$\boldsymbol{X} \boldsymbol{\Tilde{\beta}} = \boldsymbol{X} \boldsymbol{\Tilde{W}} \boldsymbol{\Tilde{\phi}}$ with aggregated matrix  $\boldsymbol{\Tilde{X}} = \boldsymbol{X} \boldsymbol{\Tilde{W}},$ whose $\ell$-th column is the sum of the columns of $\boldsymbol{X}$ corresponding to indices $j$ for which $\Tilde{\beta}_j = \Tilde{\phi}_\ell$, and the columns with $\Tilde{\beta}_j = 0$ are omitted.
We then perform OLS estimation with $\boldsymbol{y}$ as response and $\boldsymbol{\Tilde{X}}$ as design matrix to obtain $ \boldsymbol{\hat{\phi}} = (\boldsymbol{\Tilde{X}}^\top\boldsymbol{\Tilde{X}})^{-1}\boldsymbol{\Tilde{X}}^\top\boldsymbol{y}.$ Finally, we map $\boldsymbol{\hat{\phi}}$ onto $\boldsymbol{\hat{\beta}} \in \mathbb{R}^{N\times 1}$ 
to obtain the Post StarTime estimator $\boldsymbol{\hat{\beta}} =  \boldsymbol{\Tilde{W}} \boldsymbol{\hat{\phi}}.$

\subsection{Algorithm}\label{sec:algorithm}

To solve \eqref{eq:objectivefunction}, we develop an alternating direction method of multipliers (ADMM) algorithm \citep{boyd2011admm} based on solving the
equivalent formulation of \eqref{eq:objectivefunction}:
\begin{equation}
\begin{aligned}
&\min_{\substack{
    \boldsymbol\beta^{(1)},\, \boldsymbol\beta^{(2)},\, \boldsymbol\beta^{(3)},\, \boldsymbol\beta \\
    \boldsymbol\gamma^{(1)},\, \boldsymbol\gamma^{(2)},\, \boldsymbol\gamma 
}}
\left\{
\frac{1}{2T} \| \boldsymbol{y} - \boldsymbol{X}\boldsymbol{\beta}^{(1)} \|_2^2
+ \lambda_1 \| \boldsymbol{\gamma}^{(1)} \|_1
+ \lambda_2 \| \boldsymbol{\beta}^{(2)} \|_1
\right. \\[2ex]
&\hspace{7.5em} \left.
\text{s.t. }~
\boldsymbol{\beta}^{(3)} = \boldsymbol{A}\boldsymbol{\gamma}^{(2)},~
\boldsymbol{\beta} = \boldsymbol{\beta}^{(1)} = \boldsymbol{\beta}^{(2)} = \boldsymbol{\beta}^{(3)},
\boldsymbol{\gamma} = \boldsymbol{\gamma}^{(1)} = \boldsymbol{\gamma}^{(2)}
\right\},
\end{aligned} \nonumber
\end{equation}
where additional copies of $\boldsymbol\beta$ ($\boldsymbol\beta^{(1)}, \boldsymbol\beta^{(2)}, \boldsymbol\beta^{(3)}$) and $\boldsymbol\gamma$ ($\boldsymbol\gamma^{(1)}, \boldsymbol\gamma^{(2)}$)  efficiently decouple the optimization problem, whereas the global  $\boldsymbol\beta$ and $\boldsymbol\gamma$ ensure the different subproblems are in agreement. The subproblems of our ADMM routine are identical to the subproblems of \cite{Yan2021}; their Appendix D.1 contains full details.

Regarding convergence of the ADMM algorithm, 
we evaluate the maximum absolute change between the estimated $\boldsymbol{\beta}$ at iterations $i$ and $i+1$ and terminate the ADMM procedure  when this difference falls below a threshold of $10^{-5}$. We fix the maximum number of iterations at 1000 and set the ADMM parameter to $\rho = 1$. 

\subsection{Tuning Parameters} \label{sec:tuning}
We select the tuning parameters $\lambda_1$ and $\lambda_2$ by using grid search combined with a Bayesian Information Criterion (BIC). We construct a $10 \times 10$ grid of candidate $(\lambda_1, \lambda_2)$ values. For each pair on the grid, we estimate the model and compute the following information criterion 
\begin{equation}
    \text{BIC}_{\lambda_1, \lambda_2} =  
    \tilde{T}\ln\left( \frac{\hat{\sigma}^2_{\lambda_1, \lambda_2}}{\tilde{T}} \right) 
     + 
     Q_{\lambda_1, \lambda_2} \ln(\tilde{T}) 
    + \infty \mathds{1}\left\{ \frac{Q_{\lambda_1, \lambda_2}}{\tilde{T}} > c \right\},
    \label{eq:BIC}
\end{equation}
where $\hat{\sigma}^2_{\lambda_1,\lambda_2} = || \boldsymbol{y} - \boldsymbol{X} \hat{\boldsymbol{\beta}}||_2^2$, with $\hat{\boldsymbol{\beta}}$  the Post StarTime estimate, $\tilde{T} = T - P_{\max}$ the effective sample size with $P_{\max}$ the number of lags we maximally lose, and $Q_{\lambda_1, \lambda_2}$  the number of unique, non-zero components of $\hat{\boldsymbol{\beta}}$ for a given $(\lambda_1, \lambda_2)$ pair. The extra penalty on $Q_{\lambda_1, \lambda_2}/\tilde{T}$ helps prevent the model from selecting pairs that would result in the inclusion of too many parameters, which could lead to overfitting.\endnote{The intuition comes from the EBIC of \cite{ChenChen2012} which also contains an extra penalty for model complexity in case $N \geq \tilde{T}$.} If the ratio exceeds the threshold $c \in [0,1]$, we add a large penalty to the BIC, which prevents the selection of that pair. 
For $c=1$, the usual BIC applies and no extra penalty is used. 
The threshold $c$ can be adjusted based on user preference and the complexity of the problem.  
After evaluating the model for all $100$ combinations of $(\lambda_1, \lambda_2)$ on the grid, we select the pair that yields the lowest BIC value. 

Finally, to construct the $\lambda$-grid, we determine maximum values such that respectively all $\hat{\boldsymbol\gamma}$ and $\hat{\boldsymbol\beta}$ coefficients are zero. We take $\lambda_1^{\max} = \max \big| (\boldsymbol{X} \boldsymbol{A})^\top \boldsymbol{y} \big| / \tilde{T}$ and $\lambda_2^{\max} = \max \big| \boldsymbol{X}^\top \boldsymbol{y} \big|/\tilde{T}$. The minimum value is set as a fixed fraction of this maximum: $\text{ If } \tilde{T} > N, \text{ then we set } \lambda_i^{\min} = \lambda_i^{\max} \times 10^{-4}$ and $\text{if } N \geq \tilde{T}, \text{ then } \lambda_i^{\min} = \lambda_i^{\max} \times 10^{-7}$ (for $i =1,2$) to cover a broader range in case of high dimensionality. The sequence of $\lambda$ values is spaced evenly on a logarithmic scale between these endpoints. 

\section{Theoretical Properties}
\label{sec:performance}

We derive an error bound for  StarTime  (Theorem \ref{theorem:oracle}) which is then used to establish prediction and estimation consistency (Corollary \ref{corollary:separatebounds}), allowing for the number of parameters to grow faster than the sample size. We adopt the general framework of Chapter 6 in \cite{buhlmann2011statistics}. 
We first introduce the assumptions underlying our analysis, and then present the main results. All auxiliary lemmas and proofs are collected in Appendix A.1 and A.2 respectively.

We use the following notation. 
For any vector $\boldsymbol{u} \in \mathbb{R}^N$ and $r \geq 1$, the $\ell_r$-norm is defined as $\left \lVert \boldsymbol{u} \right \rVert_r = \left ( \sum_{i=1}^N |u_i|^r \right )^{1/r}$. For $r=0$, we use the convention $\left \lVert \boldsymbol{u}\right \rVert_0 = \sum_{i=1}^N \mathbb{I}(u_i \neq 0)$, with $\mathbb{I}(\cdot)$ the indicator function. The $\ell_\infty$-norm is defined as $\left \lVert \boldsymbol{u}\right \rVert_{\infty} = \max_i |u_i|$ for vectors, and element-wise as $\left \lVert \boldsymbol{A}\right \rVert_{\infty} = \max_{i,j} |A_{i,j}|$ for matrices. In this section and accompanying proofs, $C$ denotes a generic positive finite constant. Its value may change from line to line, absorbing other constants to streamline exposition, but it is always independent of time and the cross-sectional dimension.

We consider the linear model in matrix notation
    $\boldsymbol{y} = \boldsymbol{X} \boldsymbol{\beta}^0 + \boldsymbol{\varepsilon},$
where $\boldsymbol{y}\in\mathbb{R}^{T\times 1}$, $\boldsymbol{X}\in\mathbb{R}^{T\times N}$,   $\boldsymbol{\varepsilon}\in\mathbb{R}^{T \times 1}$ and $\boldsymbol{\beta}^0\in\mathbb{R}^{N\times 1}$ denotes the true parameter value.
Since our analysis concerns dependent processes, we adopt a general time series framework where the covariates $\boldsymbol{x}_t$ (denoting the $t$-th row of $\boldsymbol{X}$) and the errors $\varepsilon_t$ are permitted to be non-Gaussian, serially correlated and heteroskedastic. To accommodate this, we adopt the following assumption on ${\boldsymbol z}_t:=(\boldsymbol{x}_t^\top,\varepsilon_t)^\top$.

\begin{assumption}[Assumption 1 in \citealp{Adamek2023}]
\label{assumption:ned}
Let $\boldsymbol{z}_t = (\boldsymbol{x}_t^\top, \varepsilon_t)^\top$,
and let there exist constants $\bar m>\tilde{m}>2$, and $d\geq \max\{1,(\bar m/\tilde{m}-1)/(\bar m-2)\}$ such that
\begin{enumerate}[label=(\roman*)]
\item\label{ass:dgpStationary}  Let $\E\left[{\boldsymbol  z}_t\right]=\bf{0}$, $\E\left[\boldsymbol{x}_t \varepsilon_t\right]=\boldsymbol{0}$, and $\max\limits_{1\leq j\leq N+1,\ 1\leq t\leq T}\E\abs{z_{j,t}}^{2\bar m} \leq C$.
\item\label{ass:dgpNED}Let $\boldsymbol{s}_{T,t}$ denote a $k(T)$-dimensional triangular array that is $\alpha$-mixing of size $-d/(1/\tilde{m}-1/\bar{m})$ with $\sigma\text{-field}$ $\mathcal{F}^{\boldsymbol{s}}_t:=\sigma\left\lbrace\boldsymbol{s}_{T,t},\boldsymbol{s}_{T,t-1},\dots\right\rbrace$ such that $\boldsymbol{z}_t$ is $\mathcal{F}^{\boldsymbol{s}}_t$-measurable. The process $\left\lbrace z_{j,t}\right\rbrace$ is $L_{2\tilde{m}}$-near-epoch-dependent (NED) of size 
$-d$ on $\boldsymbol{s}_{T,t}$ with positive bounded NED constants, uniformly over $j=1,\ldots,N + 1$.
\end{enumerate}
\end{assumption}

The near-epoch dependence (NED) framework in Assumption \ref{assumption:ned} encompasses a broad class of commonly encountered time series processes. These include, but are not limited to, linear processes including ARMA models, various  stochastic volatility and GARCH
specifications  and strong mixing
processes. For further discussion and formal definitions, we refer the reader to \citet{Adamek2023}.

To derive an error bound for StarTime in a high-dimensional time series set-up, we first rewrite the objective function such that the penalties are uniquely expressed in terms of ${\boldsymbol\gamma}$.
We do this to streamline notation and since we can express the effective sparsity in terms of the non-zero elements of ${\boldsymbol\gamma}^0$ (see Assumption \ref{assumption:sparsity}).
Using $\boldsymbol{\beta}=\boldsymbol{A}\boldsymbol{\gamma}$, we define $w(\boldsymbol{\gamma}) \in\mathbb{R}^{(N+M) \times 1}$ as given by
\begin{equation*}
    w(\boldsymbol{\gamma})
:=
\begin{pmatrix}
\boldsymbol{A}\\
\boldsymbol{I}_M
\end{pmatrix}\boldsymbol{\gamma} := {\boldsymbol D}\boldsymbol{\gamma}.
\end{equation*} 
We also set $\lambda_1=\lambda_2=\lambda$  to avoid notational clutter. This restriction can be relaxed at the cost of more cumbersome notation. The objective function is then given by

\begin{align*}
\frac{1}{2T}\left\lVert \boldsymbol{y}-\boldsymbol{X}\boldsymbol{\beta}\right\rVert_2^2
+\lambda\left\lVert \boldsymbol{\gamma}\right\rVert_1
+\lambda\left\lVert \boldsymbol{A}\boldsymbol{\gamma}\right\rVert_1
&=
\frac{1}{2T}\left\lVert \boldsymbol{y}-\boldsymbol{X}\boldsymbol{\beta}\right\rVert_2^2
+\lambda\left\lVert w(\boldsymbol{\gamma})\right\rVert_1. 
\end{align*}
To apply StarTime successfully, $\boldsymbol{\beta}^0$ must be `sparse' under the chosen penalty. 
\begin{assumption}[Sparsity]\label{assumption:sparsity}
Assume ${\boldsymbol\beta}^0$ and ${\boldsymbol\gamma}^0$ admit the  representation
\begin{equation} \nonumber
{\boldsymbol\beta}^0 = \boldsymbol{W} {\boldsymbol\phi}^0 = {\boldsymbol A}{\boldsymbol S}{\boldsymbol\phi}^0 \quad \text{and} \quad 
{\boldsymbol\gamma}^0 = {\boldsymbol S}{\boldsymbol\phi}^0,
\end{equation}
where $\boldsymbol{W}\in \{0,1\}^{N \times Q}$ 
is an aggregation matrix describing the coarsest aggregation structure of ${\boldsymbol\beta}^0$.
We assume that this aggregation structure is contained in
${\boldsymbol A}$: There exists a column-selection matrix ${\boldsymbol S}\in \mathbb{R}^{M\times Q}$, with exactly one non-zero entry in each column, such that 
$\boldsymbol{W} = {\boldsymbol A}{\boldsymbol S}$.
The parameter ${\boldsymbol\gamma}^0$ is uniquely determined by ${\boldsymbol S}$ and ${\boldsymbol\phi}^0$.
We then define the active set on ${\boldsymbol\gamma}^0$ as 
\begin{equation*}
    S_0 := \Bigl\{l\in\{1,\dots,M\}:{\gamma}^0_l\neq 0\Bigr\}, \qquad s_0:=|S_0| = Q.
\end{equation*}
\end{assumption}

Multiple $\boldsymbol\gamma$'s may induce the same aggregation in  $\boldsymbol\beta$ (see Remark \ref{rem:non-unique-gamma}). 
Following \cite{Yan2021}, we therefore adopt their coarsest aggregating representation to uniquely define $\boldsymbol\gamma^0$ in terms of ${\boldsymbol S}$ and ${\boldsymbol\phi}^0$.
The effective sparsity in our error bound is $s_0 =  \left \lVert \boldsymbol{\gamma}^0\right \rVert_0$. The parameter of interest $\boldsymbol\beta^0$ inherits its aggregation structure from $\boldsymbol{\gamma}^0$ and does not introduce additional free parameters.
In the remainder,
let $w(\boldsymbol{\gamma})_{S_0}:= \boldsymbol D_{S_0}\boldsymbol{\gamma}_{S_0}$ denote the subvector restricted to the non-zero indices in $\boldsymbol{\gamma}^0$, where $\boldsymbol D_{S_0}$ subsets the columns  of $\boldsymbol D$ in accordance.
Let $w(\boldsymbol{\gamma})_{S_0^c}:= \boldsymbol D_{S_0^c}\boldsymbol{\gamma}_{S_0^c}$ denote the subvector restricted to the zero indices in $\boldsymbol{\gamma}^0$; by construction  $w({\boldsymbol\gamma}^0)_{S_0^c} = {\bf 0}$. 

\begin{remark}
Sparsity in $\boldsymbol\gamma^0$ encourages aggregation in $\boldsymbol\beta^0$  by limiting the number of unique (non-zero) components in $\boldsymbol{\beta}^0$. While Assumption \ref{assumption:sparsity} restricts the total number of non-zero components in $\boldsymbol{\gamma}^0$, setting an internal node in the tree to zero produces effective aggregation in practice only if several of its descendants are also zero.
\end{remark}

To obtain our desired bounds, we require a regularity condition on $\boldsymbol{X}$ that rules out excessive collinearity in sparse directions. The estimation error can be restricted to
vectors that satisfy the cone condition
$\left \lVert w(\boldsymbol{\gamma})_{S_0^c}\right \rVert_1 \leq
3\left \lVert w(\boldsymbol{\gamma})_{S_0}\right \rVert_1$
(see Lemma 4 in Appendix A.1),
so that deviations outside $S_0$ are controlled by deviations on $S_0$. Let $
\boldsymbol{\Sigma} := \sum_{t=1}^T\mathbb{E}[ \boldsymbol{x}_t\boldsymbol{x}_t^\top]/T$
denote the population covariance matrix of $\boldsymbol{X}$. The compatibility condition formalizes that, on this cone, $\boldsymbol{\beta}^\top\boldsymbol{\Sigma}\boldsymbol{\beta}$ is sufficiently
large relative to the $\ell_1$-norm of the coordinates on $S_0$.

\begin{assumption}[Compatibility Condition]
\label{assumption:compatibility}
Assume the compatibility condition holds for the set $S_0$: There exists a $\rho_{\Sigma} > 0$ such that for every $\boldsymbol{\gamma}$ satisfying $\left \lVert w(\boldsymbol{\gamma})_{S_0^c}\right \rVert_1 \leq
    3\left \lVert w(\boldsymbol{\gamma})_{S_0}\right \rVert_1 $, we have that
    
    $\rho^2_{{\Sigma}} \left \lVert w(\boldsymbol{\gamma})_{S_0} \right \rVert_1^2 \leq s_0 \boldsymbol{\beta}^\top \boldsymbol{\Sigma}\boldsymbol{\beta}.$
    
\end{assumption}
This assumption is standard in Lasso theory and weaker than requiring the
full Gram matrix to be well-conditioned;  it only imposes control on
$\boldsymbol{\Sigma}$ along a restricted set of directions. We impose the compatibility condition on the population covariance rather than on the sample covariance $\hat{\boldsymbol{\Sigma}} :=  \boldsymbol{X}^\top\boldsymbol{X}/T$, because verifying conditions on the former is generally easier than doing so  on the sample matrix (e.g., \citealp{MedeirosMendes16, Adamek2023}). The proof then proceeds by showing that, under the conditions of Assumption \ref{assumption:compatibility}, $\hat{\boldsymbol{\Sigma}}$ is sufficiently close to $\boldsymbol{\Sigma}$ with high probability, allowing the compatibility condition to carry over to the sample-based version.

While Assumption \ref{assumption:compatibility} addresses the deterministic part of the argument, we also need to control the stochastic term $\left \lVert \boldsymbol{X}^\top\boldsymbol{\varepsilon}\right \rVert_\infty/T$, which determines the choice of the tuning parameter $\lambda$. We therefore introduce a high-probability event under which choosing $\lambda$ sufficiently large ensures that the penalty dominates the stochastic term.
Define the set
$\mathcal{J} = \left\{ \|\boldsymbol{X}^\top \boldsymbol{\varepsilon}\|_\infty/T \leq \lambda_0 \right\}$
where $\lambda_0>0$ is a bound for the empirical process. 
On $\mathcal{J}$, choosing $\lambda \ge 2\lambda_0$ ensures that the penalty dominates the empirical process term in the basic inequality of Lemma 1 in Appendix A.1. 

Lemma 2 provides a high-probability bound for the empirical process and thus for the event $\mathcal{J}$, which motivates the choice of $\lambda$. In addition, since Assumption \ref{assumption:compatibility} is imposed on the population covariance matrix $\boldsymbol\Sigma$, we introduce the event $\mathcal{C}(S_0):=\Bigl\{\|\hat{\boldsymbol\Sigma}-\boldsymbol\Sigma\|_\infty \le C/s_0\Bigr\}$, under which Lemma 6 implies that the compatibility condition carries over to the sample Gram matrix $\hat{\boldsymbol\Sigma}$ on the relevant cone. Conditional on $\mathcal{J}\cap \mathcal{C}(S_0)$, combining the basic inequality in Lemma 1, the cone restriction (Lemma 4) and sample compatibility then yields the bound in Theorem \ref{theorem:oracle}.

\begin{theorem}
\label{theorem:oracle}
Suppose that Assumptions \ref{assumption:ned}, \ref{assumption:sparsity}, \ref{assumption:compatibility} hold for the active set $S_0$ with constant $\rho_0 > 0$. 
Then, on the events $\mathcal{J}$ and $\mathcal{C}(S_0)$, for $\lambda \geq 2 \lambda_0$, we have: 
\begin{equation}
\label{eq:oracle}
    \frac{1}{2T} \left \lVert\boldsymbol{X}(\boldsymbol{\hat{\beta}} - \boldsymbol{\beta}^0) \right \rVert_2^2 + \lambda \left \lVert w(\boldsymbol{\hat{\gamma}} - \boldsymbol{\gamma}^0) \right \rVert_1 \leq \frac{C \lambda^2 s_0}{\rho_0^2}.
\end{equation}
Setting $\lambda := 2C\,\frac{N^{1/\tilde{m}}(\ln(\ln(T)))^{1/\tilde{m}}}{\sqrt{T}}$, the bound in \eqref{eq:oracle} holds with probability
$\mathbb{P}(\mathcal{J} \cap \mathcal{C}(S_0)) \ge 1- C(\ln(\ln(T)))^{-1}$
for $N,T$ sufficiently large.
\end{theorem}

\begin{corollary}
\label{corollary:separatebounds}
 Under Assumptions \ref{assumption:ned}, \ref{assumption:sparsity} and \ref{assumption:compatibility}, Theorem \ref{theorem:oracle} gives the following bounds with probability at least $\mathbb{P}(\mathcal{J} \cap \mathcal{C}(S_0)) \ge 1- C(\ln(\ln(T)))^{-1}$ for $N,T$ large enough,
 \begin{enumerate}[label=(\roman*)]
     \item $\frac{1}{T} \left \lVert\boldsymbol{X}(\boldsymbol{\hat{\beta}} - \boldsymbol{\beta}^0) \right \rVert_2^2 \leq \frac{C \lambda^2 s_0}{\rho_0^2}$
     \item $\left \lVert w(\boldsymbol{\hat{\gamma}} - \boldsymbol{\gamma}^0) \right \rVert_1 \leq \frac{C \lambda s_0}{\rho_0^2}$.
 \end{enumerate}
\end{corollary}
Corollary \ref{corollary:separatebounds} follows directly from the oracle inequality  in Theorem \ref{theorem:oracle}. On the event $\mathcal{J} \cap \mathcal{C}(S_0)$, the bound in \eqref{eq:oracle} controls the sum of two non-negative terms.
Consequently, each term is individually bounded by the same right-hand side.
This yields the stated prediction and $\ell_1$-type error bounds with probability at least $\mathbb{P}(\mathcal{J}\cap\mathcal{C}(S_0))$.

\section{Simulations}
\label{sec:simulations}

We assess the performance of StarTime against suitable benchmarks in autoregressive (Section \ref{sec:AR-sim}) and mixed-frequency data generating processes (DGPs, Section \ref{sec:MF-sim}). 

We evaluate the performance of the estimators on three performance metrics related to estimation accuracy, aggregation  and sparsity recovery. Regarding estimation accuracy, we compute the
\textit{Mean Squared Error} (the lower, the better):
\begin{equation*} 
\mathrm{MSE} = \frac{1}{N} \sum_{i=1}^{D} \sum_{j=1}^{P_i} \left(\hat{\beta}^{[i]}_j - \beta^{[i]}_j \right)^2.
\end{equation*}
For aggregation recovery, we let each set of identical (zero or non-zero) components of $\boldsymbol\beta$ form a group and compare this `true' grouping structure to the estimated one based on $\hat{\boldsymbol\beta}$ using the \textit{Adjusted Rand Index} \citep{hubert1985comparing}:
\begin{equation*} 
\mathrm{ARI} = \frac{ \displaystyle
          \sum_{ij} \binom{n_{ij}}{2}
          - \frac{ \sum_i \binom{a_i}{2} \sum_j \binom{b_j}{2} }{ \binom{N}{2} }
        }{
          \frac{1}{2} \left[ \sum_i \binom{a_i}{2} + \sum_j \binom{b_j}{2} \right]
          - \frac{ \sum_i \binom{a_i}{2} \sum_j \binom{b_j}{2} }{ \binom{N}{2} }
        },
\end{equation*}        
where $n_{ij}$ is the number of elements in both group $i$ of 
$\boldsymbol\beta$ and group $j$ of $\hat{\boldsymbol\beta}$, $a_i$ and $b_j$ are the totals for the $i$th true and $j$th estimated group, and $N$ is the  number of elements in $\boldsymbol\beta$. The closer the ARI to one, the better the aggregation recovery. 

For sparsity recovery, we compute the 
\textit{F1 score}:
\begin{equation}\mathrm{F1} = \frac{2 \cdot \mathrm{TP}}{2 \cdot \mathrm{TP} + \mathrm{FP} + \mathrm{FN}}, \nonumber
\end{equation} 
where $\mathrm{TP}$ denotes the number of true positives (both $\beta_j^{[i]}$ and $\hat{\beta}_j^{[i]}$ are non-zero), 
$\mathrm{FP}$ the number of false positives ($\hat{\beta}_j^{[i]}$ is non-zero, $\beta_j^{[i]}$  is zero), and
$\mathrm{FN}$ the number of false negatives ($\hat{\beta}_j^{[i]}$ is zero, $\beta_j^{[i]}$ is non-zero). The closer the F1 score to one, the better the sparsity recovery.
We compute these metrics in each simulation run, and report results averaged across all simulation runs. We take 500 simulation runs.
All simulations were conducted in \texttt{R} \citep{Rcoreteam2025}  using the \texttt{StarTime} package. 

\subsection{Autoregressive DGPs} \label{sec:AR-sim}

\subsubsection{Data Generating Processes}
We consider three  autoregressive DGPs, inspired by the HAR model.

\noindent
\textbf{DGP 1} covers the case of aggregation but no sparsity:
\begin{equation*} 
y_{t} = 0.3 y_{t-1} - 0.2 \sum_{j=2}^5 y_{t-j} + 0.02 \sum_{j=6}^{15} y_{t-j} + 0.01 \sum_{j=16}^{20} y_{t-j} + \varepsilon_{t}.
\end{equation*}
\textbf{DGP 2} covers the case of aggregation and sparsity:
\begin{equation*} y_{t} = 0.5 y_{t-1} - 0.1 \sum_{j=2}^5 y_{t-j} + 0 \sum_{j=6}^{20} y_{t-j} + \varepsilon_{t}.
\end{equation*}
\textbf{DGP 3} covers the case of sparsity without aggregation:
\begin{equation*} y_{t} = 0.5 y_{t-1} - 0.1 y_{t-2} + 0.2 y_{t-3} + 0.05 y_{t-4} -0.3 y_{t-5} + 0 \sum_{j=6}^{20} y_{t-j} + \varepsilon_{t}.
\end{equation*}
For all three cases, we take $\varepsilon_{t} \sim \mathcal{N}(0,1)$, $T=100$ and $T=200$ with a burn-in of 200 observations. To facilitate comparisons of sparsity and aggregation structures across DGPs, we explicitly display covariates with zero coefficients. While these covariates do not contribute to the data generating process, they are included during estimation.

\subsubsection{Estimators}
We compare StarTime to three benchmarks: \textit{OLS}, \textit{Lasso} \citep{TibshiraniLasso} and  \textit{Ridge} \citep{HoerlRidge}.
OLS serves as a baseline estimator without penalization. Lasso serves as a sparse benchmark that does not perform temporal aggregation. Lasso corresponds to StarTime with $\lambda_1=0$. For fair comparison, we compute the Lasso with the same ADMM algorithm as StarTime and also include a post-selection variant. The
tuning parameter is selected using the BIC. 
Ridge serves as a regularized benchmark without sparsity or aggregation. We use the implementation in the \texttt{glmnet} package \citep{glmnet}.
The tuning parameter is selected using the BIC. 

For StarTime, we consider the Simple  and Post version (Section \ref{sec:estimators}) 
with a tree consisting of 20 (daily) leaves, aggregated into four (weekly) internal nodes and one (monthly) root. While this tree structure is inspired by the HAR model, we do not impose any temporal aggregation constraints a priori. Instead, StarTime explores the lag space in a computationally efficient manner, imposing aggregation and sparsity only when supported by the data.
To select the tuning parameters, we use the BIC with $c=1$. For Simple StarTime, we use $\tilde{\boldsymbol\beta}$ instead of $\hat{\boldsymbol\beta}$ when computing  the BIC. We report `Oracle' results for selected tuning parameters that minimize the MSE against which the BIC results can be compared.  

\subsubsection{Results}
\begin{figure}[t]
  \centering

  \begin{subfigure}[t]{0.45\textwidth}
    \centering
    \includegraphics[width=\textwidth]{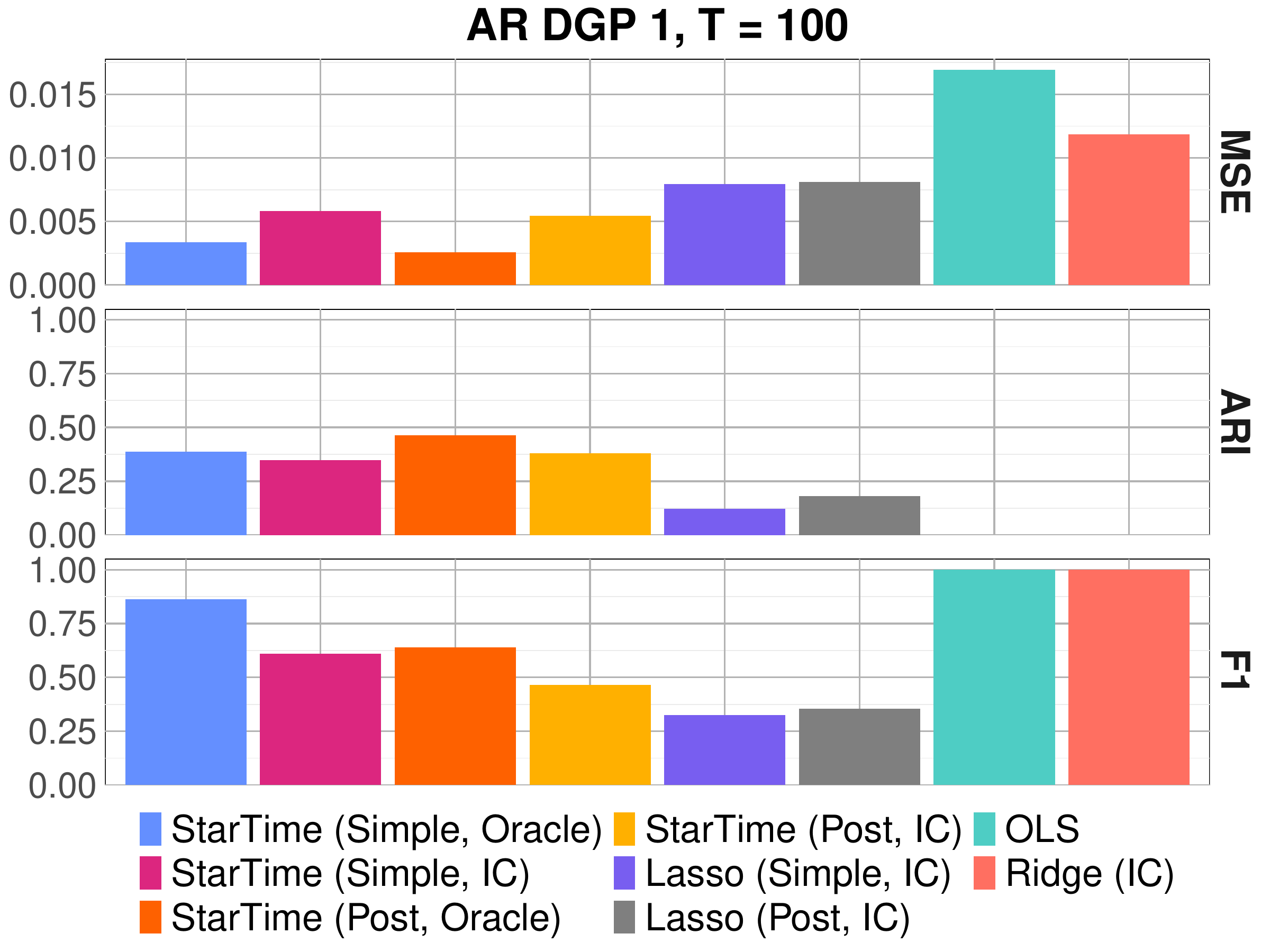}
    \caption{DGP 1, $T=100$}
    \label{fig:ARdgp1_n100}
  \end{subfigure}\hfill

  \begin{subfigure}[t]{0.45\textwidth}
    \centering
    \includegraphics[width=\textwidth]{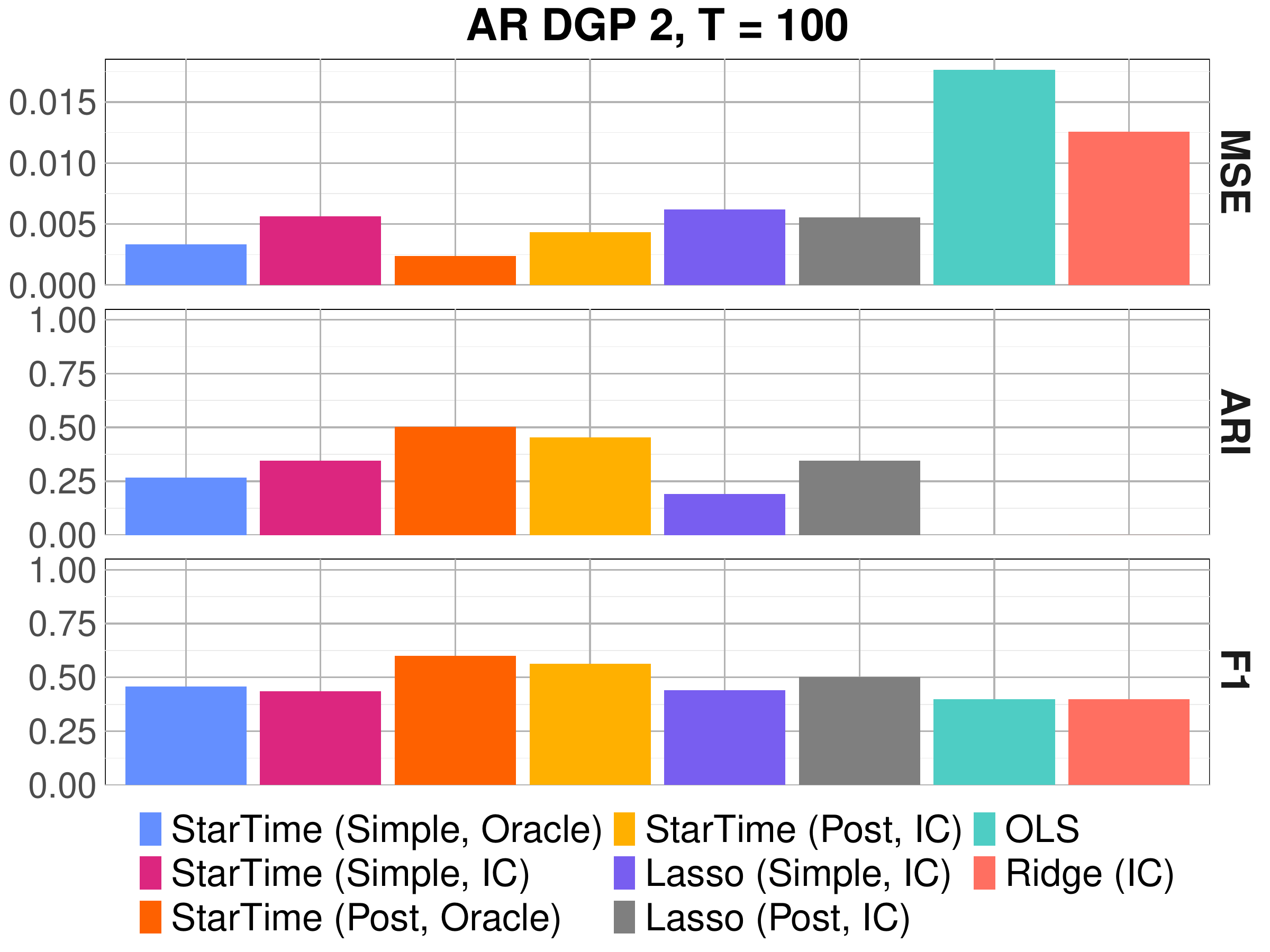}
    \caption{DGP 2, $T=100$}
    \label{fig:ARdgp2_n100}
  \end{subfigure}\hfill
  
  \begin{subfigure}[t]{0.45\textwidth}
    \centering
    \includegraphics[width=\textwidth]{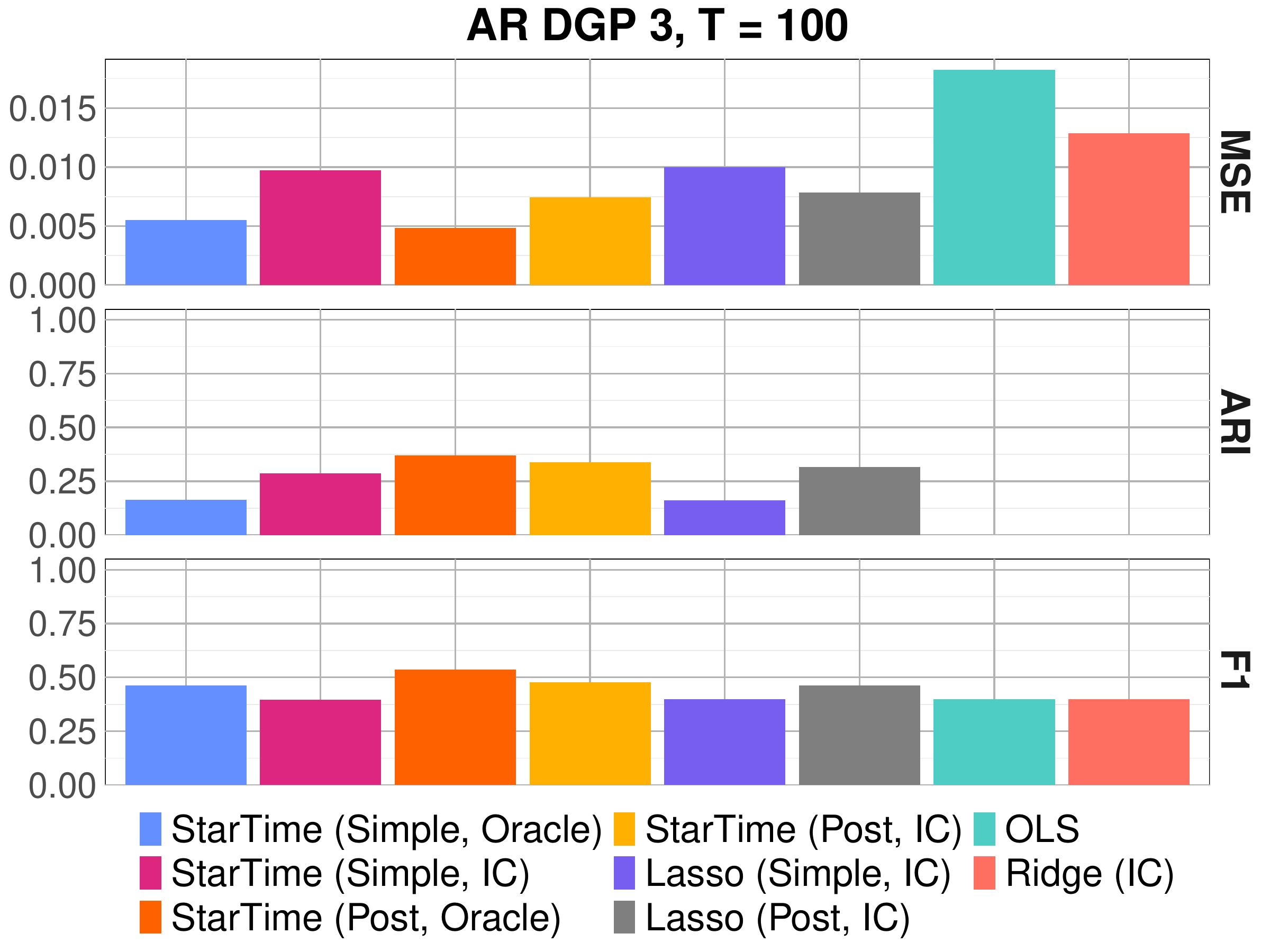}
    \caption{DGP 3, $T=100$}
    \label{fig:ARdgp3_n100}
  \end{subfigure}\hfill

  \caption{Performance metrics of the estimators across the three AR-based DGPs. 
  }
  \label{fig:simulation_results_AR}
\end{figure}

The results for the AR simulations are shown in Figure \ref{fig:simulation_results_AR} for $T=100$, Appendix B contains the results for $T=200$. 
Across all DGPs and sample sizes, the StarTime estimators, particularly its post-selection variant, demonstrate strong performance in estimation accuracy, aggregation and sparsity recovery.

With respect to MSE, StarTime consistently outperforms OLS, Ridge, and Lasso, particularly in settings where the true DGP exhibits aggregation (DGP 1 and 2). In DGP 1 (aggregation without sparsity), Simple  StarTime achieves lower MSE than competing methods, with Post StarTime displaying further improvements as the sample size increases. The performance gaps between the BIC and Oracle versions tend to reduce as the sample size increases. Both OLS and Ridge perform poorly. 
Lasso, despite its inability to exploit the underlying aggregation structure, is the most competitive benchmark. 
When aggregation is combined with sparsity (DGP 2), the StarTime estimators maintain their advantage, achieving the lowest MSE values, with Post StarTime  delivering the most accurate estimates.  
The performance of Lasso compared to StarTime improves, which is to be expected as it can now leverage the DGP's sparsity.
DGP 3 (sparsity without aggregation), naturally favors the Lasso; it performs very well.  Still, Post StarTime (IC)  matches the performance of Lasso (Post, IC), demonstrating StarTime's  flexibility in successfully identifying that sparsity, rather than aggregation,  is the relevant simplicity structure to leverage. 

Turning to aggregation recovery, StarTime, and especially Post StarTime, recovers the underlying aggregation structure in DGP 1 and DGP 2 most accurately. The low ARI is primarily attributable to StarTime’s tendency to merge the two most distant lagged groups into a single group instead. The ARI penalizes this type of clustering discrepancy severely; however, its effect on MSE remains limited.
In DGP 3, where no true aggregation exists apart from the zero group, Post StarTime occasionally wrongly identifies groups, yet it still achieves the highest ARI, competitive to that of Post Lasso.  
As expected, OLS, Ridge, and  Lasso yield an ARI of zero across all DGPs as none of them are designed to induce a grouping structure. 

Finally, regarding sparsity recognition, the F1 scores align with our expectations. 
In DGP 1, where all coefficients are non-zero, OLS and Ridge correctly refrain from penalizing parameters. Simple StarTime outperforms its post-selection variant, as the latter tends to over-penalize, introducing unnecessary shrinkage. However, this over-penalization is concentrated on more distant rather than recent lags, where coefficients are already close to zero, so the impact on the MSE remains limited. 
In DGP 2, both StarTime and Post Lasso effectively identify zero coefficients, with Post StarTime demonstrating improved recovery as $T$ increases. Finally, in DGP 3, the StarTime variants and Lasso variants exhibit very similar performance.

\subsection{Mixed-Frequency DGPs}   \label{sec:MF-sim}

\subsubsection{Data Generating Processes}
We consider three mixed-frequency DGPs, similar to \cite{Babii2022}. The response $y_t$ is observed quarterly, the covariates $x_h$ monthly; hence $m=3$. 

\textbf{DGP 1} is given by
\begin{equation}
y_{mt} = 0.3 y_{m(t-1)} + 0.01 y_{m(t-2)} + 0 y_{m(t-3)} + 0 y_{m(t-4)} 
+ \sum_{i=1}^{3} \sum_{j=1}^{12} \omega^{[i]}_{j} 
x_{mt-(j-1)}^{[i]}
+ \varepsilon_{mt}, \nonumber
\end{equation}
where we include an autoregressive component with two effective (non-zero) lags and three covariates with twelve monthly lags each. 
Each lag weight $\omega^{[i]}_{j}$ is generated from a Beta distribution with parameters Beta(1,3), Beta(2,3), and Beta(2,2), respectively. 

\textbf{DGP 2} is given by 
\begin{equation}
y_{mt} = 0.3 y_{m(t-1)} +  0.01 y_{m(t-2)} + 0 y_{m(t-3)} + 0 y_{m(t-4)} + \sum_{i=1}^{10} \sum_{j=1}^{12} \omega^{[i]}_{j}
x_{mt-(j-1)}^{[i]}
+ \varepsilon_{mt}, \nonumber
\end{equation}
where $\omega^{[i]}_{j}$ is generated as in DGP 1 for the first three covariates and zero otherwise. DGP 2 is thus the same as DGP 1, but when estimating, we include the seven additional irrelevant covariates.

\textbf{DGP 3} is given by
\begin{equation}
y_{mt} = 0.3 y_{m(t-1)} + 0.01 y_{m(t-2)} + 0 y_{m(t-3)} + 0 y_{m(t-4)} + \sum_{i=1}^{3} \bar{\omega}^{[i]}\sum_{j=1}^{12} x_{mt-(j-1)}^{[i]}
+ \varepsilon_{mt}, \nonumber
\end{equation}
with $\bar{\omega}^{[i]} = \frac{1}{12} \sum_{j=1}^{12} \omega_{j}^{[i]}$ the average of the monthly Beta weights for covariate $i$.

Across all DGPs, all covariates are generated as AR(1) processes $x_h = 0.2 x_{h-1} + \varepsilon_h$ with $\varepsilon_h \sim \mathcal{N}(0,1)$, where $h$ denotes the monthly frequency, and the errors  $\varepsilon_{mt} \sim \mathcal{N}(0,1)$. For DGP 3, we retain the same Beta distribution parameters but replace the unit-sum normalization with covariate-specific scaling, so that the weights of each covariate sum to different values.

\subsubsection{Estimators}
We compare StarTime to three benchmarks: \textit{OLS}, \textit{Restricted MIDAS} \citep{MIDAStouch} and \textit{MIDAS-ML} \citep{Babii2022}.
OLS serves as a baseline estimator without penalization.
Restricted MIDAS serves as a baseline estimator  for mixed-frequency regressions.
We use the implementation in the \texttt{midasr} package \citep{midasr} with a normalized Almon lag polynomial using parameters $(1, -0.5)$.
MIDAS-ML serves as a sparse benchmark estimator  for mixed-frequency regressions, which allows for variable selection in high-dimensional mixed-frequency regressions. We use the default implementation in the \texttt{midasml} package \citep{midasml}.

For StarTime, we use the trees in Figure \ref{fig:MultiTree} for the monthly covariates and the quarterly one.
For DGP 1 and DGP 2, the tree structure is misspecified; the natural groupings implied by the Beta-distributed weights do not align perfectly with the tree hierarchy. In contrast, 
DGP 3 constructs covariate weights by averaging the Beta-generated lag weights within each covariate, assigning a common coefficient to all lags of a given covariate. The tree structure can capture such temporal aggregation.
We report results for the Oracle and BIC-based selection of the tuning parameters. For  the BIC, we set $c=1$ in DGP 1 and 3 but lower it to $c=0.5$ 
in DGP 2 with $T=100$ 
to reduce the risk of overfitting since the number of parameters ($N=124$) exceeds the sample size. This adjustment limits the maximum number of non-zero estimated parameters to 47, which still allows the model to capture key signals.

\subsubsection{Results}
The results for the mixed-frequency simulations are shown in  
Figure \ref{fig:simulation_results_Mixed} for $T=100$, Appendix B  contains the results for $T=200$.
In terms of MSE, distinct performance profiles emerge depending on the structure of the lag weights. In DGP 1,  MIDAS  and the Oracle-tuned StarTime perform best. The smooth Beta constraints of MIDAS are well-suited to this structure. The BIC-tuned StarTime slightly outperforms MIDAS-ML, suggesting that StarTime remains competitive even when the true DGP lacks the tree-based structure it is designed to exploit.  Post StarTime  exhibits a slightly higher MSE than the Simple StarTime, due to the more pronounced imposition of temporal aggregation where none exists. 

\begin{figure}
  \centering
  \begin{subfigure}[t]{0.45\textwidth}
    \centering
    \includegraphics[width=\textwidth]{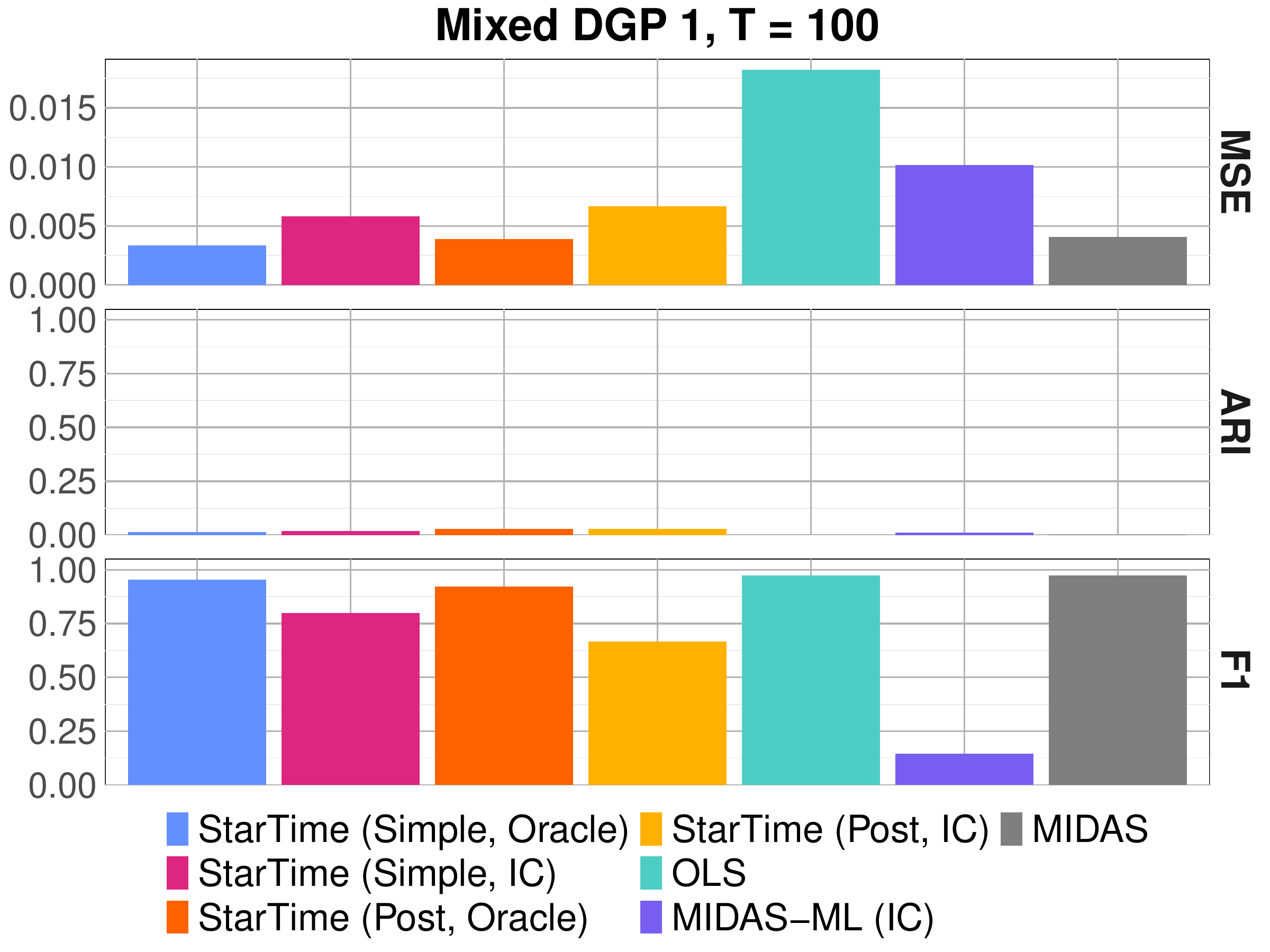}
    \caption{DGP 1, $T=100$}
    \label{fig:Mixeddgp1_n100}
  \end{subfigure}\hfill

  \begin{subfigure}[t]{0.45\textwidth}
    \centering
    \includegraphics[width=\textwidth]{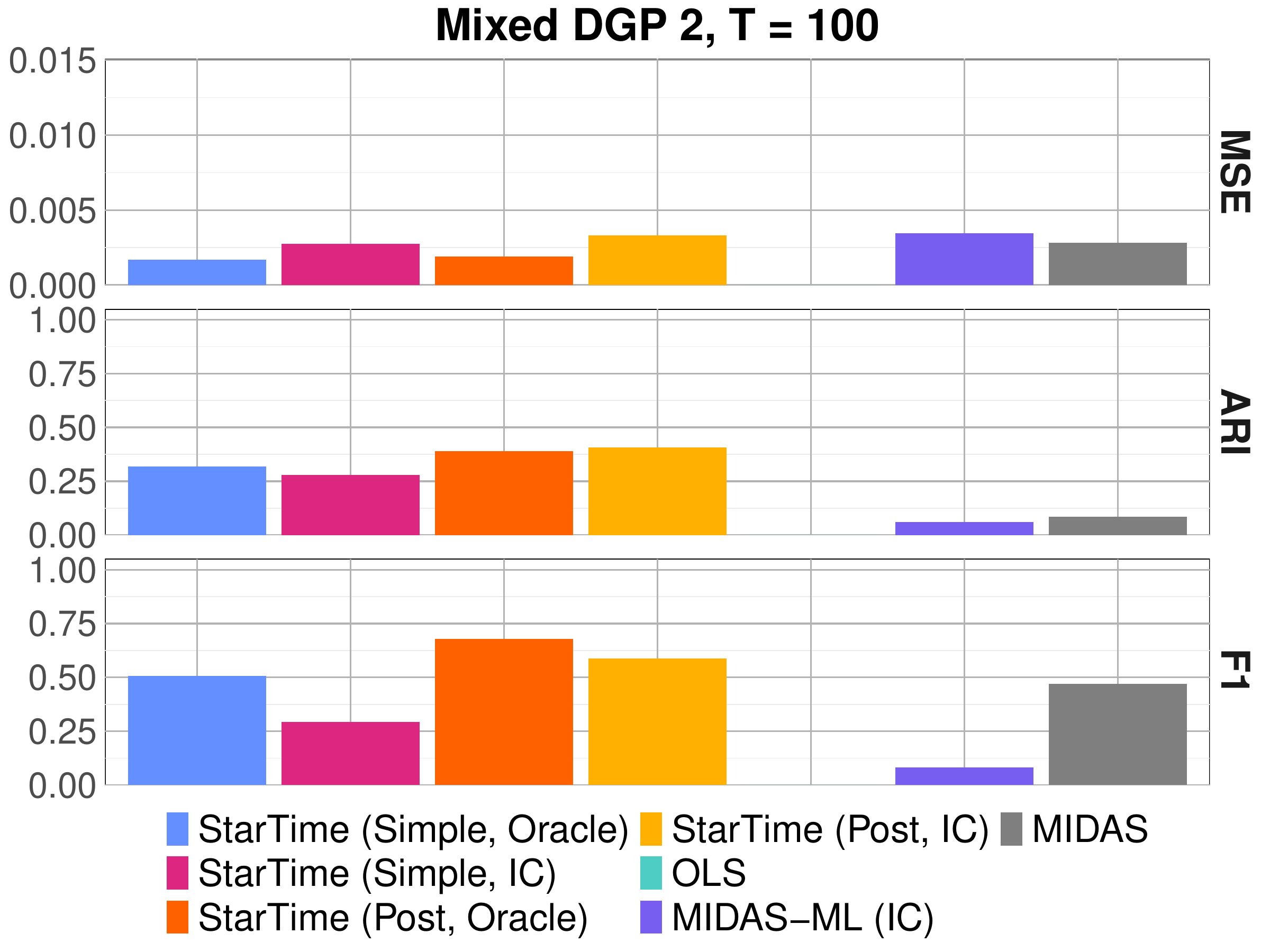}
    \caption{DGP 2, $T=100$}
    \label{fig:Mixeddgp2_n100}
  \end{subfigure}\hfill

  \begin{subfigure}[t]{0.45\textwidth}
    \centering
    \includegraphics[width=\textwidth]{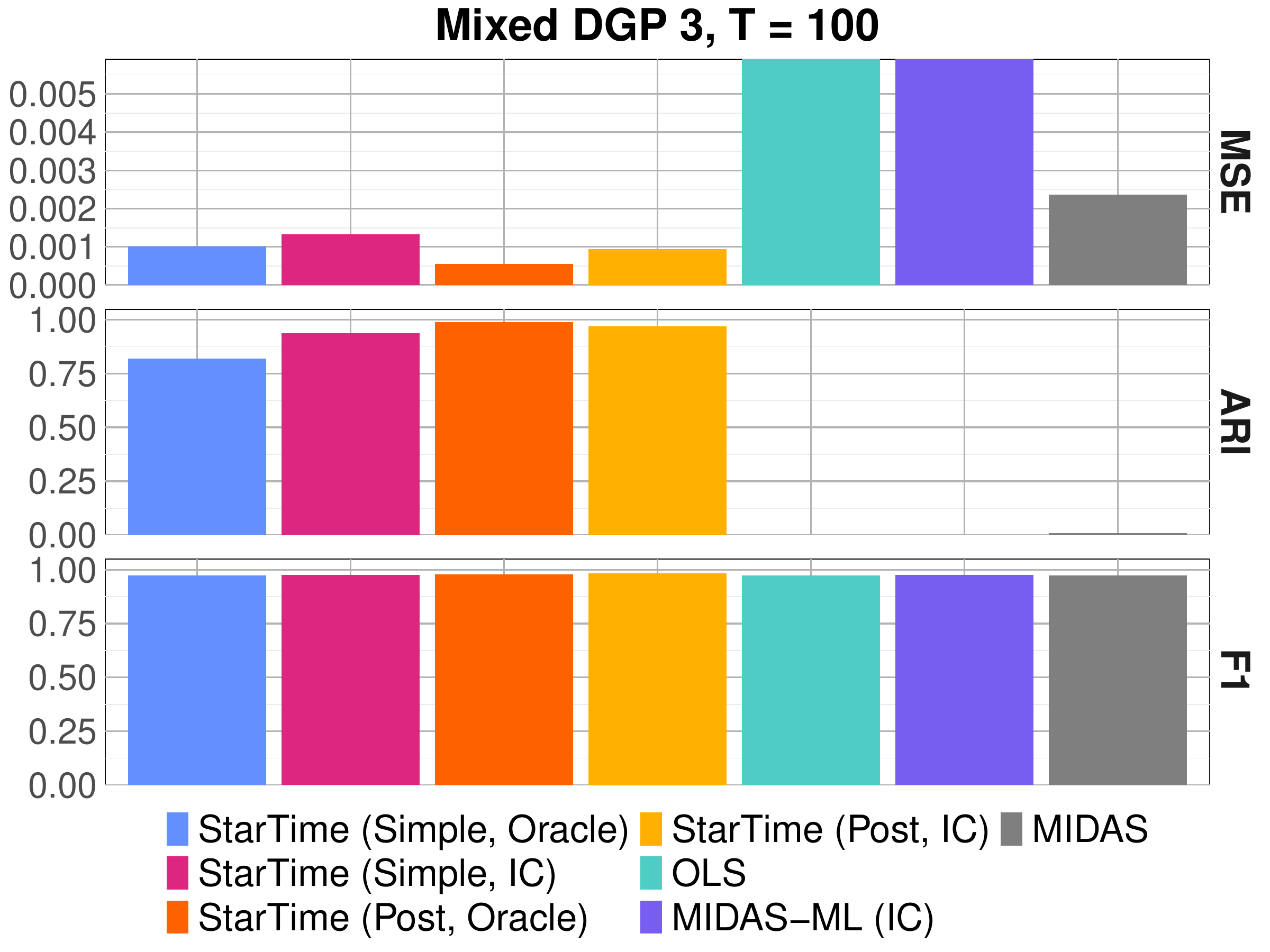}
    \caption{DGP 3, $T=100$}
    \label{fig:Mixeddgp3_n100}
  \end{subfigure}\hfill

  \caption{Performance metrics of the estimators across the three mixed-frequency 
  DGPs. 
  }
  \label{fig:simulation_results_Mixed}
\end{figure}

In the high-dimensional setting of DGP 2, the limitations of unpenalized methods are apparent. For $T=100$, where the number of parameters exceeds the sample size, OLS is infeasible. For $T=200$, OLS yields a substantially higher MSE than the penalized methods, including StarTime, which maintain their performance. In DGP 3,  where the weights exhibit a clear aggregation structure, 
the StarTime estimators perform best, as expected.
MIDAS-ML performs considerably worse: Its MSE is 0.017 for $T=100$ and 0.006 for $T=200$; for clarity of presentation, the vertical axis is capped. MIDAS-ML appears sensitive to the chosen lag structure when it does not align with its smoothness restrictions across the lags. In contrast, StarTime remains competitive in DGP 1 and 2 to MIDAS-ML, even with a misspecified tree structure. 

Regarding aggregation recovery, 
in DGP 1, all ARIs are low since there is no aggregation to uncover. 
In DGP 2, the only relevant grouping structure to uncover is the separation of the `group' of irrelevant covariates from the relevant ones. The StarTime estimators best capture this separation. In DGP 3, where temporal aggregation is present, the StarTime estimators stand out as the strongest performers. 

Regarding sparsity recovery, in DGP 1, where nearly all coefficients are non-zero, all methods apart from MIDAS-ML achieve high F1 scores. In DGP 2, where 7 irrelevant covariates are added compared to DGP 1, StarTime and especially its post-selection variant demonstrate superior performance. The refitting step considerably improves sparsity recovery, especially in high-dimensional settings where the signal is weaker. In DGP 3, where coefficients feature strong aggregation but minimal sparsity, all methods achieve excellent recovery rates, especially for $T=200$.

Overall, across both autoregressive and mixed-frequency DGPs, the StarTime estimators demonstrate strong performance in terms of estimation accuracy, aggregation and sparsity recovery compared to alternatives. 
They adaptively emphasize aggregation or sparsity in line with the underlying DGP. Even when the temporal tree structure is misspecified, StarTime remains competitive.

\section{Empirical Applications}
\label{sec:application}
We demonstrate StarTime's versatility  on two empirical applications: One on realized variance forecasting 
and another on
macroeconomic nowcasting and forecasting.

\subsection{Financial Application} \label{sec:finance}

\subsubsection{Realized Variance Forecasting}
We apply StarTime to realized variance forecasting. The HAR model by \cite{Corsi2009} is known to perform extremely well; it consistently outperforms a range of statistical and machine learning approaches \citep{HARdToBeat}. The HAR relies on a fixed temporal aggregation structure grounded in financial knowledge, reflecting the investment horizons of different investors. Our goal is not to outperform the HAR in terms of forecasting ability. While StarTime contains the HAR as a special case, it does not impose the HAR’s aggregation structure a priori. Instead, it allows for multiple levels of temporal aggregation across a richer lag space, potentially combined with sparsity. Using StarTime, we aim to provide data-driven evidence on whether the HAR’s aggregation structure is supported across a wide range of stocks and fitting schemes.
We employ the realized variance dataset from \cite{Margaritella2023granger} 
which consists of 10-min realized variances (RV) from March 20th, 2008 to February 17th, 2017 ($T=2236$ trading days) for a set of 30 major U.S. financial assets. Appendix C contains an overview of the considered stocks.  

\subsubsection{Model and Benchmarks}
Following standard practice in the literature, we model the logarithmic transformed realized variance, $y_{t} = \log(RV_{t})$. We consider an autoregressive model where $y_{t+h}$ is predicted using $N$ lags of $y_t$.
We focus on forecast horizons of $h=1$ (one day), $h=5$ (one week), and $h=20$ (one month) and use a direct forecast approach. 

For the StarTime estimators, we specify a maximum lag length of $N=20$ days (approximately one trading month), and use the tree in Figure \ref{fig:HARTree}. We also perform robustness checks allowing for a deeper history of $N=40$ lags, using a tree that aggregates the 40 daily lags into 8 weekly groups (each containing 5 days), which are then aggregated into two monthly groups, and a single `two-month' root. 

We compare StarTime against three benchmarks: The HAR model, a simple AR(1), and a Random Walk. The HAR model is given by
\begin{equation}
y_{t+h} = \beta_0 + \beta^{(d)} y_t + \beta^{(w)} \left( \frac{1}{5} \sum_{j=1}^5 y_{t-j+1} \right) + \beta^{(m)} \left( \frac{1}{20} \sum_{j=1}^{20} y_{t-j+1} \right) + \varepsilon_{t+h}.
\end{equation}
To evaluate the performance of the methods, we use a rolling window approach. 
For forecast horizon $h$, we fix the window size at $W-h+1$ and consider $W \in \{125, 250, 1000\}$ (approximately one semester, one year, and four years of trading days respectively). As commonly done for penalized estimators, we standardize the response in each rolling window to have zero mean and unit variance. We tune the parameters $\lambda_1$ and $\lambda_2$ for StarTime in each rolling window using the BIC with $c=1$.

\subsubsection{Evaluation Metrics}

We evaluate out-of-sample forecasting performance using two widely used loss functions: Mean Squared Forecast Error (MSFE) and Quasi-Likelihood (QLIKE) loss. While the MSFE is a symmetric loss, the QLIKE is asymmetric and more robust to extreme outliers 
For $N_{test}$ out-of-sample forecasts, the QLIKE is defined as: 
\begin{equation*}
\text{QLIKE} = \frac{1}{N_{test}} \sum_{t=1}^{N_{test}} \left( \frac{\exp(y_{t})}{\exp(\hat{y}_{t})} - \log \left( \frac{\exp(y_{t})}{\exp(\hat{y}_{t})} \right) - 1 \right),
\end{equation*}
where we exponentiate the forecasts and actuals to compute the QLIKE on the original scale of realized variances (e.g., \citealp{HARdToBeat}).

We assess forecast accuracy using the Diebold-Mariano (DM) test \citep{DMtest} for pairwise comparisons and the Model Confidence Set (MCS, \citealp{MCS}) for multi-model evaluation. For the DM tests, implemented in the \texttt{forecast} package \citep{forecast}, we test the null of equal predictive accuracy between Post StarTime and each benchmark  for every stock. The tests are conducted at the 5\% significance level. We summarize results by reporting the Win Rate and Loss Rate  across stocks for each configuration.
The Win Rate is the percentage of stocks where Post StarTime significantly outperforms the benchmark ($p$-value $< 0.05$ and test statistic $< 0$).
The Loss Rate is the percentage of stocks where Post StarTime significantly underperforms ($p$-value $< 0.05$ and test statistic $> 0$). The remaining proportion represents cases where the models are statistically indistinguishable (ties).
For the MCS evaluation, we use the $T_{max}$ statistic with $B=5000$ bootstrap replications and $\alpha=0.05$, implemented in the package \texttt{MCS} \citep{MCSpack}. We report the MCS inclusion rate, defined as the percentage of stocks for which a method is included in the final MCS.

\subsubsection{Results}
Figure \ref{fig:financeapplication} presents the results for $h=1$ on the $250$ and $1000$-day window, Appendix D contains the results for the $125$-day window. We display the median MSFE across stocks, over time for the five methods. We also compute the ARI that compares the grouping structure of Post StarTime with that of the HAR. Since the raw daily results are rather noisy, we use a centered moving average filter to smooth the results for visual purposes: For a chosen window size $S$, the smoothed value at time $t$ is the simple average of observations in the interval $[t - \lfloor S/2 \rfloor, t + \lfloor S/2 \rfloor]$; we use $S=20$.

\begin{figure}[t]
  \centering

  \begin{subfigure}[t]{0.49\textwidth}
    \centering
    \includegraphics[width=0.9\textwidth]{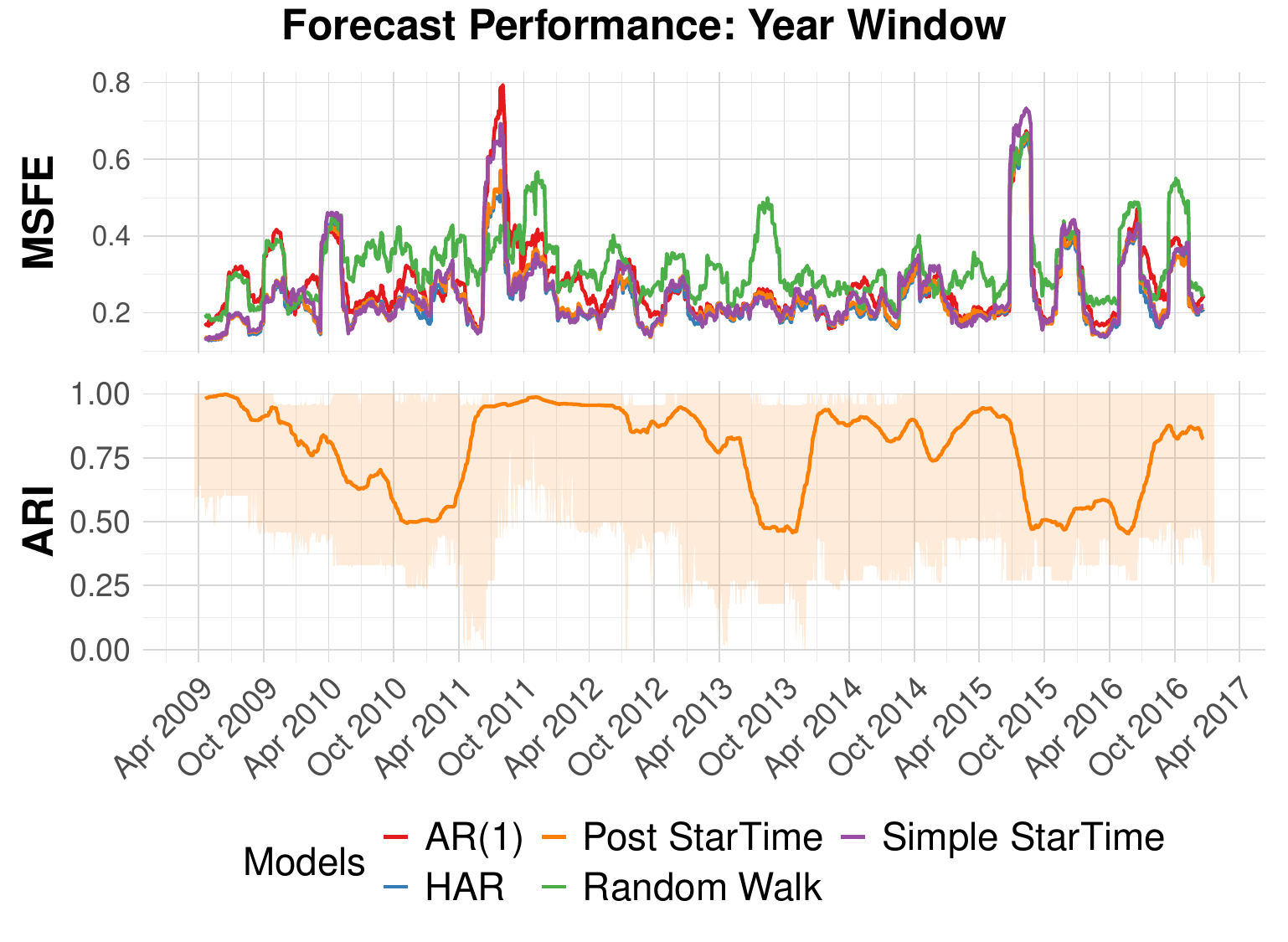}
    \caption{250 Day-window, $h=1$, 20 lags}
    \label{fig:financeapplication_250}
  \end{subfigure}
  
  \begin{subfigure}[t]{0.49\textwidth}
    \centering
    \includegraphics[width=0.9\textwidth]{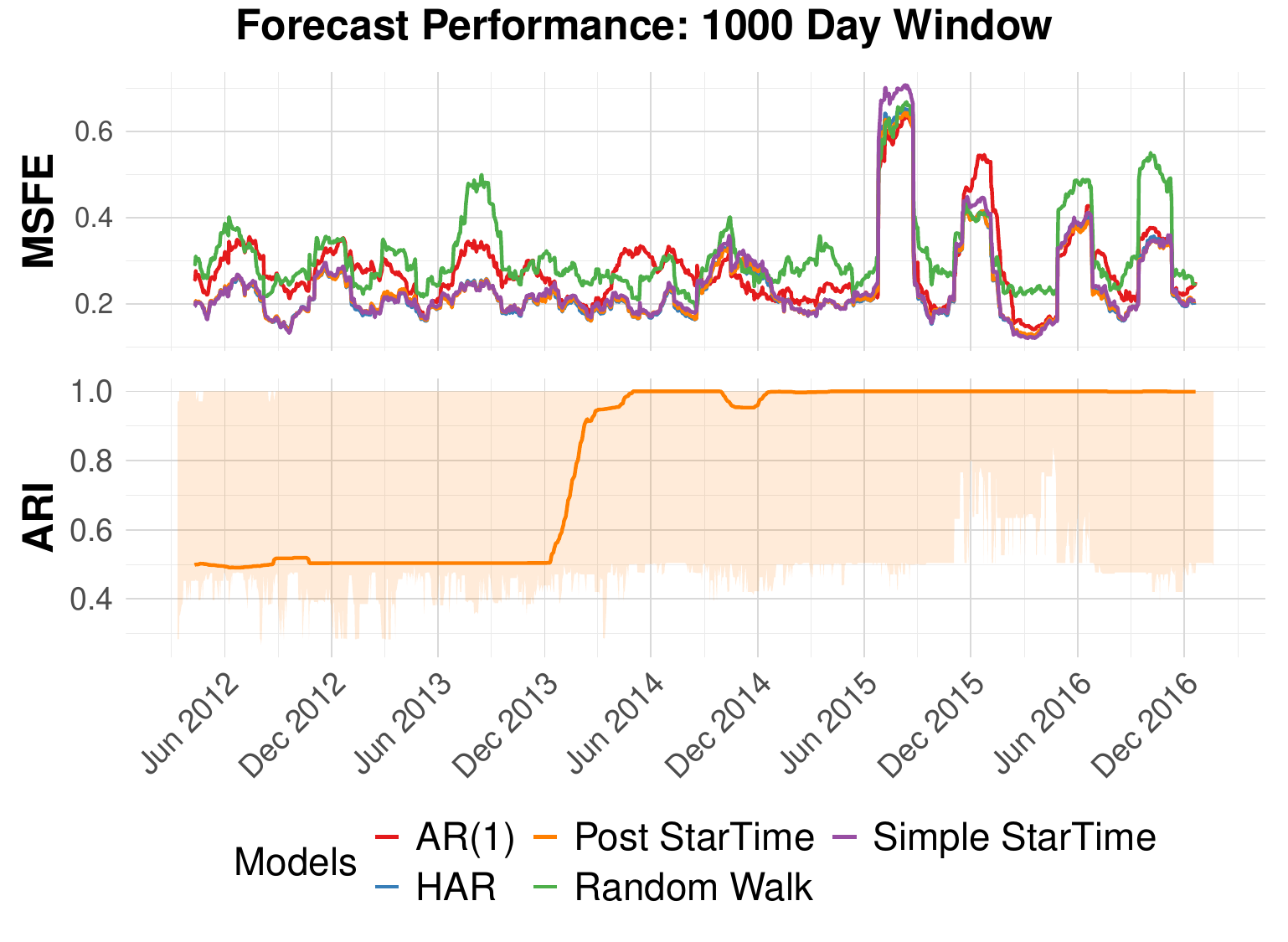}
    \caption{1000 Day-window,  $h=1$, 20 lags}
    \label{fig:financeapplication_1000}
  \end{subfigure}

  \caption{Comparative forecast performance across 250- and 1000-day windows. Each panel displays the smoothed daily median MSFE (top) and smoothed daily median ARI (bottom). Shaded regions represent the 10\textsuperscript{th} and 90\textsuperscript{th} percentiles.} 
  \label{fig:financeapplication}
\end{figure}

We first examine the MSFE curves for $h=1$ in the 20-lag configuration. The HAR model is the overall best performer, with Simple and Post StarTime tracking its performance closely. Simple StarTime occasionally exhibits high error peaks, especially in July--October 2015. This instability is mitigated by Post StarTime which produces more parsimonious, aggregated estimates.
The AR(1) and, most notably, the Random Walk exhibit higher errors throughout the sample period.  

The MCS results (Table D.1, Appendix D) confirm that the HAR is the best performer. Its dominance is, however, less pronounced for $h=20$ compared to $h=1$ and $h=5$. Indeed, the inclusion rate of the StarTime estimators increases as the forecast horizon and/or the window size increases.
The AR(1) is competitive for longer horizons. 
The Random Walk yields the worst performance, a result consistent with the known mean-reverting properties of volatility that the Random Walk fails to capture at such short horizons \citep{PoonGranger2003}. 

The Diebold-Mariano (DM) tests (Table D.2 of Appendix D) further confirm that the HAR model performs, overall, best for horizon $h=1$. The StarTime estimators are  competitive to the HAR for longer horizons, suggesting that our data-driven tree-based aggregation effectively captures longer-run dynamics. Post StarTime consistently 
outperforms the Random Walk and AR(1), and
outperforms or matches  Simple StarTime for next-day forecasts under MSFE loss.

Next, we turn to the ARI curves in  Figure \ref{fig:financeapplication}. The results for the 1000-day window demonstrate most stability: The median ARI never falls below 50\% and remains steadily around one from May 2014 onward. 
Even for the 250-window size, the ARI mostly fluctuates between 0.5 and 1.
Figure \ref{fig:financeapplication_matrix} further confirms that StarTime provides data-driven support for the HAR lag structure. 
The heatmaps for two stocks, Nike Inc.\ and JPMorgan Chase \&Co, visualize the temporal aggregation structure recovered by Post StarTime on the 1000-day window. Without a priori imposing the HAR-based structure,  it is oftentimes exactly recovered across the vast majority of time points. 
The $250$-day window  displays a  more volatile heatmap (see Appendix D), even though the HAR-based structure still shines through. 

\begin{figure}[t]
  \centering

  \begin{subfigure}[t]{0.45\textwidth}
    \centering
    \includegraphics[width=0.8\textwidth]{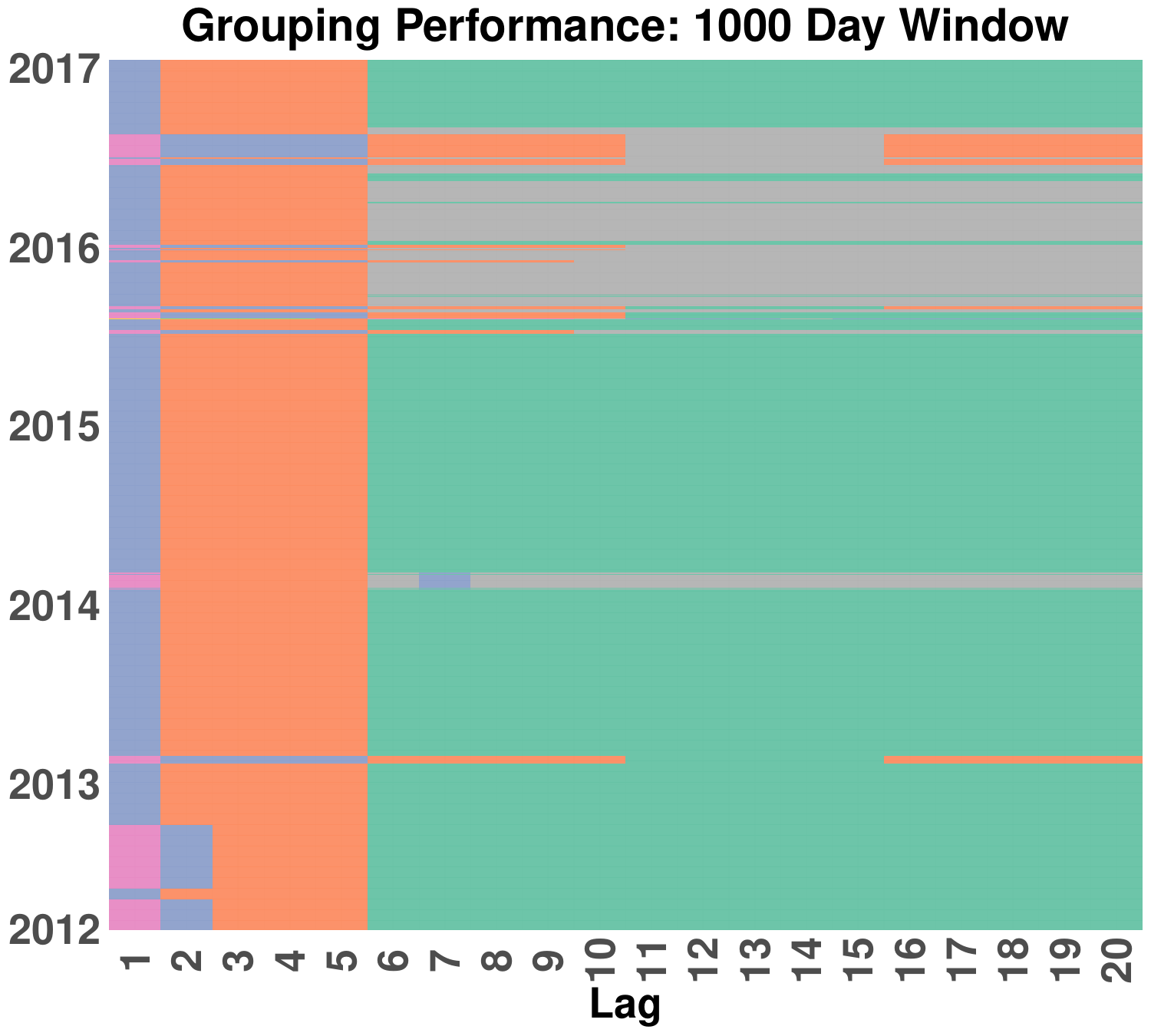}
    \caption{Nike Inc.} 
    \label{fig:finance_matrix_1000_nke}
  \end{subfigure}\hfill
  \begin{subfigure}[t]{0.45\textwidth}
    \centering
    \includegraphics[width=0.8\textwidth]{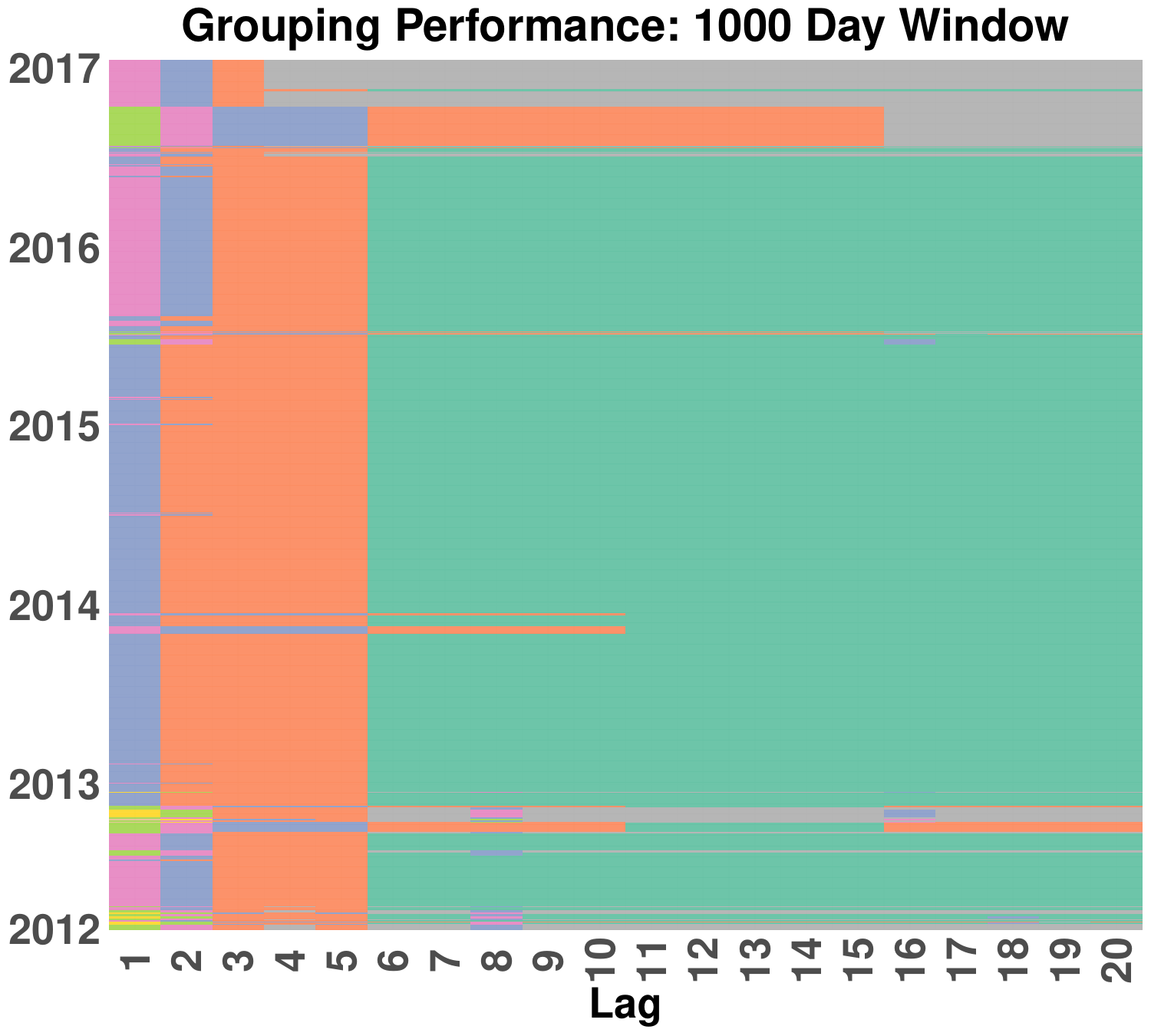}
    \caption{JPMorgan Chase \& Co.} 
    \label{fig:finance_matrix_1000_jpm}
  \end{subfigure} 
  \caption{Coefficient heatmaps for Post StarTime ($h=1$, 20 lags) applied to Nike Inc. (NKE) and JPMorgan Chase \& Co (JPM), with 1000-day window size. 
  Each row indexes an estimation date, columns correspond to daily lags $j=1,\ldots,20$. Cell colors encode the aggregation structure selected by Post StarTime: Lags with identical colors are fused into a single coefficient group, distinct colors within a row reflect heterogeneous lagged effects. Coefficients set to zero are indicated in gray.} 
  \label{fig:financeapplication_matrix}
\end{figure}

QLIKE results are included in Tables D.3 and D.4 of Appendix D. 
The same conclusions as for the MSFE hold; we highlight the differences. The Random Walk occasionally appears in the MCS, but the pairwise DM-tests indicate that it never statistically outperforms Post StarTime. 
When comparing StarTime with the AR(1)  using the QLIKE, the former is generally favored overall, while the latter remains  robust at longer horizons, as it is consistently included in the QLIKE MCS for $h=20$.

Finally, the above findings largely extend to the configurations with 40 lags (Tables D.1 to D.4 in Appendix D). 
A notable difference occurs for the QLIKE loss, where the AR(1) tends to outperform StarTime. 
The inclusion of 20 additional lags substantially increases the dimensionality of the lag space, which inherently complicates lag selection and aggregation by StarTime, especially in smaller rolling window fitting schemes.
Since the HAR structure (i.e., the first 20 lags) already captures volatility dynamics well, additional lags may be redundant. While the AR(1) avoids this overfitting by design, StarTime's performance hinges on shrinking these distant lags to zero. The aggregation patterns for individual stocks confirm this: 
For $W=1000$, StarTime removes the  20–40 daily lags and maintains its accuracy. In smaller windows, however, it struggles to consistently eliminate these  distant lags, resulting in a performance dip relative to the parsimonious benchmarks. When domain knowledge about the appropriate lag length is available, we advise incorporating this directly rather than relying on an excessively large lag space. 

Overall, our findings suggest that StarTime remains a strong and flexible estimator that offers  data-driven evidence that the HAR’s temporal aggregation structure is strongly supported across a wide range of stocks, particularly for larger window sizes.

\subsection{Macroeconomic Application} \label{sec:macro}

\subsubsection{Nowcasting and Forecasting GDP Growth}
We now evaluate StarTime's ability to nowcast and forecast GDP growth using mixed-frequency economic and financial time series.
We construct a dataset from 1992:Q2 to 2025:Q2 that captures a broad overview of the US economy. 
The dataset comprises 30 variables covering diverse economic categories: Market indices, economic uncertainty policy, financial indicators, labor market metrics, income and sales, capacity utilization, prices, housing, and production inventories.
Our variable choices are grounded in economic theory and inspired by recent mixed-frequency literature \citep{midasml, TernesHierarchical}. The raw series are collected from the FRED, Yahoo Finance, and the Policy Uncertainty website \citep{EPUweb}. 

We transform the variables following the transformation codes in FRED-MD \citep{mccracken2016fredmd} and FRED-QD \citep{mccracken2021fredqd}. For financial and uncertainty series not covered therein, we apply logarithmic transformations when supported by the literature: We use the log of the Economic Policy Uncertainty (EPU) index \citep{baker2016EPU} and the log of the VIX \citep{andersen2003VIX, taylor2019VIX}. The National Financial Conditions Index (NFCI) enters the model in levels \citep{Amburgey2022NFCI}. For any remaining variables, unit root tests are conducted using the \texttt{bootUR} package \citep{BootUR} to determine the appropriate order of integration. 
Appendix E contains a detailed overview of the variables, their sampling frequency, seasonal-adjustment status, and the applied transformations.
We include all series at the highest sampling frequency available across sources, and use StarTime to provide data-driven guidance on the appropriate temporal aggregation level at which each variable should enter the model.

\subsubsection{Model and Benchmarks}
We consider model \eqref{eq:BaseEquation} with quarterly GDP growth as response and daily, weekly, monthly and quarterly covariates. We use a lag structure that  spans the previous quarter and set $m=60$. 
For the daily covariates, we include 60 lags (approximately 3 months), for the weekly covariates, we include 12 lags (approximately 12 weeks per quarter), for the monthly ones, we include 3 lags and for the quarterly variables we include one lag.
The full model contains 5 daily, 3 weekly, 20 monthly, and 2 quarterly regressors (including the response), yielding a total of $N=398$ predictors. 
We also consider a `reduced' model with $N=49$ predictors;  2 weekly and 8 monthly core variables alongside the GDP lag  (see Tables E.1 to E.4 of Appendix E). 

\begin{remark}
All series are aligned to the quarterly forecast target. Weekly variables require attention, as they do not map cleanly into months or quarters. To prevent look-ahead bias, we employ a hard assignment rule based on the week-ending Saturday: A week is attributed to a quarter only if its Saturday falls within that quarter. A week straddling two quarters is then assigned to the period where it ends. 
\end{remark}

Model \eqref{eq:BaseEquation} allows us to forecast GDP growth for the next quarter.
We also consider a nowcasting set-up where we include lags of the daily, weekly and monthly covariates in the ongoing quarter; hence using lags $x_{mt - m_i(j - 1)}$
instead of  $x_{m(t-1) - m_i(j - 1)}$ for those covariates $i$ with $m_i < m$ in model \eqref{eq:BaseEquation}. We employ a rolling window with window sizes $W \in \{66, 105\}$, corresponding to approximately 50\% and 80\% of the sample, respectively. Within each window, all predictors and the response are standardized to have zero mean and unit variance.

For StarTime, we specify one temporal tree per regressor. For daily regressors, the tree starts from 60 leaves, aggregated into 12 weeks of 5 days, then into 3 `months' of 4 weeks, and  finally into a single quarterly node. 
Weekly trees follow the same logic, but start from 12 weekly lags,  aggregated into 3 month-level nodes and then into one quarterly node. Monthly trees have two levels, with 3 monthly lags aggregated into a quarter, while quarterly regressors are represented by a single-node tree.

We tune $\lambda_1$ and $\lambda_2$ for StarTime in each rolling window using the BIC with $c=0.3$. Due to the pronounced high-dimensionality of the full model, the regularization grid may, however, contain a large region of tuning parameter values that yield too complex models. To address this, we implement a two-step grid construction: We start with the initial grid, then redefine the grid by restricting $\lambda_1$ and $\lambda_2$ to the empirically relevant range (between the largest and smallest values that yield admissible solutions according to the threshold $c$), and re-estimate on this refined grid. 

We benchmark the performance of StarTime in terms of MSFE against three alternatives: The Random Walk and AR(1) serve as standard forecast baselines but they cannot deliver a nowcast.
We therefore also compare our approach to MIDAS-ML \citep{midasml}, a state-of-the-art method for high-dimensional mixed-frequency regressions that also utilizes penalization and lends itself well to a nowcast and forecast set-up. As in the simulation study, MIDAS-ML is implemented using the \texttt{midasml} package with default settings, specifically we set the mixing parameter to one which corresponds to a sparse-group Lasso. 

\begin{remark}
A fully fledged nowcasting exercise requires explicitly accounting for real-time data vintages and data revisions. We deliberately adopt a simplified nowcasting framework to isolate and illustrate the potential of StarTime in a nowcast-versus-forecast setting. 
For a fair comparison, we impose the same set-up for both StarTime and MIDAS-ML.
While this simplified set-up is informative in its own right, the results should be interpreted accordingly. Extending to a full real-time nowcasting exercise is feasible within our framework but beyond the scope of the  paper.
\end{remark}

\subsubsection{Results}

\begin{figure}[!h]
  \centering
    \includegraphics[width=\textwidth]{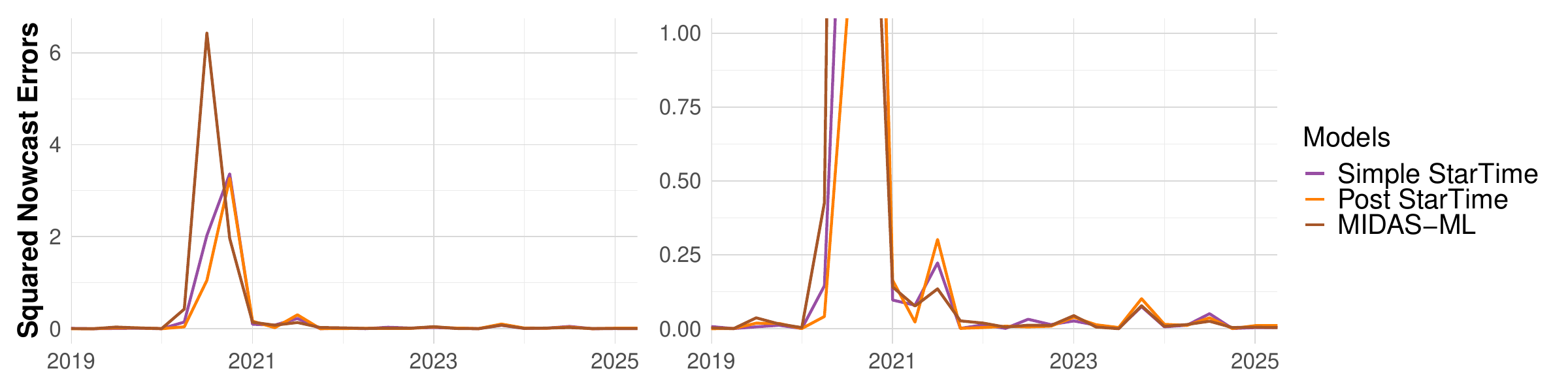}
  \caption{Evolution of squared nowcast errors for the Reduced Set (49 predictors) with $W = 105$. 
  A censored zoom  accounts for COVID-19 volatility. The plotted squared errors are multiplied by 1000 due to the small scale of GDP growth.}
  \label{fig:macro_reduced_105}
\end{figure}

Figure \ref{fig:macro_reduced_105} summarizes the evolution of the MSFEs for  the nowcasting set-up on the reduced model for $W=105$, Appendix F contains the results for the forecast set-up. 
Across all methods, the reduced model results in lower overall MSFEs than the full model. This result underwrites the value of careful, economically grounded pre-selection of relevant predictors. Detailed results on the full and reduced model are available in Figures F.1 to F.4 in Appendix F.1.

The plots show a clear spike in errors 
around 2020, corresponding to the onset of the COVID-19 pandemic. We formally define this COVID period as spanning 2020:Q1 to 2021:Q4, identified by locating where the maximum squared residual across all evaluated models exceeded five times the sample median. The Random Walk and AR(1) produce notably larger errors during the COVID period than StarTime and MIDAS-ML (Figure \ref{fig:macro_reduced_105}). Incorporating high-frequency information thus yields clear accuracy gains. 
Moreover, nowcast errors are generally smaller than forecast errors (Figure \ref{fig:macro_reduced_105} versus Figure F.1 in Appendix F.1), as expected, since nowcasts exploit additional information that becomes available within the forecast horizon.
Outside of the pandemic period,  the models  behave very similarly; differences in accuracy are small and do not persist after 2022.

We evaluate nowcast and forecast accuracy in more detail using the MSFEs reported in Table \ref{tab:macro_w105} for $W=105$. The results for $W=66$ are largely similar and available in Table F.1 of  Appendix F.2. The nowcasts ($h=0$) from Post StarTime achieve the lowest MSFE across the total out-of-sample period for the reduced model, followed by the $h=0$ results for Simple StarTime and MIDAS-ML.
Post StarTime nowcasts outperform Post StarTime forecasts by 57\%; this result is intuitive since the nowcasting model set-up explicitly leverages high-frequency inter-quarter information. A surprising result, at first sight, is that 
MIDAS-ML forecasts lead to a lower MSFE than its nowcasts. However, this is largely driven by the COVID-peak, during which MIDAS-ML seems to struggle to handle the extreme volatility. Post StarTime nowcasts are 47\% lower than the ones from MIDAS-ML during the COVID peak. Outside of the COVID-peak, MIDAS-ML nowcasts perform better than its forecasts. In fact, its nowcasts perform best and are closely followed by StarTime and Post StarTime.
Finally, the AR(1) and Random Walk are not competitive during the COVID-peak. Outside the COVID-peak, their forecast performances get considerably closer to the results of the mixed-frequency models.

\begin{table}[t]
\centering
\caption{MSFE for window $W = 105$}
\begin{threeparttable}
\begingroup
\setlength{\tabcolsep}{1.7pt}
\begin{tabular}{lcccccccc}
\toprule
& \multicolumn{2}{c}{Post StarTime}
& \multicolumn{2}{c}{Simple StarTime}
& \multicolumn{2}{c}{MIDAS-ML}
& AR(1) & Random walk \\
\cmidrule(lr){2-3} \cmidrule(lr){4-5} \cmidrule(lr){6-7}
& $h=0$ & $h=1$ & $h=0$ & $h=1$ & $h=0$ & $h=1$ &  &  \\
\midrule
\addlinespace[0.3em]
\multicolumn{9}{l}{\textit{Reduced Set}} \\
\hspace{1em}Total     & 0.198 & 0.459 & 0.239 & 0.401 & 0.365 & 0.302 & 1.336 & 1.561 \\
\hspace{1em}Non-peak  & 0.017 & 0.023 & 0.015 & 0.022 & 0.015 & 0.022 & 0.029 & 0.035 \\
\hspace{1em}COVID-peak      & 0.607 & 1.442 & 0.744 & 1.254 & 1.152 & 0.933 & 4.277 & 4.994 \\
\addlinespace[0.5em]
\multicolumn{9}{l}{\textit{Full Set}} \\
\hspace{1em}Total     & 0.332 & 0.342 & 0.695 & 0.664 & 0.749 & 0.356 & 1.336 & 1.561 \\
\hspace{1em}Non-peak  & 0.044 & 0.036 & 0.056 & 0.027 & 0.018 & 0.028 & 0.029 & 0.035 \\
\hspace{1em}COVID-peak      & 0.981 & 1.031 & 2.134 & 2.096 & 2.395 & 1.094 & 4.277 & 4.994 \\
\bottomrule
\end{tabular}
\endgroup
\begin{tablenotes}
\item \textit{Note: MSFE values are multiplied by  $1000$ due to the small scale of GDP growth.}
\end{tablenotes}
\end{threeparttable}
\label{tab:macro_w105}
\end{table}

To assess whether the differences in performance are statistically significant, we compute the MCS for $W=105$ and $W=66$ (Tables F.2 to F.5 of Appendix F.2). 
No single model dominates. We find the same conclusion when computing the MCS  for the off-peak out-of-sample period only.
The general difficulty in significantly outperforming univariate baselines for forecasting GDP growth is a well-known result. Parsimonious autoregressive benchmarks are particularly robust for stationary series like GDP growth, often proving difficult to beat (e.g., \citealp{StockWatson2007, Marcellino2006}). This aligns with our observation that especially the off-peak performance of simple models is competitive to that of the more complex models. 

Post StarTime does not only perform competitively in terms of MSFE, but uncovers economically interpretable structures through its temporal aggregation and sparsity, thereby facilitating a clearer understanding of the underlying drivers of its nowcast and forecast performance.
We analyze the structure of the coefficients estimated by Post StarTime in more detail; selection and aggregation patterns for all specifications are reported in Tables F.6 to F.14 of Appendix F.3.
We observe differences between the full and restricted model set-up, as well as  between the nowcasting  and  forecasting exercises. In the high-dimensional full forecasting model, StarTime primarily utilizes sparsity to handle the large parameter space, with almost no aggregation. In contrast, for the reduced set, it frequently exploits temporal fusion, aggregating high-frequency lags into coarser ones. 
StarTime typically yields denser coefficient vectors for nowcast than forecast specifications, reflecting that more high-frequency covariates survive among the intra-quarter information. Since StarTime performs best on the reduced model set-up, we focus our discussion on these results.

Temporal aggregation is most pronounced in the housing, financial, and manufacturing variables. In the $W=105$ window, the monthly log-levels of New Privately-Owned Housing Units Authorized ($\log{\texttt{PERMIT}}$) are fully aggregated to a single quarterly lag across
65\% of the rolling windows for the nowcasts and even
96\% for the next quarter forecasts. This finding aligns with the economic intuition that housing permits represent a pipeline of future activity; aggregating the monthly volatility reveals the persistent trend that drives quarterly GDP \citep{Cheung2016, LaPoint2024}. Similarly, the weekly Chicago Fed National Financial Conditions Index (\texttt{NFCI}) is aggregated 
91\% of the time for the $W=66$ nowcasts window, and 
69\% of the time for the forecasts.
As noted by \cite{Brave2011}, while financial stress can be volatile at high frequencies, it is the persistent tightening of conditions over the course of a quarter that materially impacts real output, supporting StarTime's choice to smooth these signals to a quarterly frequency. In the reduced nowcasting set ($W=105$), the Value of Manufacturers' New Orders (\texttt{AMTMNO}) is fully aggregated 96\% of the time. Smoothing the intrinsic volatility of new orders helps extract the underlying signal of industrial demand, reinforcing its established role as a stable leading indicator of aggregate output \citep{Zarnowitz1992}.

The sparsity patterns selected by StarTime reveal which predictors contain unique and relevant information for GDP growth, showcasing distinct lag preferences between nowcasting and forecasting. 
The first lag of quarterly GDP growth, $\Delta \log(\texttt{GDP})_{t-1}$, is selected in 69\% of the nowcasts, indicating that persistence in GDP growth provides a useful baseline for assessing current-quarter activity, consistent with standard dynamic factor nowcasting models \citep{Giannone2008}. 
The nowcasting set-up also reveals a distinct `anchoring' behavior heavily reliant on early-quarter data.
For variables like Nonfarm Payrolls (\texttt{PAYEMS}), Retail Sales (\texttt{RSAFS}), and Industrial Production (\texttt{INDPRO}), the nowcasting model predominantly selects their second lag, representing the earliest month of the quarter, over more recent releases. For example, $\texttt{PAYEMS}_{t-2}$ is retained 100\% of the time (with subsequent periods at 85\%), and similarly, $\texttt{INDPRO}_{t-2}$ is always retained (with $\texttt{INDPRO}_{t}$ and $\texttt{INDPRO}_{t-1}$ kept only 35\% of the time). This structural choice suggests that early-quarter information establishes the primary foundation for the trajectory of the quarter, while subsequent intra-quarter releases serve merely as marginal adjustments.

For the forecast model set-up with $W=105$, an interesting selection pattern appears for the labor market variables.  
We observe a structural difference between employment and jobless claims: While the model selects the first lag of Nonfarm Payrolls ($\texttt{PAYEMS}_{t-1}$) in 85\% of the cases, ignoring subsequent lags, it favours the second lag of Initial Claims ($\texttt{ICSA}_{t-2}$, 73\%) over the first (35\%). This aligns with the chronological classification of business cycle indicators established by \cite{StockWatson1989}: Payrolls act as a coincident indicator, where the most recent momentum is most predictive, whereas initial claims serve as a leading indicator, where signals often precede output changes by several months. Similarly, the model exhibits a preference for specific real-activity indicators, classified as `hard data' by \cite{Giannone2008} for their high information content. For Retail Sales ($\texttt{RSAFS}$), the first lag is consistently selected (with subsequent lags excluded), capturing the immediate impact of consumption news. Conversely, for Industrial Production ($\texttt{INDPRO}$), the first and second lags are retained (100\% and 54\%, respectively). This difference likely reflects the distinct nature of these series: While sales data is volatile and its predictive content dissipates quickly, industrial production exhibits greater persistence, leading the model to use multiple lags to capture the underlying trend.

Finally, although StarTime's high-dimensional full model-set exhibits lower accuracy than its reduced model set-up, it still offers several noteworthy insights. In particular, it also exhibits a preference for specific real-activity indicators. For the forecast model with $W=105$, 
the first lags of Industrial Production growth ($\Delta \log(\texttt{INDPRO})$) and Real Manufacturing and Trade Sales ($\Delta \log(\texttt{CMRMTSPL})$) are always retained. Housing Sales ($\texttt{HSN1F}$), Retail Sales ($\texttt{RSAFS}$), and Payrolls ($\texttt{PAYEMS}$) also all exhibit retention rates exceeding 90\%. This selection mirrors the variables identified by the \cite{NBERdating} as primary coincident indicators \citep{StockWatson1989}. The retention of housing sales underscores the critical role of the housing sector as an early warning system \citep{Leamer2007}. In the nowcasting full model set-up, additional distinct patterns emerge. The \texttt{NFCI} remains highly aggregated (89\%), confirming the robustness of persistent financial signals. Notably, Total Business Inventories (\texttt{BUSINV}) are temporally grouped 81\% of the time for $W=105$, but in terms of sparsity, its middle lag ($\texttt{BUSINV}_{t-1}$) is consistently retained more frequently than its contemporaneous and final lags across both window sizes. Inventories are heavily procyclical yet notoriously volatile and difficult to nowcast \citep{BlinderMaccini1991}. The preference for the mid-quarter inventory lag suggests the model isolates this intermediate intra-quarter data point to gauge turning points in the inventory cycle before the quarter concludes. In the most high-dimensional setting ($W=66$), the nowcasting model strictly filters information to manage the acute dimensionality: $\texttt{RSAFS}_{t-2}$ is retained 35\% while its earlier terms are ignored entirely, and similarly, only the second lag of \texttt{CMRMTSPL} is selected. This strict selection reinforces the notion that in high dimensions, early-quarter releases contain the most valuable information for concurrent output.

Overall, StarTime achieves competitive nowcast and forecast performance relative to the considered benchmarks, while uniquely offering data-driven insights into the frequency at which the predictors  are most relevant for modeling GDP growth.

\section{Conclusion} \label{sec:conclusion}
We propose a tree-structured penalized method to reduce dimensionality in high-order autoregressions and mixed-frequency regressions. StarTime flexibly selects coefficients to be temporally aggregated at varying frequencies,  sparse, or a combination of both, as supported by the data. We establish its theoretical properties by deriving an error bound under general conditions, thereby allowing for settings encountered in many econometric applications where the errors may be non-Gaussian, autocorrelated, heteroskedastic and weakly dependent. Simulations demonstrate strong performance relative to benchmarks, and empirical applications illustrate StarTime's usefulness for analyzing financial and macroeconomic data. Replication material for all analyses is available at \url{https://github.com/MarieCorillon/StarTime_Replication}.

Several avenues for future research remain. We focus on $\ell_1$-penalization to promote sparsity  and temporal aggregation in $\hat{\boldsymbol\beta}$; alternative penalties could explore other structured forms of aggregation and sparsity (e.g., \citealp{nicholson2020high, Bien17}). 
We select tuning parameters via a Bayesian Information Criterion for simplicity; extending data-driven tuning parameter approaches such as \cite{kock2025data} may further narrow the gap to oracle performance. High-dimensional inference  developed for sparse time series models (e.g., \citealp{Adamek2023}) can be extended to settings that jointly incorporate temporal aggregation and sparsity. 

\section*{Acknowledgments}

The first and third author are financially supported by a grant from the Dutch Research Council (NWO), research programme Vidi (VI.Vidi.211.032).
Previous versions of this paper were presented at ISNPS 2026, QFFE 2026, NESG 2025 and ICSDS 2025;  we gratefully acknowledge comments by the participants. 
We thank Jacob Bien, Enrico Wegner and Etienne Wijler for the helpful discussions. 

\theendnotes

\bibliographystyle{asa}
\bibliography{references}

\newpage

\appendix

\section{Theory}
\subsection{Preliminary Results}
\label{appendix:lemmas}

\begin{lemma}[Basic Inequality]
\label{lemma:basicinequality}
We have 
\begin{align*}
\frac{1}{2T}\left\lVert \boldsymbol{X}(\hat{\boldsymbol{\beta}}-\boldsymbol{\beta}^0)\right\rVert_2^2
+\lambda \left\lVert  w(\hat{\boldsymbol{\gamma}})\right\rVert_1
\leq
\frac{\boldsymbol{\varepsilon}^\top\boldsymbol{X}(\hat{\boldsymbol{\beta}}-\boldsymbol{\beta}^0)}{T}
+\lambda \left\lVert  w(\boldsymbol{\gamma}^0)\right\rVert_1.
\end{align*}
\end{lemma}

\begin{proof}[Proof of Lemma \ref{lemma:basicinequality}]
Using $\boldsymbol{y} = \boldsymbol{X} \boldsymbol{\beta}^0 + \boldsymbol{\varepsilon}$, we have
\begin{align*}
& \frac{1}{2T}\left\lVert \boldsymbol{y}-\boldsymbol{X}\hat{\boldsymbol{\beta}}\right\rVert_2^2
+\lambda \left\lVert  w(\hat{\boldsymbol{\gamma}})\right\rVert_1
\leq
\frac{1}{2T}\left\lVert \boldsymbol{y}-\boldsymbol{X}\boldsymbol{\beta}^0\right\rVert_2^2
+\lambda \left\lVert  w(\boldsymbol{\gamma}^0)\right\rVert_1 \\
\iff & \frac{1}{2T}\left\lVert \boldsymbol{X} \boldsymbol{\beta}^0 + \boldsymbol{\varepsilon} -\boldsymbol{X}\hat{\boldsymbol{\beta}}\right\rVert_2^2
+\lambda \left\lVert  w(\hat{\boldsymbol{\gamma}})\right\rVert_1
\leq
\frac{1}{2T}\left\lVert \boldsymbol{X} \boldsymbol{\beta}^0 + \boldsymbol{\varepsilon} -\boldsymbol{X}\boldsymbol{\beta}^0\right\rVert_2^2
+\lambda \left\lVert  w(\boldsymbol{\gamma}^0)\right\rVert_1 \\
\iff & \frac{1}{2T}\Big ( \left\lVert \boldsymbol{\varepsilon} \right\rVert_2^2 + 2\boldsymbol{\varepsilon}^\top \boldsymbol{X}(\hat{\boldsymbol{\beta}} - \boldsymbol{\beta}^0) + \left\lVert \boldsymbol{X}(\hat{\boldsymbol{\beta}} - \boldsymbol{\beta}^0) \right\rVert_2^2 \Big ) + \lambda \left\lVert  w(\hat{\boldsymbol{\gamma}})\right\rVert_1
\leq \frac{1}{2T} \left\lVert \boldsymbol{\varepsilon} \right\rVert_2^2 + \lambda \left\lVert  w(\boldsymbol{\gamma}^0)\right\rVert_1 \\
\iff & \frac{1}{2T}\left\lVert \boldsymbol{X}(\hat{\boldsymbol{\beta}}-\boldsymbol{\beta}^0)\right\rVert_2^2
+\lambda \left\lVert  w(\hat{\boldsymbol{\gamma}})\right\rVert_1
\leq
\frac{\boldsymbol{\varepsilon}^\top\boldsymbol{X}(\hat{\boldsymbol{\beta}}-\boldsymbol{\beta}^0)}{T}
+\lambda \left\lVert  w(\boldsymbol{\gamma}^0)\right\rVert_1 \qedhere
\end{align*}
\end{proof}

\begin{lemma}[Concentration]
\label{lemma:concentration}
Let Assumption 1 be satisfied and recall $\mathcal{J} = \left\{ \frac{\|\boldsymbol{X}^\top \boldsymbol{\varepsilon}\|_\infty}{T} \leq \lambda_0 \right\}
$. Then, for $\lambda_0 := C\,\frac{N^{1/\tilde{m}}(\ln{\ln{T}})^{1/\tilde{m}}}{\sqrt{T}}$,
\begin{align*}
\mathbb{P}(\mathcal{J}) \ge 1- C(\ln\ln T)^{-1}.
\end{align*}
\end{lemma}

\begin{proof}[Proof of Lemma \ref{lemma:concentration}]
Let $V_j := \sum_{t=1}^{T} x_{j,t}\varepsilon_t$ and $C_1 > 0$. Then, the complement of $\mathcal{J}$, $\mathcal{J}^c$, can be expressed as
\begin{equation*}
    \mathbb{P}(\mathcal{J}^c)
= \mathbb{P}\Big( \max_{1\le j\le N} |V_j| > T\lambda_0 \Big)
\le C_1 N\Big(\frac{1}{\lambda_0\sqrt{T}}\Big)^{\tilde{m}},
\end{equation*}
where the inequality follows from \citet{Adamek2023}'s Lemma A.4.  
Let $L_T\to\infty$ as $T \to \infty$ be any deterministic sequence and choose
\begin{equation*}
\lambda_0 := C_2\,\frac{N^{1/\tilde{m}}L_T^{1/\tilde{m}}}{\sqrt{T}},
\end{equation*}
for some $C_2>0$. Then
\begin{equation*}
\mathbb{P}(\mathcal{J}^c)
\le C_1 N\Big(\frac{1}{C_2 N^{1/\tilde{m}}L_T^{1/\tilde{m}}}\Big)^{\tilde{m}}
= \frac{C}{L_T},
\end{equation*}
and
\begin{equation*}
\mathbb{P}(\mathcal{J}) \ge 1-\frac{C}{L_T}.
\end{equation*}
In particular, for $L_T=\ln\ln T$, we obtain $\mathbb{P}(\mathcal{J}) \ge 1- C(\ln\ln T)^{-1}\to 1$ as $T\to\infty$.
\end{proof}

\begin{lemma}[Initial Error Bound]
\label{lemma:concentration2}
Under Assumption 1 and for $C > 0$, let the tuning parameter be
\begin{align*}
\lambda := 2 C\,\frac{N^{1/\tilde{m}}(\ln{\ln{T}})^{1/\tilde{m}}}{\sqrt{T}}.
\end{align*}
With probability at least $1 - \frac{C}{\ln{\ln{T}}}$,
we have
\begin{align*}
\frac{1}{2T} \left \lVert \boldsymbol{X}(\hat{\boldsymbol{\beta}} - \boldsymbol{\beta}^0)\right \rVert_2^2 \leq \frac{3}{2}\lambda \left \lVert w(\boldsymbol{\gamma}^0) \right \rVert_1.
\end{align*}
\end{lemma}

\begin{proof}[Proof of Lemma \ref{lemma:concentration2}]
    On the event
\begin{equation*}
\mathcal{J}
:=
\left\{
\frac{\left\lVert \boldsymbol{X}^\top\boldsymbol{\varepsilon}\right\rVert_{\infty}}{T}
\le \lambda_0
\right\},
\end{equation*} Lemma \ref{lemma:basicinequality} yields
\begin{align*}
\frac{1}{2T}\left\lVert \boldsymbol{X}(\hat{\boldsymbol{\beta}}-\boldsymbol{\beta}^0)\right\rVert_2^2
+\lambda \left\lVert  w(\hat{\boldsymbol{\gamma}})\right\rVert_1
\le
\frac{\boldsymbol{\varepsilon}^\top\boldsymbol{X}(\hat{\boldsymbol{\beta}}-\boldsymbol{\beta}^0)}{T}
+\lambda \left\lVert  w(\boldsymbol{\gamma}^0)\right\rVert_1.
\end{align*}
Using H\"older's inequality and the definition of $\mathcal{J}$, we have
\begin{align*}
\left|
\frac{\boldsymbol{\varepsilon}^\top\boldsymbol{X}(\hat{\boldsymbol{\beta}}-\boldsymbol{\beta}^0)}{T}
\right|
&=
\left|
\frac{(\boldsymbol{X}^\top\boldsymbol{\varepsilon})^\top(\hat{\boldsymbol{\beta}}-\boldsymbol{\beta}^0)}{T}
\right|
\le
\frac{\left\lVert \boldsymbol{X}^\top\boldsymbol{\varepsilon}\right\rVert_{\infty}}{T}
\left\lVert \hat{\boldsymbol{\beta}}-\boldsymbol{\beta}^0\right\rVert_1
\le
\lambda_0 \left\lVert \hat{\boldsymbol{\beta}}-\boldsymbol{\beta}^0\right\rVert_1 .
\end{align*}
Assume now that $\lambda \ge 2\lambda_0$, then
\begin{equation*}
\frac{\boldsymbol{\varepsilon}^\top\boldsymbol{X}(\hat{\boldsymbol{\beta}}-\boldsymbol{\beta}^0)}{T}
\le
\left|
\frac{\boldsymbol{\varepsilon}^\top\boldsymbol{X}(\hat{\boldsymbol{\beta}}-\boldsymbol{\beta}^0)}{T}
\right|
\le
\frac{\lambda}{2}\left\lVert \hat{\boldsymbol{\beta}}-\boldsymbol{\beta}^0\right\rVert_1 .
\end{equation*}
Substituting this bound into the basic inequality gives
\begin{align*}
\frac{1}{2T}\left\lVert \boldsymbol{X}(\hat{\boldsymbol{\beta}}-\boldsymbol{\beta}^0)\right\rVert_2^2
+\lambda \left\lVert  w(\hat{\boldsymbol{\gamma}})\right\rVert_1
\le
\frac{\lambda}{2}\left\lVert \hat{\boldsymbol{\beta}}-\boldsymbol{\beta}^0\right\rVert_1
+\lambda \left\lVert  w(\boldsymbol{\gamma}^0)\right\rVert_1.
\end{align*}
Finally, using $\left\lVert \hat{\boldsymbol{\beta}}-\boldsymbol{\beta}^0\right\rVert_1
\le \left\lVert  w(\hat{\boldsymbol{\gamma}})- w(\boldsymbol{\gamma}^0)\right\rVert_1$
and the triangle inequality
\begin{equation*}
\left\lVert  w(\hat{\boldsymbol{\gamma}})- w(\boldsymbol{\gamma}^0)\right\rVert_1
\le
\left\lVert  w(\hat{\boldsymbol{\gamma}})\right\rVert_1
+
\left\lVert  w(\boldsymbol{\gamma}^0)\right\rVert_1 ,
\end{equation*}
we obtain
\begin{align*}
\frac{1}{2T}\left\lVert \boldsymbol{X}(\hat{\boldsymbol{\beta}}-\boldsymbol{\beta}^0)\right\rVert_2^2
+\lambda \left\lVert  w(\hat{\boldsymbol{\gamma}})\right\rVert_1
&\le
\frac{\lambda}{2}\Big(
\left\lVert  w(\hat{\boldsymbol{\gamma}})\right\rVert_1
+
\left\lVert  w(\boldsymbol{\gamma}^0)\right\rVert_1
\Big)
+\lambda \left\lVert  w(\boldsymbol{\gamma}^0)\right\rVert_1 \\
&=
\frac{\lambda}{2}\left\lVert  w(\hat{\boldsymbol{\gamma}})\right\rVert_1
+\frac{3\lambda}{2}\left\lVert  w(\boldsymbol{\gamma}^0)\right\rVert_1 .
\end{align*}
Subtract $\frac{\lambda}{2}\left\lVert  w(\hat{\boldsymbol{\gamma}})\right\rVert_1$ from both sides to obtain
\begin{align*}
\frac{1}{2T}\left\lVert \boldsymbol{X}(\hat{\boldsymbol{\beta}}-\boldsymbol{\beta}^0)\right\rVert_2^2
+\frac{\lambda}{2}\left\lVert  w(\hat{\boldsymbol{\gamma}})\right\rVert_1
\le
\frac{3\lambda}{2}\left\lVert  w(\boldsymbol{\gamma}^0)\right\rVert_1.
\end{align*}
Since $\frac{\lambda}{2}\left\lVert  w(\hat{\boldsymbol{\gamma}})\right\rVert_1 \ge 0$, we can drop this non-negative term from the left-hand side, yielding 
\begin{align*}
\frac{1}{2T}\left\lVert \boldsymbol{X}(\hat{\boldsymbol{\beta}}-\boldsymbol{\beta}^0)\right\rVert_2^2
\le
\frac{3\lambda}{2}\left\lVert  w(\boldsymbol{\gamma}^0)\right\rVert_1.
\end{align*}
The associated probability stems from Lemma \ref{lemma:concentration}.
\end{proof}

\begin{lemma}[Cone Decomposition]
\label{lemma:cone}
On the event $\mathcal{J}$ of Lemma \ref{lemma:concentration} and under the conditions of
Lemma \ref{lemma:basicinequality} and Lemma \ref{lemma:concentration2},
the StarTime estimator $\hat{\boldsymbol{\beta}}$ satisfies
\begin{align*}
\frac{1}{T}\left\lVert \boldsymbol{X}(\hat{\boldsymbol{\beta}}-\boldsymbol{\beta}^0)\right\rVert_2^2
+\lambda \left\lVert  w(\hat{\boldsymbol{\gamma}})_{S_0^c}\right\rVert_1
\le
3\lambda \left\lVert
 w(\hat{\boldsymbol{\gamma}})_{S_0}- w(\boldsymbol{\gamma}^0)_{S_0}
\right\rVert_1.
\end{align*}
\end{lemma}

\begin{proof}[Proof of Lemma \ref{lemma:cone}]
Starting from Lemma \ref{lemma:basicinequality} multiplied by a factor of $2$, we have
\begin{equation}
\label{eq:lemma3-start}
\frac{1}{T}\left\lVert \boldsymbol{X}(\hat{\boldsymbol{\beta}}-\boldsymbol{\beta}^0)\right\rVert_2^2
+2\lambda \left\lVert  w(\hat{\boldsymbol{\gamma}})\right\rVert_1
\le
\frac{2}{T}\boldsymbol{\varepsilon}^\top\boldsymbol{X}(\hat{\boldsymbol{\beta}}-\boldsymbol{\beta}^0)
+2\lambda \left\lVert  w(\boldsymbol{\gamma}^0)\right\rVert_1 .
\end{equation}
On the event $\mathcal{J}$, the empirical-process bound gives
\begin{equation*}
\frac{1}{T}\left|\boldsymbol{\varepsilon}^\top\boldsymbol{X}(\hat{\boldsymbol{\beta}}-\boldsymbol{\beta}^0)\right|
\le
\lambda_0\left\lVert \hat{\boldsymbol{\beta}}-\boldsymbol{\beta}^0\right\rVert_1.
\end{equation*}
Substituting this bound into \eqref{eq:lemma3-start} and using $\lambda\ge 2\lambda_0$ yields
\begin{equation}
\label{eq:lemma3-basic}
\frac{1}{T}\left\lVert \boldsymbol{X}(\hat{\boldsymbol{\beta}}-\boldsymbol{\beta}^0)\right\rVert_2^2
+2\lambda \left\lVert  w(\hat{\boldsymbol{\gamma}})\right\rVert_1
\le
\lambda \left\lVert \hat{\boldsymbol{\beta}}-\boldsymbol{\beta}^0\right\rVert_1
+2\lambda \left\lVert  w(\boldsymbol{\gamma}^0)\right\rVert_1 .
\end{equation}
Since $w(\boldsymbol{\gamma}^0)_{S_0^c}=\bf 0$, we can write
\begin{equation*}
\left\lVert  w(\hat{\boldsymbol{\gamma}})\right\rVert_1
=
\left\lVert  w(\hat{\boldsymbol{\gamma}})_{S_0}\right\rVert_1
+
\left\lVert  w(\hat{\boldsymbol{\gamma}})_{S_0^c}\right\rVert_1,
\qquad
\left\lVert  w(\boldsymbol{\gamma}^0)\right\rVert_1
=
\left\lVert  w(\boldsymbol{\gamma}^0)_{S_0}\right\rVert_1.
\end{equation*}
We first focus on the penalty term on the left-hand side of \eqref{eq:lemma3-basic}.
Using the triangle inequality in the form $\left\lVert a\right\rVert_1 \ge \left\lVert b\right\rVert_1 - \left\lVert a-b\right\rVert_1$, we have
\begin{align*}
\left\lVert  w(\hat{\boldsymbol{\gamma}})_{S_0}\right\rVert_1
&=
\left\lVert
 w(\boldsymbol{\gamma}^0)_{S_0}
+
\Big( w(\hat{\boldsymbol{\gamma}})_{S_0}- w(\boldsymbol{\gamma}^0)_{S_0}\Big)
\right\rVert_1 \\
&\ge
\left\lVert  w(\boldsymbol{\gamma}^0)_{S_0}\right\rVert_1
-
\left\lVert  w(\hat{\boldsymbol{\gamma}})_{S_0}- w(\boldsymbol{\gamma}^0)_{S_0}\right\rVert_1.
\end{align*}
Substituting this lower bound into the left-hand side of \eqref{eq:lemma3-basic} gives
\begin{align*}
\frac{1}{T}\left\lVert \boldsymbol{X}(\hat{\boldsymbol{\beta}}-\boldsymbol{\beta}^0)\right\rVert_2^2
&\,+2\lambda \left\lVert  w(\hat{\boldsymbol{\gamma}})_{S_0}\right\rVert_1
+2\lambda \left\lVert  w(\hat{\boldsymbol{\gamma}})_{S_0^c}\right\rVert_1 \\
\ge\
\frac{1}{T}\left\lVert \boldsymbol{X}(\hat{\boldsymbol{\beta}}-\boldsymbol{\beta}^0)\right\rVert_2^2
&\,+2\lambda \left\lVert  w(\boldsymbol{\gamma}^0)_{S_0}\right\rVert_1
-2\lambda \left\lVert  w(\hat{\boldsymbol{\gamma}})_{S_0}- w(\boldsymbol{\gamma}^0)_{S_0}\right\rVert_1 +2\lambda \left\lVert  w(\hat{\boldsymbol{\gamma}})_{S_0^c}\right\rVert_1.
\end{align*}
Now we turn to the term $\left\lVert \hat{\boldsymbol{\beta}}-\boldsymbol{\beta}^0\right\rVert_1$ on the right-hand side of \eqref{eq:lemma3-basic}.
We use that
\begin{align*}
\left\lVert \hat{\boldsymbol{\beta}}-\boldsymbol{\beta}^0\right\rVert_1
&\le
\left\lVert \hat{\boldsymbol{\gamma}}-\boldsymbol{\gamma}^0\right\rVert_1
+\left\lVert \boldsymbol{A}\left(\hat{\boldsymbol{\gamma}}-\boldsymbol{\gamma}^0\right)\right\rVert_1 \\
&=
\left\lVert  w(\hat{\boldsymbol{\gamma}})- w(\boldsymbol{\gamma}^0)\right\rVert_1 \\
&=
\left\lVert  w(\hat{\boldsymbol{\gamma}})_{S_0}- w(\boldsymbol{\gamma}^0)_{S_0}\right\rVert_1
+\left\lVert  w(\hat{\boldsymbol{\gamma}})_{S_0^c}\right\rVert_1.
\end{align*}
Therefore, the right-hand side of \eqref{eq:lemma3-basic} is bounded above by
\begin{equation*}
\lambda\left\lVert  w(\hat{\boldsymbol{\gamma}})_{S_0}- w(\boldsymbol{\gamma}^0)_{S_0}\right\rVert_1
+\lambda\left\lVert  w(\hat{\boldsymbol{\gamma}})_{S_0^c}\right\rVert_1
+2\lambda\left\lVert  w(\boldsymbol{\gamma}^0)_{S_0}\right\rVert_1.
\end{equation*}
Combining the preceding two displays with \eqref{eq:lemma3-basic} and rearranging terms yields
\begin{align*}
&\frac{1}{T}\left\lVert \boldsymbol{X}(\hat{\boldsymbol{\beta}}-\boldsymbol{\beta}^0)\right\rVert_2^2
+\lambda \left\lVert  w(\hat{\boldsymbol{\gamma}})_{S_0^c}\right\rVert_1
\le
3\lambda \left\lVert  w(\hat{\boldsymbol{\gamma}})_{S_0}- w(\boldsymbol{\gamma}^0)_{S_0}\right\rVert_1. \qedhere
\end{align*}
\end{proof}

\begin{lemma}[Active Set Bounding]
\label{lemma:cone_implication} 
Under Assumption 2 and for every $\boldsymbol{\gamma}$ on the cone, that is, $ \Big \{\boldsymbol{\gamma} \in \mathbb{R}^{M}: \left \lVert w(\boldsymbol{\gamma})_{S_0^c}\right \rVert_1 \leq
3\left \lVert w(\boldsymbol{\gamma})_{S_0}\right \rVert_1 \Big \}$, we have 
\begin{equation}
    \left \lVert \boldsymbol{\beta} \right \rVert_1 \leq \left \lVert w(\boldsymbol{\gamma}) \right \rVert_1 \leq 4 \left \lVert w(\boldsymbol{\gamma})_{S_0} \right \rVert_1.
\end{equation}
\end{lemma}

\begin{proof}[Proof of Lemma \ref{lemma:cone_implication}]
    One the cone, 
    \begin{equation*}
    \left \lVert w(\boldsymbol{\gamma}) \right \rVert_1 = \left \lVert w(\boldsymbol{\gamma})_{S_0^c} \right \rVert_1 + \left \lVert w(\boldsymbol{\gamma})_{S_0} \right \rVert_1 \leq 4 \left \lVert w(\boldsymbol{\gamma})_{S_0} \right \rVert_1.
\end{equation*}
The initial cone condition is a direct result from Lemma \ref{lemma:cone}.
\end{proof}

\begin{lemma}[Sample Compatibility]
\label{lemma:pop_to_sample_compatibility} 
For every $\boldsymbol{\gamma}$ on the cone, that is, $ \Big \{\boldsymbol{\gamma} \in \mathbb{R}^{M}: \left \lVert w(\boldsymbol{\gamma})_{S_0^c}\right \rVert_1 \leq
3\left \lVert w(\boldsymbol{\gamma})_{S_0}\right \rVert_1 \Big \}$, and under Assumptions 1, 2 and 3, we have the following sample compatibility condition on the set $\mathcal{C}(S_0) := \Bigl\{ \left \lVert \boldsymbol{\hat{\Sigma}} - \boldsymbol{\Sigma} \right \rVert_{\infty} \leq C/s_0 \Bigr \}$,
\begin{equation*} s_0\boldsymbol{\beta}^\top\boldsymbol{\hat{\Sigma}}\boldsymbol{\beta} \geq \rho_0^2 \left \lVert w(\boldsymbol{\gamma})_{S_0} \right \rVert_1^2,
\end{equation*}
where $\rho_0 > 0$.
\end{lemma}

\begin{proof}[Proof of Lemma \ref{lemma:pop_to_sample_compatibility}]
    Let $\boldsymbol{\gamma} \in \mathbb{R}^{M}$ satisfy $\left \lVert w(\boldsymbol{\gamma})_{S_0^c}\right \rVert_1 \leq
3\left \lVert w(\boldsymbol{\gamma})_{S_0}\right \rVert_1$. Observe that, on the set $\mathcal{C}(S_0) := \Bigl\{ \left \lVert \boldsymbol{\hat{\Sigma}} - \boldsymbol{\Sigma} \right \rVert_{\infty} \leq C/s_0 \Bigr \}$ and using Lemma \ref{lemma:cone_implication} for the third inequality,
\begin{align*}
    \Big |\boldsymbol{\beta}^\top(\boldsymbol{\hat{\Sigma}} - \boldsymbol{\Sigma})\boldsymbol{\beta} \Big | &\leq \left \lVert \boldsymbol{\hat{\Sigma}} - \boldsymbol{\Sigma} \right \rVert_{\infty} \left \lVert \boldsymbol{\beta} \right \rVert_1^2 \leq \frac{C}{s_0} \left \lVert \boldsymbol{\beta} \right \rVert_1^2 \leq \frac{C}{s_0} \left \lVert \boldsymbol{w(\boldsymbol\gamma)} \right \rVert_1^2\\
    & \leq \frac{16C}{s_0} \left \lVert w(\boldsymbol{\gamma})_{S_0} \right \rVert_1^2.
\end{align*}
So, on $\mathcal{C}(S_0)$, we have $\boldsymbol{\beta}^\top(\boldsymbol{\hat{\Sigma}} - \boldsymbol{\Sigma})\boldsymbol{\beta} \geq - \frac{16C}{s_0} \left \lVert w(\boldsymbol{\gamma})_{S_0} \right \rVert_1^2$.
Then, under Assumption 3,
\begin{align*}
s_0\boldsymbol{\beta}^\top\boldsymbol{\hat{\Sigma}}\boldsymbol{\beta} &= s_0 \boldsymbol{\beta}^\top\boldsymbol{\Sigma}\boldsymbol{\beta} + s_0\boldsymbol{\beta}^\top(\boldsymbol{\hat{\Sigma}} - \boldsymbol{\Sigma})\boldsymbol{\beta} \\
& \geq \rho_{\boldsymbol\Sigma}^2 \left \lVert w(\boldsymbol{\gamma})_{S_0}\right \rVert_1^2 - 16C \left \lVert w(\boldsymbol{\gamma})_{S_0} \right \rVert_1^2 \\
& = \Big (\rho_{\boldsymbol\Sigma}^2 - 16C \Big ) \left \lVert w(\boldsymbol{\gamma})_{S_0} \right \rVert_1^2
\end{align*}
with $\frac{\rho_{\boldsymbol\Sigma}^2}{2} \geq 16C$, so $\Big (\rho_{\boldsymbol\Sigma}^2 - 16C \Big ) \geq \frac{\rho_{\boldsymbol\Sigma}^2}{2}$.
Taking $\rho_0^2 = \frac{\rho_{\boldsymbol\Sigma}^2}{2}$ leads to the sample compatibility condition 
\begin{equation}
s_0\boldsymbol{\beta}^\top\boldsymbol{\hat{\Sigma}}\boldsymbol{\beta} \geq \rho_0^2 \left \lVert w(\boldsymbol{\gamma})_{S_0} \right \rVert_1^2. \label{eq:sample_compatibility}
\end{equation}
Using Lemma A.3. in \citet{Adamek2023} for the exact sparsity case ($r=0$), we have that for a sequence $\eta_T \to 0$ with $\eta_T \leq \frac{N^2}{e}$, and such that $s_0 \leq C \eta_T^{\frac{d+\tilde{m}-1}{d\tilde{m}+\tilde{m}-1}} \Big [ \frac{\sqrt{T}}{N^{\frac{2}{d}+\frac{2}{\tilde{m}-1}}} \Big ]^{\frac{1}{\frac{1}{d}+\frac{\tilde{m}}{\tilde{m}-1}}}$, then $\mathbb{P}(\mathcal{C}(S_0)) \geq 1 - \eta_T \to 1 \text{ as } N,T \to \infty$. Hence, the bound in \eqref{eq:sample_compatibility} holds with the same probability.
\end{proof}

\subsection{Proofs of the Main Results}
\label{appendix:proofs}

\begin{proof}[Proof of Theorem 1]
On the events $\mathcal{J}$ and $\mathcal{C}(S_0)$, and for $\lambda \ge 2\lambda_0$, we have
\begin{align*}
\frac{2}{2T}\left\lVert \boldsymbol{X}(\hat{\boldsymbol{\beta}}-\boldsymbol{\beta}^0)\right\rVert_2^2
+\lambda\left\lVert w(\hat{\boldsymbol{\gamma}}-\boldsymbol{\gamma}^0)\right\rVert_1
&=
\frac{2}{2T}\left\lVert \boldsymbol{X}(\hat{\boldsymbol{\beta}}-\boldsymbol{\beta}^0)\right\rVert_2^2
+\lambda\left\lVert \big(w(\hat{\boldsymbol{\gamma}})-w(\boldsymbol{\gamma}^0)\big)_{S_0}\right\rVert_1 \\
&+\lambda\left\lVert w(\hat{\boldsymbol{\gamma}})_{S_0^c}\right\rVert_1 \\
&\le
3\lambda\left\lVert \big(w(\hat{\boldsymbol{\gamma}})-w(\boldsymbol{\gamma}^0)\big)_{S_0}\right\rVert_1
+\lambda\left\lVert \big(w(\hat{\boldsymbol{\gamma}})-w(\boldsymbol{\gamma}^0)\big)_{S_0}\right\rVert_1 \\
&\le
4\lambda \frac{\sqrt{s_0}}{\rho_0}\,
\frac{\left\lVert \boldsymbol{X}(\hat{\boldsymbol{\beta}}-\boldsymbol{\beta}^0)\right\rVert_2}{\sqrt{T}} \\
&\le
\frac{\left\lVert \boldsymbol{X}(\hat{\boldsymbol{\beta}}-\boldsymbol{\beta}^0)\right\rVert_2^2}{2T}
+
\frac{8\lambda^2 s_0}{\rho_0^2}.
\end{align*}
The first inequality uses Lemma \ref{lemma:cone}, the second inequality uses Lemma
\ref{lemma:pop_to_sample_compatibility}, and the last inequality uses the identity $4uv \le u^2 + 4v^2$.
Subtracting $\frac{1}{2T}\left\lVert \boldsymbol{X}(\hat{\boldsymbol{\beta}}-\boldsymbol{\beta}^0)\right\rVert_2^2$ from both sides yields
\begin{equation*}
\frac{1}{2T}\left\lVert \boldsymbol{X}(\hat{\boldsymbol{\beta}}-\boldsymbol{\beta}^0)\right\rVert_2^2
+\lambda\left\lVert w(\hat{\boldsymbol{\gamma}}-\boldsymbol{\gamma}^0)\right\rVert_1
\le
\frac{8\lambda^2 s_0}{\rho_0^2}.
\end{equation*}

From Lemma \ref{lemma:concentration}, we know that  $\mathbb{P}(\mathcal{J}) \geq 1 - C(\ln{\ln{T}})^{-1}$ for $\lambda := 2 C\,\frac{N^{1/\tilde{m}}(\ln{\ln{T}})^{1/\tilde{m}}}{\sqrt{T}}$. 
From Lemma \ref{lemma:pop_to_sample_compatibility}, we have that $\mathbb{P}(\mathcal{C}(S_0)) \geq 1 - \eta_T \to 1$ for $\eta_T \leq \frac{N^2}{e}$ and if $s_0 \leq C \eta_T^{\frac{d+\tilde{m}-1}{d\tilde{m}+\tilde{m}-1}} \Big [ \frac{\sqrt{T}}{N^{\frac{2}{d}+\frac{2}{\tilde{m}-1}}} \Big ]^{\frac{1}{\frac{1}{d}+\frac{\tilde{m}}{\tilde{m}-1}}}$ as $N,T \to \infty$. Choosing $\eta_T = (\ln{\ln{T}})^{-1}$ satisfies those conditions. Hence, $\mathbb{P}(\mathcal{J} \cap \mathcal{C}(S_0)) = 1 - (1 - \mathbb{P}(\mathcal{J})) - (1 - \mathbb{P}(\mathcal{C}(S_0))) \geq 1 - C_1(\ln{\ln{T}})^{-1} - C_2(\ln{\ln{T}})^{-1} = 1 - C(\ln{\ln{T}})^{-1}$.
\end{proof}

\begin{proof}[Proof of Corollary 1]
    It follows from Theorem 1 that if $\frac{1}{2T}\left\lVert \boldsymbol{X}(\hat{\boldsymbol{\beta}}-\boldsymbol{\beta}^0)\right\rVert_2^2
+\lambda\left\lVert w(\hat{\boldsymbol{\gamma}}-\boldsymbol{\gamma}^0)\right\rVert_1
\le
\frac{8\lambda^2 s_0}{\rho_0^2}$, then each component on the left-hand side of the inequality is smaller than $8\lambda^2 s_0/\rho_0^2$, with at least the same probability.
\end{proof}


\section{Additional Simulation Results} \label{app:sims}

\begin{figure}[H]
  \centering

  \begin{subfigure}[t]{0.45\textwidth}
    \centering
    \includegraphics[width=\textwidth]{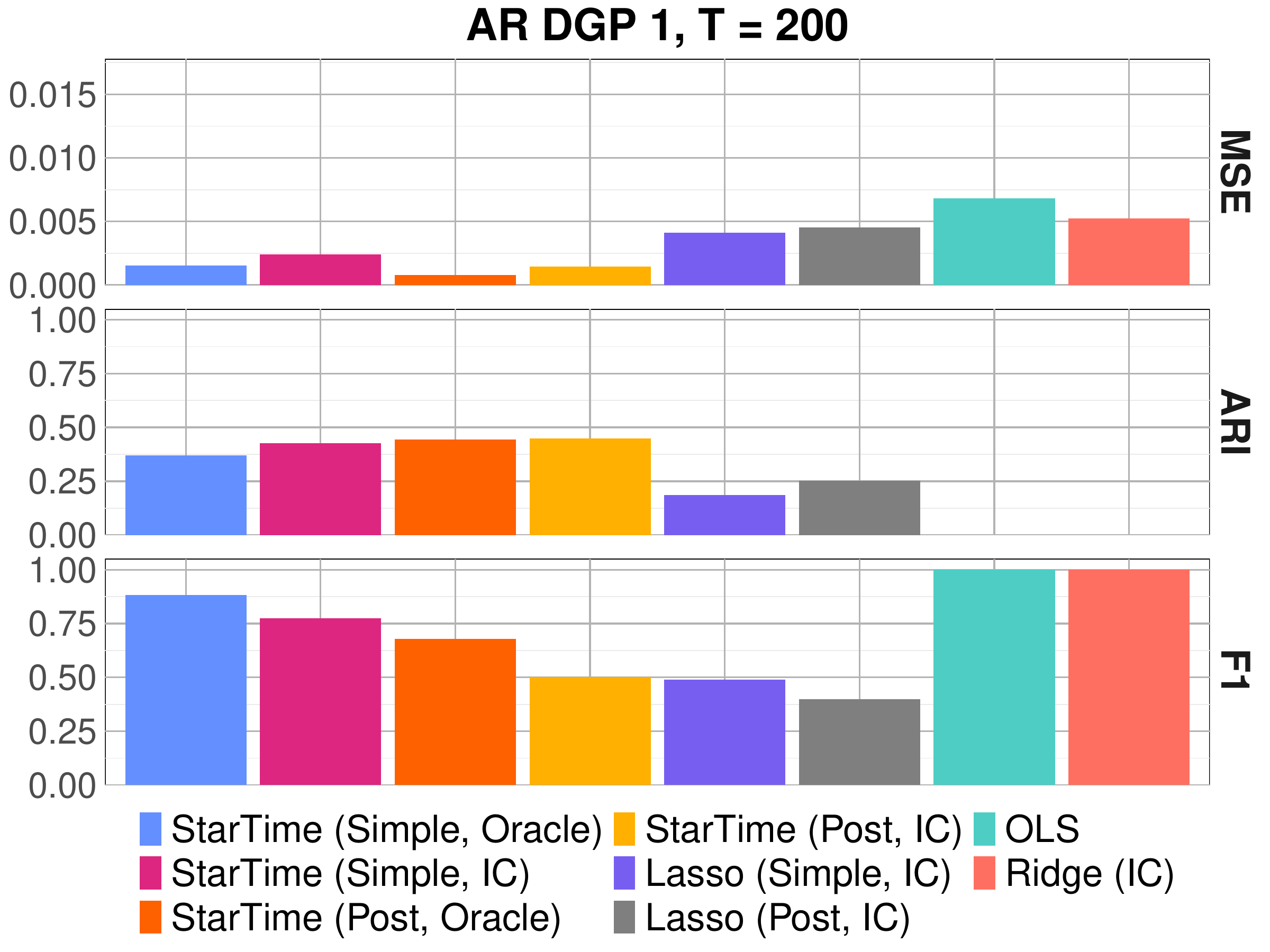}
    \caption{DGP 1, $T=200$}
    \label{fig:ARdgp1_n200}
  \end{subfigure}
  \begin{subfigure}[t]{0.45\textwidth}
    \centering
    \includegraphics[width=\textwidth]{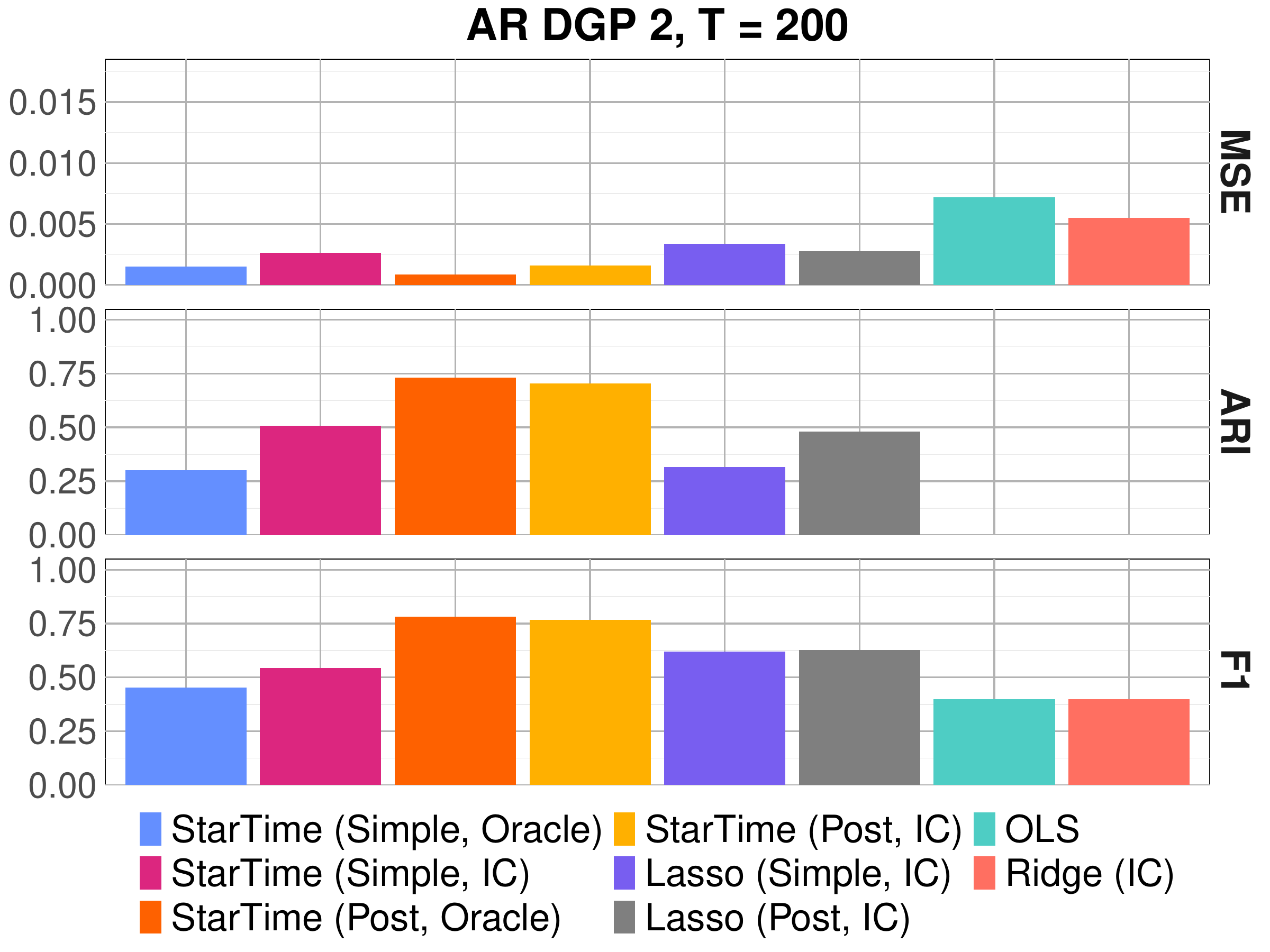}
    \caption{DGP 2, $T=200$}
    \label{fig:ARdgp2_n200}
  \end{subfigure}

  \begin{subfigure}[t]{0.45\textwidth}
    \centering
    \includegraphics[width=\textwidth]{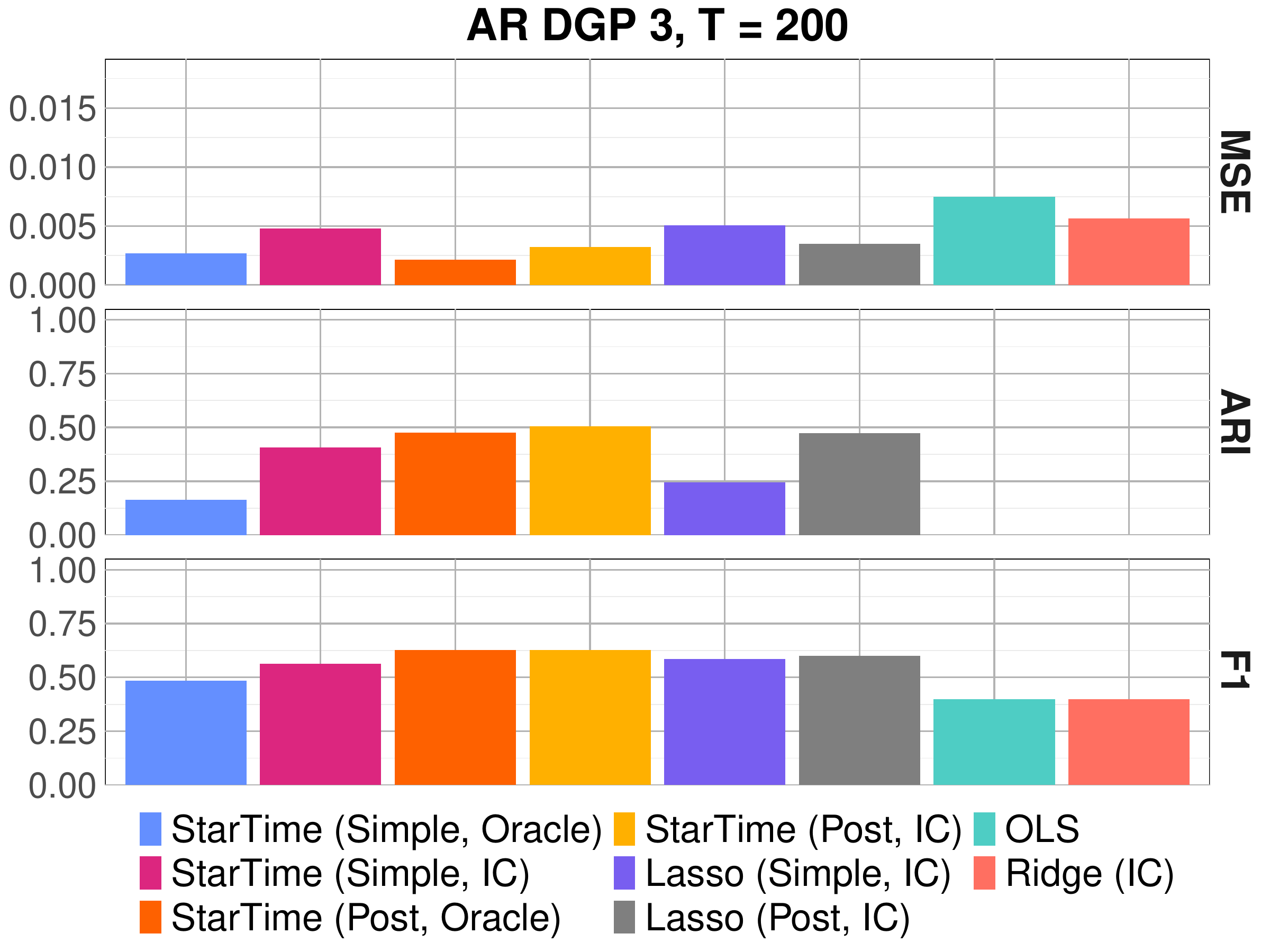}
    \caption{DGP 3, $T=200$}
    \label{fig:ARdgp3_n200}
  \end{subfigure}
  \caption{Performance metrics of the estimators across the three AR-based DGPs. 
  }
  \label{fig:simulation_results_AR_T200}
\end{figure}

  \begin{figure}
  \centering
  \begin{subfigure}[t]{0.45\textwidth}
    \centering
    \includegraphics[width=\textwidth]{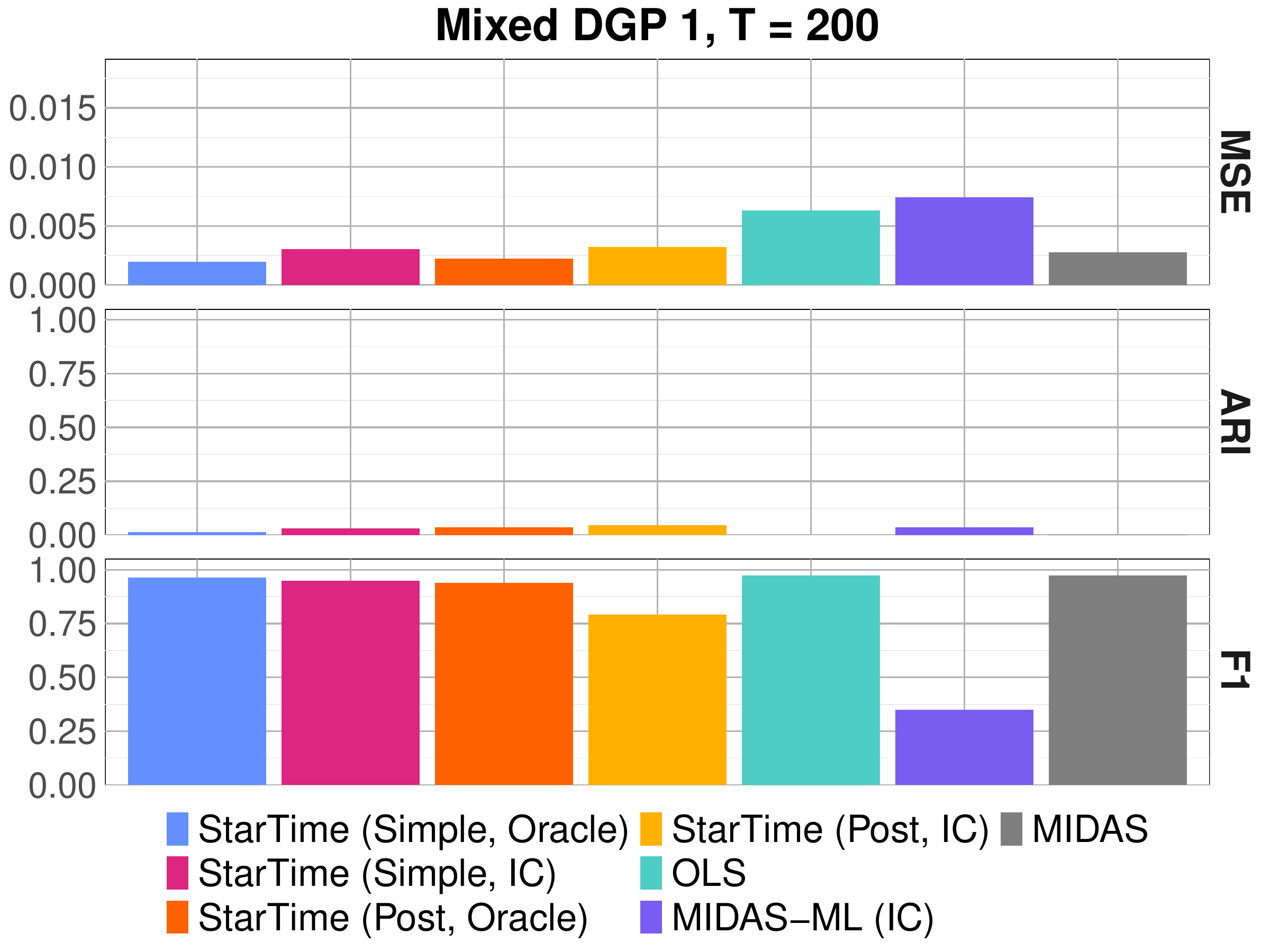}
    \caption{DGP 1, $T=200$}
    \label{fig:Mixeddgp1_n200}
  \end{subfigure}

  \begin{subfigure}[t]{0.45\textwidth}
    \centering
    \includegraphics[width=\textwidth]{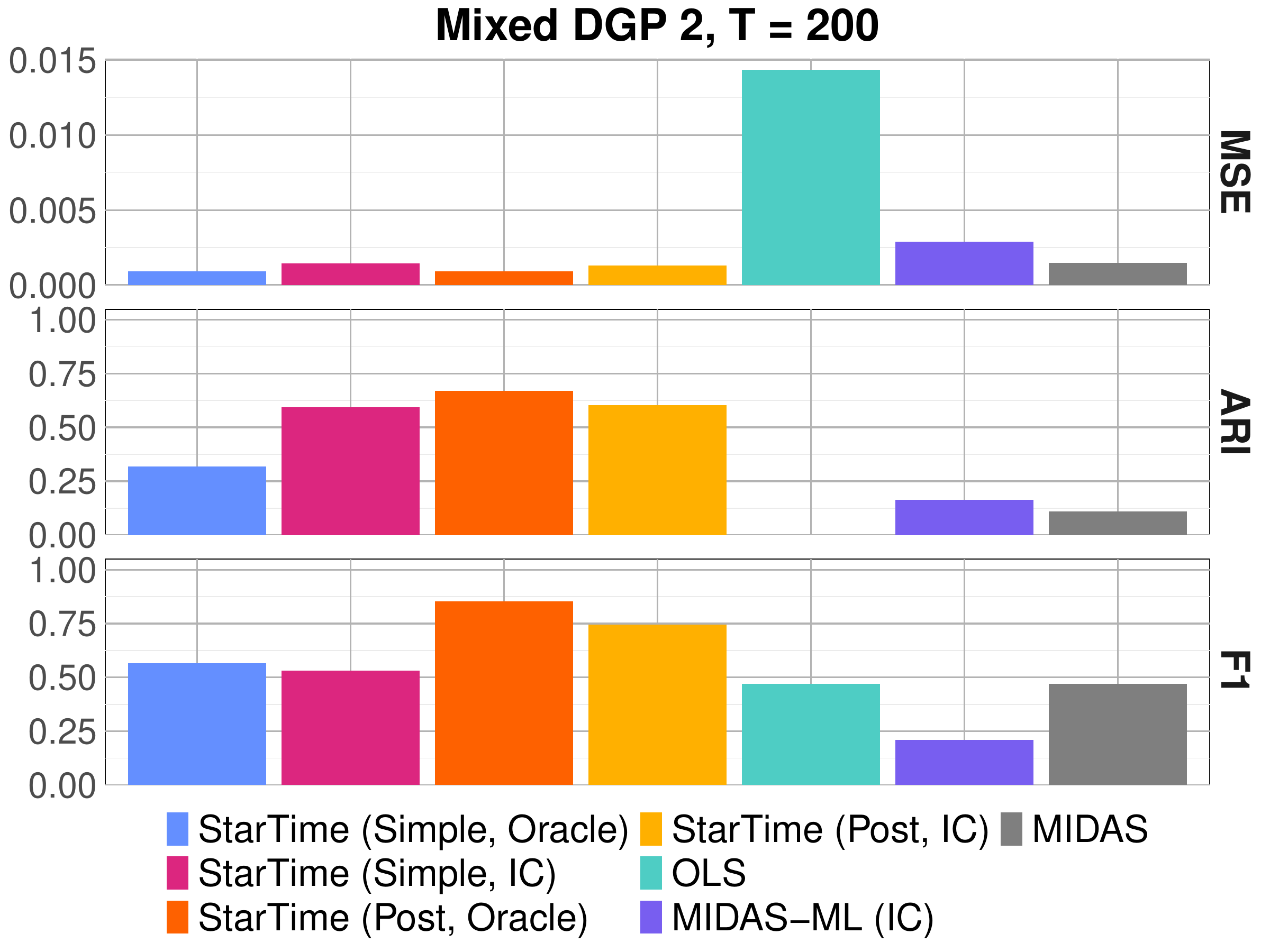}
    \caption{DGP 2, $T=200$}
    \label{fig:Mixeddgp2_n200}
  \end{subfigure} 

  \begin{subfigure}[t]{0.45\textwidth}
    \centering
    \includegraphics[width=\textwidth]{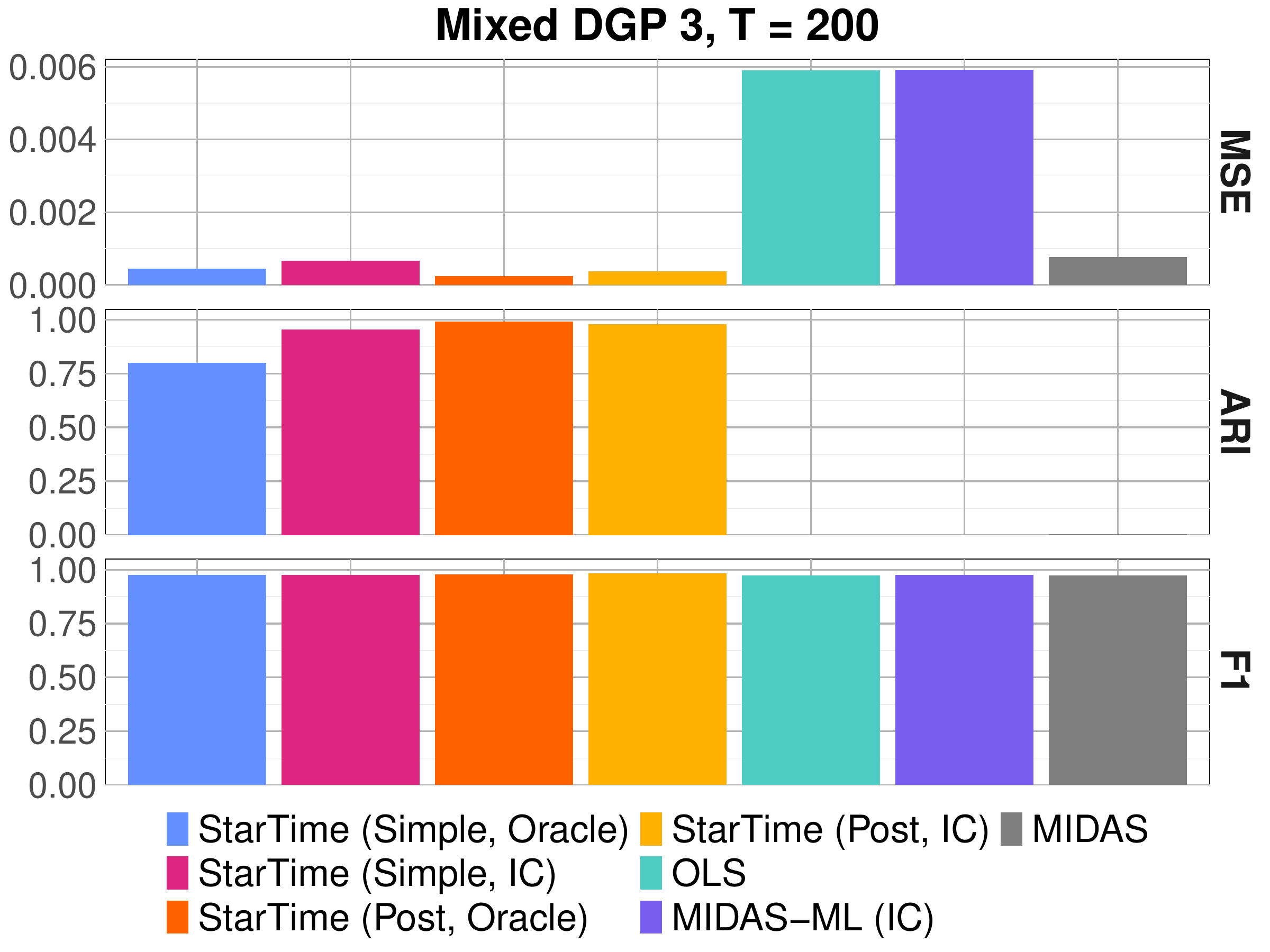}
    \caption{DGP 3, $T=200$}
    \label{fig:Mixeddgp3_n200}
  \end{subfigure}

  \caption{Performance metrics of the estimators across the three mixed-frequency 
  DGPs. 
  }
  \label{fig:simulation_results_Mixed_T200}
\end{figure}

\clearpage
\section{Variables Used in the Financial Application}
\label{appendix:financialvariables}

\begin{table}[h]
\centering
\caption{Overview of the 30 Major U.S. Financial Assets}
\begin{threeparttable}
\begin{tabular}{ll}
\toprule
Ticker & Company Name \\
\midrule
AAPL & Apple Inc. \\
AXP  & American Express Company \\
BA   & The Boeing Company \\
CAT  & Caterpillar Inc. \\
CSCO & Cisco Systems, Inc. \\
CVX  & Chevron Corporation \\
DD   & DuPont de Nemours, Inc. \\
DIS  & The Walt Disney Company \\
GE   & General Electric Company \\
GS   & The Goldman Sachs Group, Inc. \\
HD   & The Home Depot, Inc. \\
IBM  & International Business Machines Corporation \\
INTC & Intel Corporation \\
JNJ  & Johnson \& Johnson \\
JPM  & JPMorgan Chase \& Co. \\
KO   & The Coca-Cola Company \\
MCD  & McDonald's Corporation \\
MMM  & 3M Company \\
MRK  & Merck \& Co., Inc. \\
MSFT & Microsoft Corporation \\
NKE  & NIKE, Inc. \\
PFE  & Pfizer Inc. \\
PG   & The Procter \& Gamble Company \\
TRV  & The Travelers Companies, Inc. \\
UNH  & UnitedHealth Group Incorporated \\
UTX  & United Technologies Corporation \\
V    & Visa Inc. \\
VZ   & Verizon Communications Inc. \\
WMT  & Walmart Inc. \\
XOM  & Exxon Mobil Corporation \\
\bottomrule
\end{tabular}
\end{threeparttable}
\label{tab:empirical_stocks}
\end{table}

\clearpage

\section{Additional Results for the Financial Application} 
\label{appendix:AppendixFin}

\begin{figure}[!h]
  \centering
  \begin{subfigure}[t]{0.75\textwidth}
    \centering
    \includegraphics[width=0.9\textwidth]{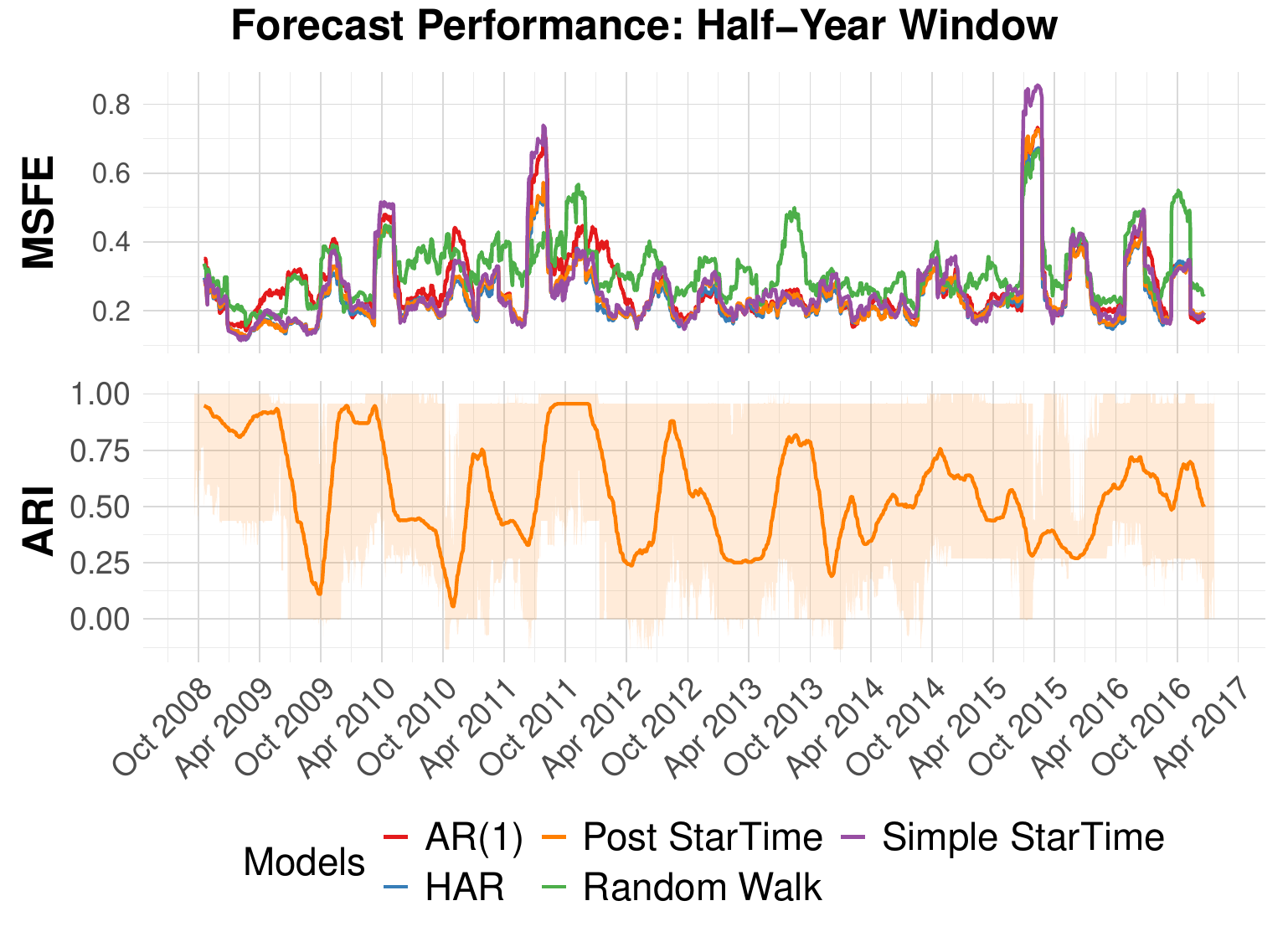}
    \caption{125 Day-window,  $h=1$, 20 lags}
    \label{fig:financeapplication_125}
  \end{subfigure}

  \caption{Comparative forecast performance for $125$-day window. Each panel displays the smoothed daily median MSFE (top) and smoothed daily median ARI (bottom). Shaded regions represent the 10\textsuperscript{th} and 90\textsuperscript{th} percentiles.} 
  \label{fig:financeapplication_T125}
\end{figure}

\begin{figure}[!h]
  \centering

  \begin{subfigure}[t]{0.45\textwidth}
    \centering
    \includegraphics[width=0.8\textwidth]{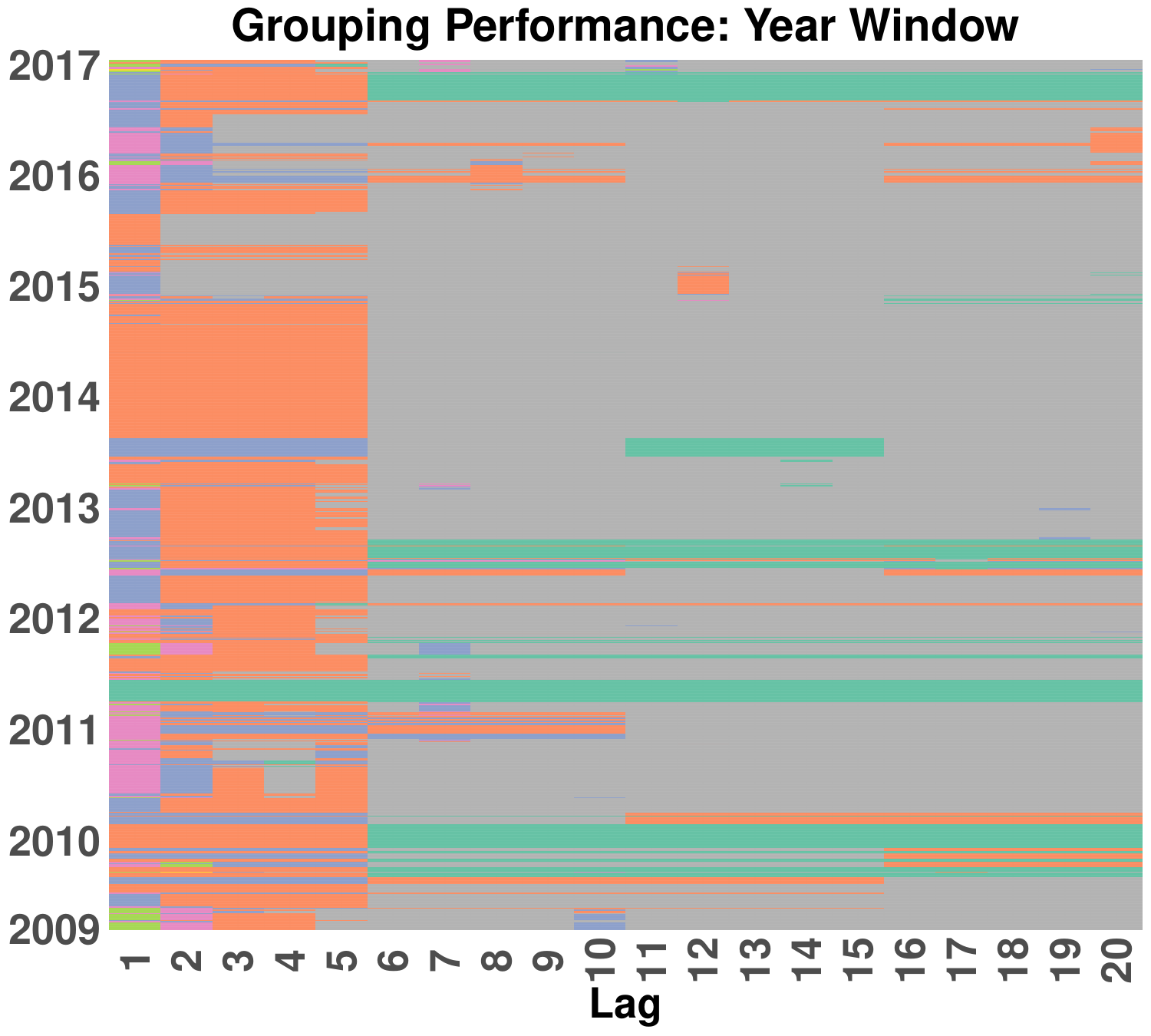}
    \caption{Nike Inc.} 
    \label{fig:finance_matrix_250_nke}
  \end{subfigure}\hfill
  \begin{subfigure}[t]{0.45\textwidth}
    \centering
    \includegraphics[width=0.8\textwidth]{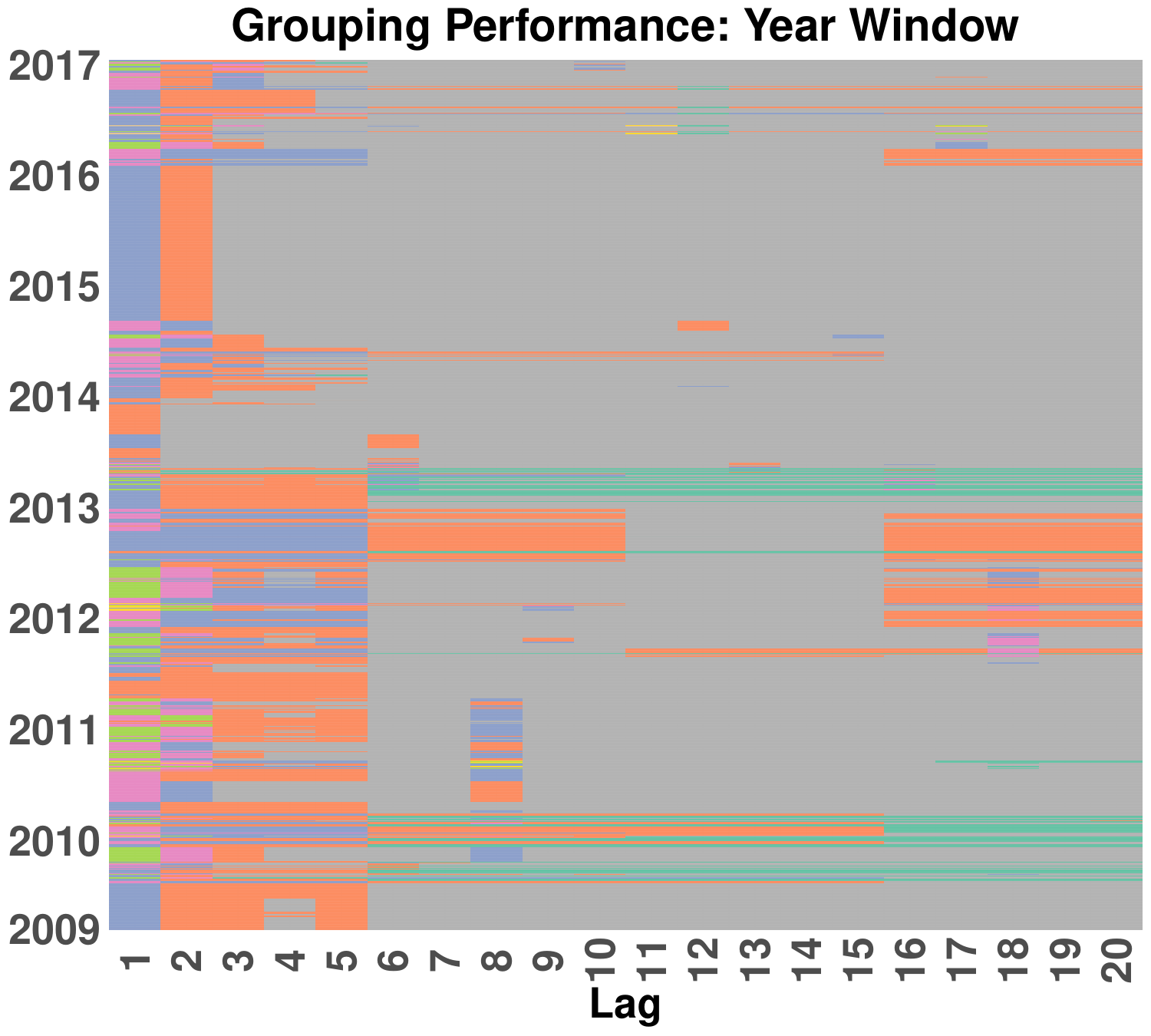}
    \caption{JPMorgan Chase \& Co.} 
    \label{fig:finance_matrix_250_jpm}
  \end{subfigure}

  \caption{Coefficient heatmaps for Post StarTime ($h=1$, 20 lags) applied to Nike Inc. (NKE) and JPMorgan Chase \& Co (JPM), with 250-day window size. 
  Each row indexes an estimation date, columns correspond to daily lags $j=1,\ldots,20$. Cell colors encode the aggregation structure selected by Post StarTime: Lags with identical colors are fused into a single coefficient group, distinct colors within a row reflect heterogeneous lagged effects. Coefficients set to zero are indicated in gray.} 
  \label{fig:financeapplication_matrix_250}
\end{figure}

\begin{table}[!h]
\centering
\caption{\label{tab:mcs_mse} Model Confidence Set Inclusion Rates (MSFE)}
\centering
\begin{threeparttable}
\footnotesize
\setlength{\tabcolsep}{3pt}
\begin{tabular}[t]{cccccccc}
\toprule
\multicolumn{3}{c}{Configuration} & \multicolumn{5}{c}{Models} \\
\cmidrule(l{3pt}r{3pt}){1-3} \cmidrule(l{3pt}r{3pt}){4-8}
Lags & Window & Horizon & Random Walk & AR1 & HAR & Post StarTime & Simple StarTime\\
\midrule
20 & 125 & 1 & 0.000 & 0.000 & 1.000 & 0.000 & 0.000\\
20 & 125 & 5 & 0.000 & 0.633 & 1.000 & 0.700 & 0.767\\
20 & 125 & 20 & 0.000 & 1.000 & 0.667 & 0.567 & 0.867\\
20 & 250 & 1 & 0.000 & 0.000 & 1.000 & 0.067 & 0.000\\
20 & 250 & 5 & 0.000 & 0.033 & 1.000 & 0.433 & 0.433\\
20 & 250 & 20 & 0.000 & 0.967 & 0.967 & 0.900 & 0.967\\
20 & 1000 & 1 & 0.000 & 0.000 & 1.000 & 0.567 & 0.433\\
20 & 1000 & 5 & 0.000 & 0.000 & 1.000 & 0.633 & 0.867\\
20 & 1000 & 20 & 0.000 & 0.333 & 1.000 & 0.967 & 1.000\\
\addlinespace
40 & 125 & 1 & 0.000 & 0.000 & 1.000 & 0.000 & 0.000\\
40 & 125 & 5 & 0.000 & 0.400 & 1.000 & 0.033 & 0.367\\
40 & 125 & 20 & 0.000 & 1.000 & 0.800 & 0.067 & 0.700\\
40 & 250 & 1 & 0.000 & 0.000 & 1.000 & 0.067 & 0.000\\
40 & 250 & 5 & 0.000 & 0.033 & 1.000 & 0.133 & 0.067\\
40 & 250 & 20 & 0.000 & 0.900 & 1.000 & 0.533 & 0.767\\
40 & 1000 & 1 & 0.000 & 0.000 & 1.000 & 0.567 & 0.267\\
40 & 1000 & 5 & 0.000 & 0.000 & 1.000 & 0.533 & 0.867\\
40 & 1000 & 20 & 0.000 & 0.233 & 1.000 & 0.600 & 0.633\\
\bottomrule
\end{tabular}
\begin{tablenotes}
\footnotesize
\item \textit{Note:} Values represent the proportion of stocks included in the 95\% MCS.
\end{tablenotes}
\end{threeparttable}
\end{table}

\begin{table}[!h]
\centering
\caption{\label{tab:dm_mse}Diebold-Mariano Win / Loss Rates: Post StarTime vs Benchmarks (MSFE)}
\centering
\begin{threeparttable}
\footnotesize
\setlength{\tabcolsep}{3pt}
\begin{tabular}[t]{ccccccc}
\toprule
\multicolumn{3}{c}{Configuration} & \multicolumn{4}{c}{Models} \\
\cmidrule(l{3pt}r{3pt}){1-3} \cmidrule(l{3pt}r{3pt}){4-7}
Lags & Window & Horizon & Random Walk & AR1 & HAR & Simple StarTime\\
\midrule
20 & 125 & 1 & 1.000 / 0.000 & 0.933 / 0.000 & 0.000 / 1.000 & 0.900 / 0.000\\
20 & 125 & 5 & 1.000 / 0.000 & 0.100 / 0.000 & 0.000 / 0.333 & 0.033 / 0.033\\
20 & 125 & 20 & 0.967 / 0.000 & 0.000 / 0.367 & 0.000 / 0.033 & 0.000 / 0.367\\
20 & 250 & 1 & 1.000 / 0.000 & 1.000 / 0.000 & 0.000 / 0.900 & 0.167 / 0.000\\
20 & 250 & 5 & 1.000 / 0.000 & 0.933 / 0.000 & 0.000 / 0.667 & 0.000 / 0.000\\
20 & 250 & 20 & 1.000 / 0.000 & 0.033 / 0.000 & 0.000 / 0.067 & 0.000 / 0.000\\
20 & 1000 & 1 & 1.000 / 0.000 & 1.000 / 0.000 & 0.000 / 0.367 & 0.067 / 0.000\\
20 & 1000 & 5 & 1.000 / 0.000 & 1.000 / 0.000 & 0.000 / 0.400 & 0.000 / 0.067\\
20 & 1000 & 20 & 1.000 / 0.000 & 0.700 / 0.000 & 0.000 / 0.100 & 0.000 / 0.033\\
\addlinespace
40 & 125 & 1 & 1.000 / 0.000 & 0.800 / 0.000 & 0.000 / 1.000 & 0.967 / 0.000\\
40 & 125 & 5 & 1.000 / 0.000 & 0.000 / 0.100 & 0.000 / 1.000 & 0.000 / 0.267\\
40 & 125 & 20 & 0.900 / 0.000 & 0.000 / 0.833 & 0.000 / 0.767 & 0.000 / 0.867\\
40 & 250 & 1 & 1.000 / 0.000 & 1.000 / 0.000 & 0.000 / 0.900 & 0.567 / 0.000\\
40 & 250 & 5 & 1.000 / 0.000 & 0.500 / 0.000 & 0.000 / 0.900 & 0.000 / 0.000\\
40 & 250 & 20 & 1.000 / 0.000 & 0.000 / 0.133 & 0.000 / 0.600 & 0.000 / 0.033\\
40 & 1000 & 1 & 1.000 / 0.000 & 1.000 / 0.000 & 0.000 / 0.467 & 0.200 / 0.000\\
40 & 1000 & 5 & 1.000 / 0.000 & 0.967 / 0.000 & 0.000 / 0.433 & 0.000 / 0.067\\
40 & 1000 & 20 & 1.000 / 0.000 & 0.567 / 0.000 & 0.000 / 0.400 & 0.033 / 0.033\\
\bottomrule
\end{tabular}
\begin{tablenotes}
\footnotesize
\item \textit{Note:} Win Rate / Loss Rate. Win = Post StarTime performs significantly better; Loss = Post StarTime performs significantly worse.
\end{tablenotes}
\end{threeparttable}
\end{table}

\begin{table}[!h]
\centering
\caption{\label{tab:mcs_qlike}Model Confidence Set Inclusion Rates (QLIKE)}
\centering
\begin{threeparttable}
\footnotesize
\setlength{\tabcolsep}{3pt}
\begin{tabular}[t]{cccccccc}
\toprule
\multicolumn{3}{c}{Configuration} & \multicolumn{5}{c}{Models} \\
\cmidrule(l{3pt}r{3pt}){1-3} \cmidrule(l{3pt}r{3pt}){4-8}
Lags & Window & Horizon & Random Walk & AR1 & HAR & Post StarTime & Simple StarTime\\
\midrule
20 & 125 & 1 & 0.133 & 0.633 & 1.000 & 0.467 & 0.300\\
20 & 125 & 5 & 0.167 & 0.967 & 1.000 & 0.867 & 0.967\\
20 & 125 & 20 & 0.100 & 1.000 & 0.667 & 0.633 & 0.867\\
20 & 250 & 1 & 0.067 & 0.700 & 1.000 & 0.800 & 0.667\\
20 & 250 & 5 & 0.067 & 0.933 & 1.000 & 0.967 & 0.967\\
20 & 250 & 20 & 0.067 & 1.000 & 0.967 & 0.933 & 1.000\\
20 & 1000 & 1 & 0.067 & 0.933 & 0.967 & 0.967 & 0.800\\
20 & 1000 & 5 & 0.033 & 0.933 & 1.000 & 0.967 & 1.000\\
20 & 1000 & 20 & 0.033 & 1.000 & 0.900 & 0.867 & 0.967\\
\addlinespace
40 & 125 & 1 & 0.133 & 0.667 & 1.000 & 0.433 & 0.267\\
40 & 125 & 5 & 0.133 & 1.000 & 1.000 & 0.600 & 0.867\\
40 & 125 & 20 & 0.033 & 1.000 & 0.733 & 0.267 & 0.767\\
40 & 250 & 1 & 0.067 & 0.700 & 1.000 & 0.800 & 0.600\\
40 & 250 & 5 & 0.067 & 0.933 & 1.000 & 0.833 & 0.867\\
40 & 250 & 20 & 0.067 & 1.000 & 0.867 & 0.767 & 0.867\\
40 & 1000 & 1 & 0.133 & 0.933 & 0.967 & 0.967 & 0.700\\
40 & 1000 & 5 & 0.067 & 0.900 & 1.000 & 0.867 & 0.933\\
40 & 1000 & 20 & 0.000 & 1.000 & 0.900 & 0.667 & 0.900\\
\bottomrule
\end{tabular}
\begin{tablenotes}
\footnotesize
\item \textit{Note}: Values represent the proportion of stocks included in the 95\% MCS.
\end{tablenotes}
\end{threeparttable}
\end{table}

\begin{table}[!h]
\centering
\caption{\label{tab:dm_qlike}Diebold-Mariano Win / Loss Rates: Post StarTime vs Benchmarks (QLIKE)}
\centering
\begin{threeparttable}
\footnotesize
\setlength{\tabcolsep}{3pt}
\begin{tabular}[t]{ccccccc}
\toprule
\multicolumn{3}{c}{Configuration} & \multicolumn{4}{c}{Models} \\
\cmidrule(l{3pt}r{3pt}){1-3} \cmidrule(l{3pt}r{3pt}){4-7}
Lags & Window & Horizon & Random Walk & AR1 & HAR & Simple StarTime\\
\midrule
20 & 125 & 1 & 0.800 / 0.000 & 0.067 / 0.033 & 0.000 / 0.533 & 0.467 / 0.000\\
20 & 125 & 5 & 0.900 / 0.000 & 0.000 / 0.067 & 0.000 / 0.167 & 0.000 / 0.067\\
20 & 125 & 20 & 0.633 / 0.000 & 0.000 / 0.333 & 0.000 / 0.033 & 0.000 / 0.300\\
20 & 250 & 1 & 0.900 / 0.000 & 0.333 / 0.000 & 0.000 / 0.300 & 0.233 / 0.000\\
20 & 250 & 5 & 0.933 / 0.000 & 0.067 / 0.000 & 0.000 / 0.167 & 0.000 / 0.000\\
20 & 250 & 20 & 0.967 / 0.000 & 0.000 / 0.033 & 0.000 / 0.067 & 0.000 / 0.167\\
20 & 1000 & 1 & 0.833 / 0.000 & 0.067 / 0.033 & 0.000 / 0.000 & 0.133 / 0.000\\
20 & 1000 & 5 & 0.967 / 0.000 & 0.100 / 0.000 & 0.000 / 0.067 & 0.000 / 0.000\\
20 & 1000 & 20 & 1.000 / 0.000 & 0.000 / 0.000 & 0.000 / 0.033 & 0.000 / 0.167\\
\addlinespace
40 & 125 & 1 & 0.833 / 0.000 & 0.033 / 0.033 & 0.000 / 0.600 & 0.433 / 0.000\\
40 & 125 & 5 & 0.733 / 0.000 & 0.000 / 0.300 & 0.000 / 0.467 & 0.000 / 0.200\\
40 & 125 & 20 & 0.633 / 0.000 & 0.000 / 0.600 & 0.000 / 0.300 & 0.000 / 0.700\\
40 & 250 & 1 & 0.867 / 0.000 & 0.300 / 0.000 & 0.000 / 0.233 & 0.433 / 0.000\\
40 & 250 & 5 & 0.900 / 0.000 & 0.033 / 0.033 & 0.000 / 0.333 & 0.000 / 0.000\\
40 & 250 & 20 & 0.933 / 0.000 & 0.000 / 0.233 & 0.000 / 0.067 & 0.000 / 0.167\\
40 & 1000 & 1 & 0.833 / 0.000 & 0.067 / 0.000 & 0.000 / 0.033 & 0.267 / 0.000\\
40 & 1000 & 5 & 0.967 / 0.000 & 0.100 / 0.000 & 0.000 / 0.133 & 0.000 / 0.067\\
40 & 1000 & 20 & 1.000 / 0.000 & 0.000 / 0.033 & 0.000 / 0.033 & 0.000 / 0.167\\
\bottomrule
\end{tabular}
\begin{tablenotes}
\small
\item \textit{Note: } Win Rate / Loss Rate. Win = Post StarTime performs significantly better; Loss = Post StarTime performs significantly worse. 
\end{tablenotes}
\end{threeparttable}
\end{table}

\clearpage

\section{Variables Used in the Macroeconomic Application}
\label{appendix:MacroVariables}

\begin{table}[h!]
\centering
\caption{Daily Macroeconomic Data Series}
\footnotesize
\setlength{\tabcolsep}{3pt}
\begin{threeparttable}
\begin{tabular}{llc}
\toprule
Series Name & Code & Transformation \\
\midrule
\addlinespace[0.3em]
\multicolumn{3}{l}{\textit{Market Indices}} \\
\hspace{1em}S\&P 500 & SP500 & $\Delta \log(x)$ \\
\hspace{1em}Dow Jones Industrial Average & DJIA & $\Delta\log(x)$ \\
\hspace{1em}NASDAQ Composite Index & NASDAQCOM & $\Delta\log(x)$ \\
\hspace{1em}CBOE Volatility Index: VIX & VIXCLS & $\log(x)$ \\
\addlinespace[0.5em]
\multicolumn{3}{l}{\textit{Policy Uncertainty}} \\
\hspace{1em}Economic Policy Uncertainty & EPU & $\log(x)$ \\
\bottomrule
\end{tabular}
\begin{tablenotes}
\footnotesize
\item \textit{Note: None of the daily macroeconomic data series are seasonally adjusted, nor are any included in the reduced set.}
\end{tablenotes}
\end{threeparttable}
\label{tab:macro_daily}
\end{table}

\begin{table}[h!]
\centering
\caption{Weekly Macroeconomic Data Series}
\footnotesize
\setlength{\tabcolsep}{2pt}
\begin{threeparttable}
\begin{tabular}{llcc}
\toprule
Series Name & Code & Transformation & Reduced Set \\
\midrule
\addlinespace[0.3em]
\multicolumn{4}{l}{\textit{Financial Indicators}} \\
\hspace{1em}Chicago Fed National Financial Conditions Index & NFCI & $x$ & \checkmark \\
\addlinespace[0.5em]
\multicolumn{4}{l}{\textit{Labor Market}} \\
\hspace{1em}Continued Claims (Insured Unemployment)$^\ast$ & CCSA & $\Delta \log(x)$ & \\
\hspace{1em}Initial Claims$^\ast$ & ICSA & $\Delta \log(x)$ & \checkmark \\
\bottomrule
\end{tabular}
\begin{tablenotes}
\footnotesize
\item \textit{Note: Series marked with an asterisk ($^\ast$) are seasonally adjusted. A checkmark (\checkmark) indicates the variable is included in the reduced set.}
\end{tablenotes}
\end{threeparttable}
\label{tab:macro_weekly}
\end{table}

\begin{table}[h!]
\centering
\caption{Monthly Macroeconomic Data Series}
\footnotesize
\begin{threeparttable}
\begin{tabularx}{\textwidth}{>{\hangindent=1em \hangafter=1 \RaggedRight}X l c c}
\toprule
Series Name & Code & Transformation & Reduced Set \\
\midrule
\addlinespace[0.3em]
\multicolumn{4}{l}{\textit{Income, Consumption, and Retail}} \\
Real Personal Income$^\ast$ & RPI & $\Delta \log(x)$ & \\
Real Disposable Personal Income$^\ast$ & DSPIC96 & $\Delta \log(x)$ & \\
Advance Retail Sales$^\ast$ & RSAFS & $\Delta \log(x)$ & \checkmark \\
Real Manufacturing and Trade Industries Sales$^\ast$ & CMRMTSPL & $\Delta \log(x)$ & \\

\addlinespace[0.5em]
\multicolumn{4}{l}{\textit{Labor Market}} \\
Unemployment Rate$^\ast$ & UNRATE & $\Delta x$ & \checkmark \\
All Employees, Total Nonfarm$^\ast$ & PAYEMS & $\Delta \log(x)$ & \checkmark \\

\addlinespace[0.5em]
\multicolumn{4}{l}{\textit{Real Economic Activity}} \\
Industrial Production: Total Index$^\ast$ & INDPRO & $\Delta \log(x)$ & \checkmark \\
Capacity Utilization: Total Index$^\ast$ & TCU & $\Delta x$ & \\
Total Business Inventories$^\ast$ & BUSINV & $\Delta \log(x)$ & \\
Manufacturers' New Orders: Total Manufacturing$^\ast$ & AMTMNO & $\Delta \log(x)$ & \checkmark \\
New Privately-Owned Housing Units Started$^\ast$ & HOUST & $\log(x)$ & \checkmark \\
New Privately-Owned Housing Units Authorized$^\ast$ & PERMIT & $\log(x)$ & \checkmark \\
New One Family Houses Sold$^\ast$ & HSN1F & $\log(x)$ & \\

\addlinespace[0.5em]
\multicolumn{4}{l}{\textit{Prices and Inflation}} \\
Consumer Price Index (CPI)$^\ast$ & CPIAUCSL & $\Delta^2 \log(x)$ & \checkmark \\
Consumer Price Index Less Food and Energy$^\ast$ & CPILFESL & $\Delta^2 \log(x)$ & \\
PCE: Chain-type Price Index$^\ast$ & PCEPI & $\Delta^2 \log(x)$ & \\
PCE Excluding Food and Energy$^\ast$ & PCEPILFE & $\Delta^2 \log(x)$ & \\
PCE: Food$^\ast$ & DFXARC1M027SBEA & $\Delta^2 \log(x)$ & \\
PCE: Energy goods and services$^\ast$ & DNRGRC1M027SBEA & $\Delta^2 \log(x)$ & \\
Producer Price Index by Commodity: All Commodities & PPIACO & $\Delta^2 \log(x)$ & \\
\bottomrule
\end{tabularx}
\begin{tablenotes}
\footnotesize
\item \textit{Note: Series marked with an asterisk ($^\ast$) are seasonally adjusted. A checkmark (\checkmark) indicates the variable is included in the reduced set.}
\end{tablenotes}
\end{threeparttable}
\label{tab:macro_monthly}
\end{table}

\begin{table}[h!]
\centering
\caption{Quarterly Macroeconomic Data Series}
\begin{threeparttable}
\footnotesize
\setlength{\tabcolsep}{2.3pt}
\begin{tabular}{llcc}
\toprule
Series Name & Code & Transformation & Reduced Set \\
\midrule
\addlinespace[0.3em]
\multicolumn{4}{l}{\textit{Aggregate Economic Activity}} \\
\hspace{1em}Gross Domestic Product$^\ast$ & GDP & $\Delta \log(x)$ & \checkmark \\
\addlinespace[0.5em]
\multicolumn{4}{l}{\textit{Labor Market}} \\
\hspace{1em}Nonfarm Business Sector: Unit Labor Costs$^\ast$ & ULCNFB & $\Delta \log(x)$ & \\
\bottomrule
\end{tabular}
\begin{tablenotes}
\footnotesize
\item \textit{Note: Series marked with an asterisk ($^\ast$) are seasonally adjusted. A checkmark (\checkmark) indicates the variable is included in the reduced set.}
\end{tablenotes}
\end{threeparttable}
\label{tab:macro_quarterly}
\end{table}

\clearpage

\section{Additional Results for the Macroeconomic Application}
\label{app:MacroAppendix}

\subsection{Plots}
\label{app:MacroAppendixPlots}

\begin{figure}[!h]
  \centering
    \includegraphics[width=\textwidth]{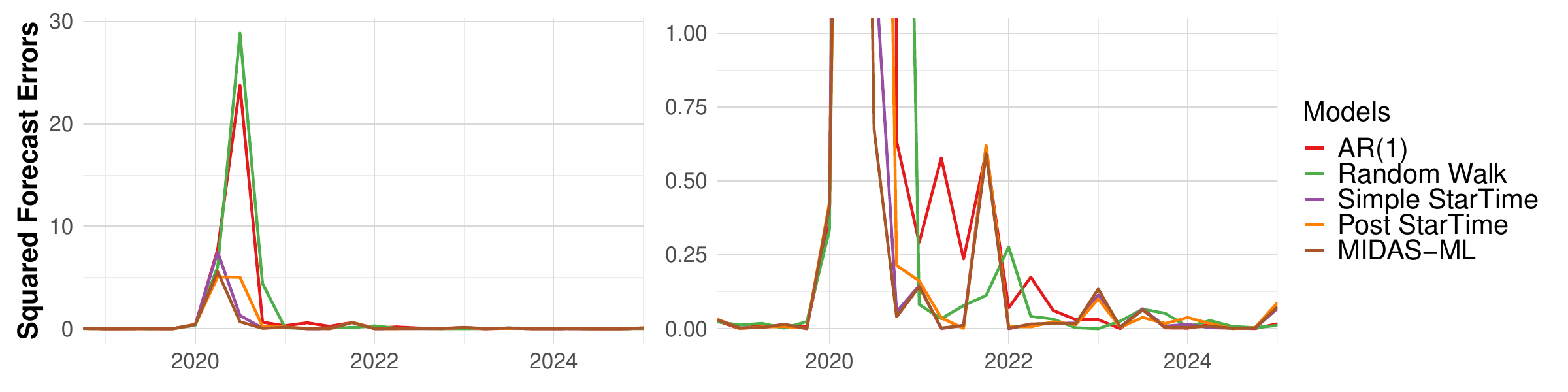}
  \caption{Evolution of squared forecast errors for the Reduced Set (49 predictors) with $W = 105$. 
  A censored zoom  accounts for COVID-19 volatility. The plotted squared forecast errors are multiplied by 1000 due to the small scale of GDP growth.}
  \label{fig:macro_reduced_105_forecast}
\end{figure}

\begin{figure}[h!]
  \centering
  \begin{subfigure}[b]{\textwidth}
  \includegraphics[width=\textwidth]{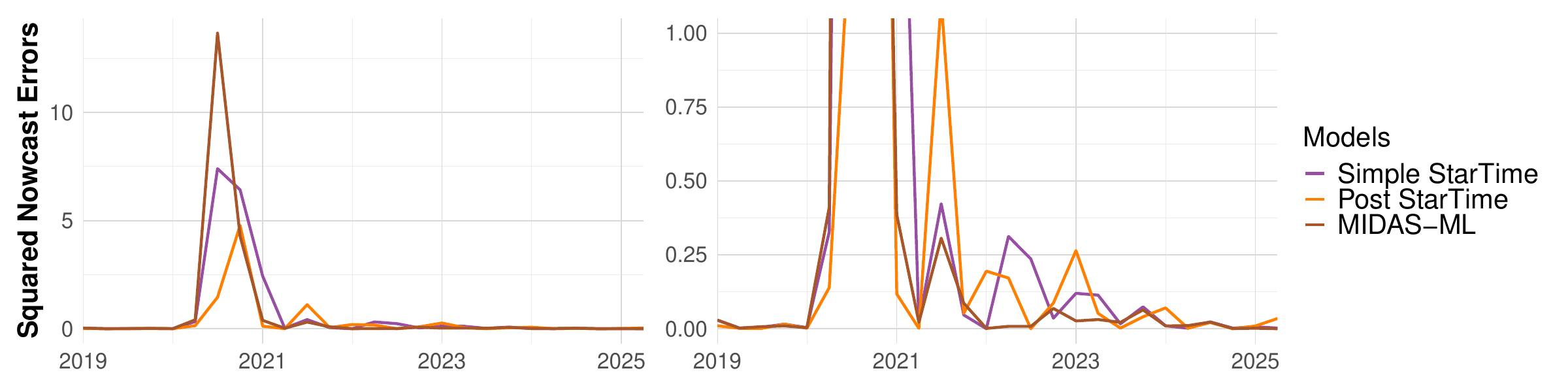}
    \caption{Nowcast}
  \end{subfigure}
  \vspace{1cm}
  \begin{subfigure}[b]{\textwidth}
    \centering
\includegraphics[width=\textwidth]{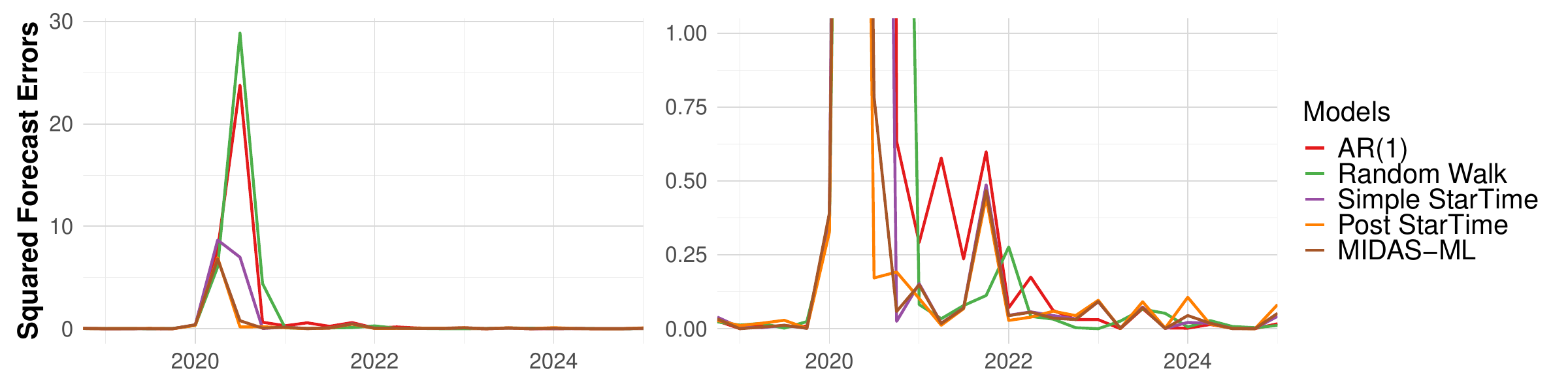}
    \caption{Forecast}
  \end{subfigure}
  \caption{Evolution of the squared forecast errors  for the Full Set (398 predictors) with $W = 105$. Both (a) and (b) include a censored zoom to account for COVID-19 volatility. The plotted squared forecast errors values are multiplied by 1000.}
\label{fig:macro_full_105}
\end{figure}

\begin{figure}[h!]
  \centering
  \begin{subfigure}[b]{\textwidth}
    \includegraphics[width=\textwidth]{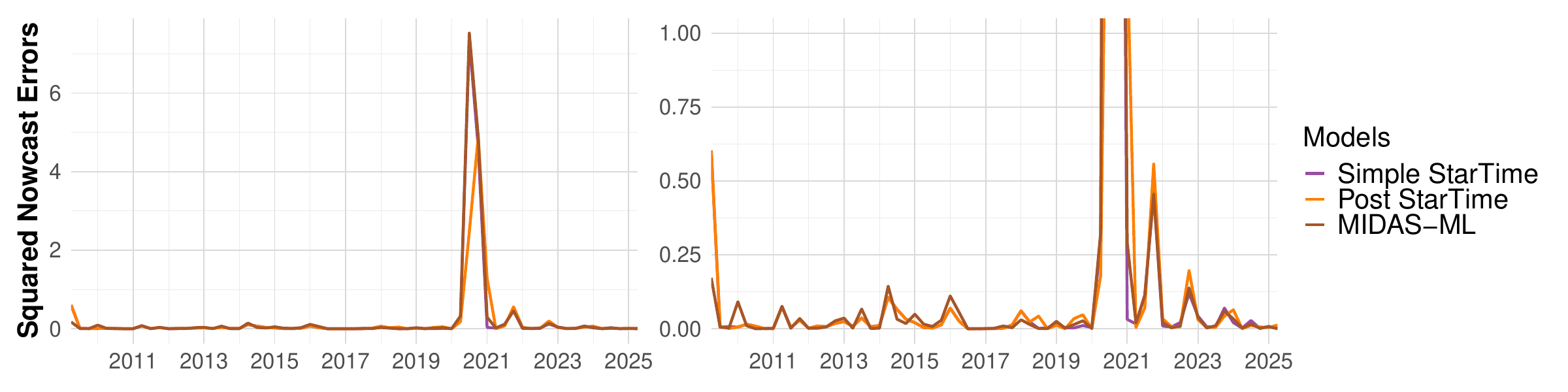}
    \caption{Nowcast}
  \end{subfigure}
  \vspace{1cm}
  \begin{subfigure}[b]{\textwidth}
    \centering
    \includegraphics[width=\textwidth]{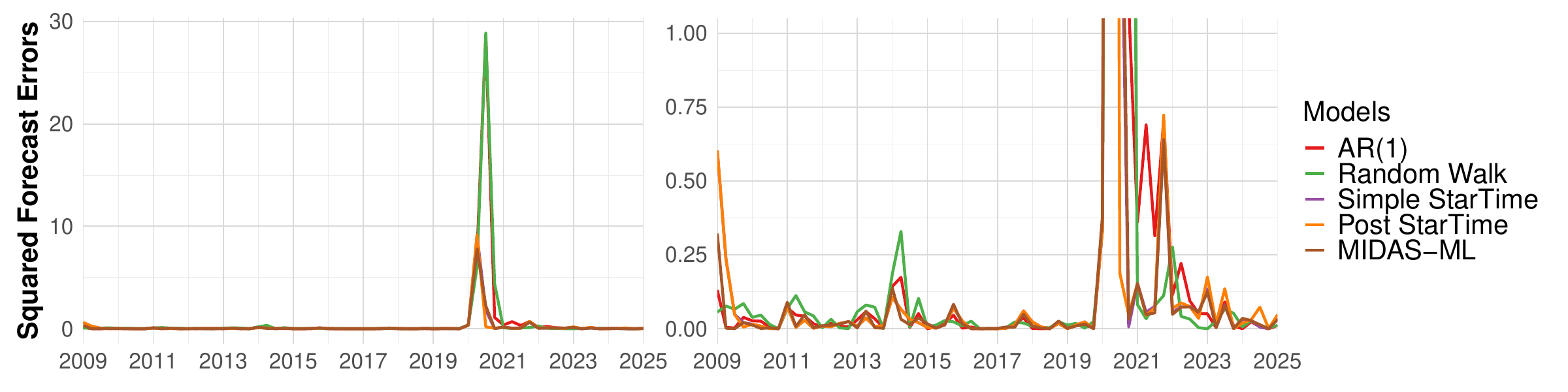}
    \caption{Forecast}
  \end{subfigure}
   \caption{Evolution of the squared forecast errors for the Full Set (398 predictors) with $W = 66$. Both (a) and (b) include a censored zoom to account for COVID-19 volatility. The plotted squared forecast errors values are multiplied by 1000.}
  \label{fig:macro_full_66}
\end{figure}

\begin{figure}[h!]
  \centering
  \begin{subfigure}[b]{\textwidth}
    \includegraphics[width=\textwidth]{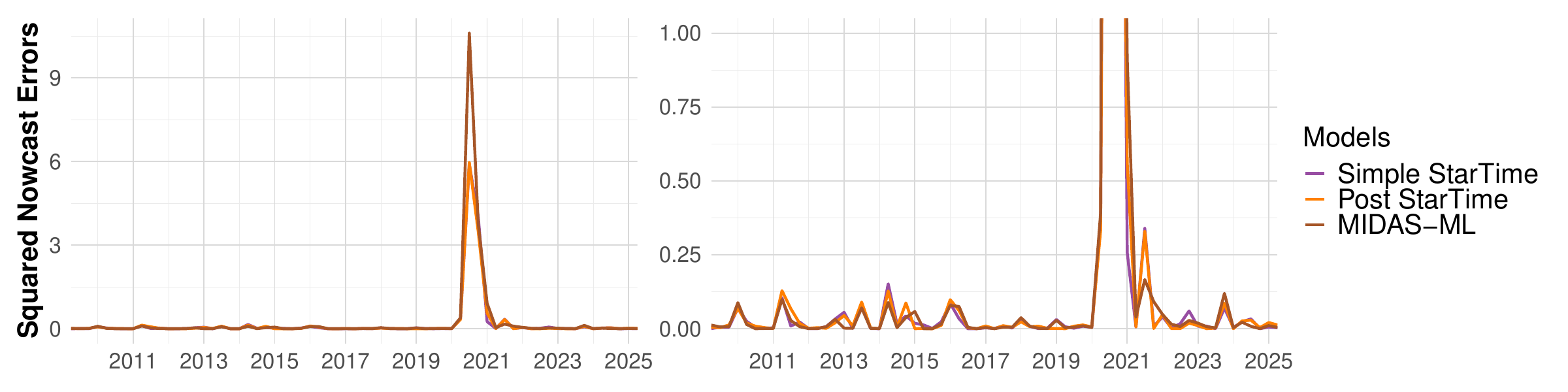}
    \caption{Nowcast}
  \end{subfigure}
  \vspace{1cm}
  \begin{subfigure}[b]{\textwidth}
    \centering
    \includegraphics[width=\textwidth]{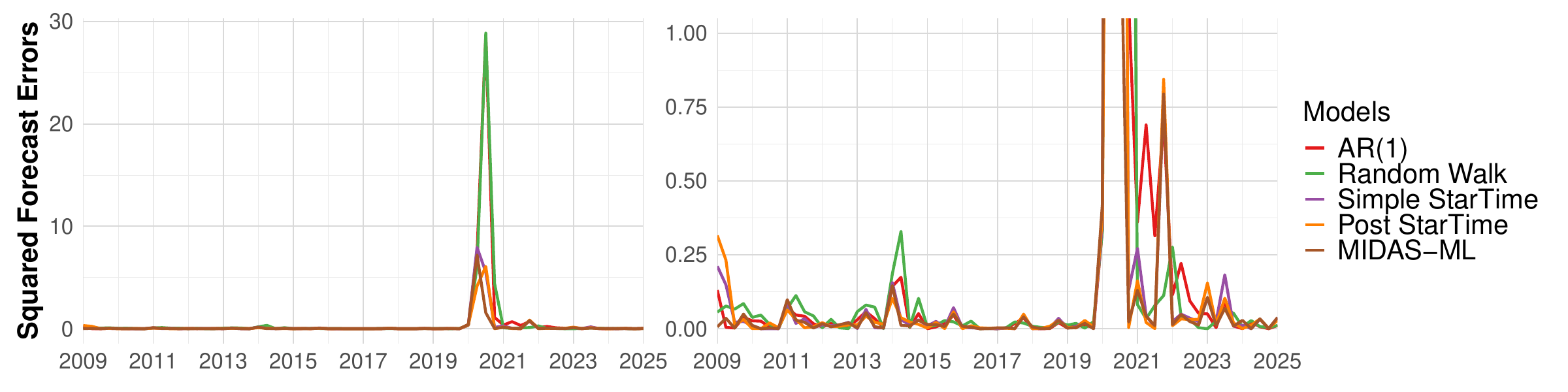}
    \caption{Forecast}
  \end{subfigure}
  \caption{Evolution of the squared forecast errors for the Reduced Set (49 predictors) with $W = 66$. Both (a) and (b) include a censored zoom to account for COVID-19 volatility. The plotted squared forecast errors values are multiplied by 1000.}
  \label{fig:macro_reduced_66}
\end{figure}

\clearpage

\subsection{Tables}
\label{app:MacroAppendixTests}

\begin{table}[h]
\centering
\caption{MSFE for window $W = 66$}
\footnotesize
\setlength{\tabcolsep}{2.5pt}
\begin{threeparttable}
\begin{tabular}{lcccccccc}
\toprule
& \multicolumn{2}{c}{Post StarTime}
& \multicolumn{2}{c}{Simple StarTime}
& \multicolumn{2}{c}{MIDAS-ML}
& AR(1) & Random walk \\
\cmidrule(lr){2-3} \cmidrule(lr){4-5} \cmidrule(lr){6-7}
& $h=0$ & $h=1$ & $h=0$ & $h=1$ & $h=0$ & $h=1$ &  &  \\
\midrule

\addlinespace[0.3em]
\multicolumn{9}{l}{\textit{Reduced Set}} \\
\hspace{1em}Total     & 0.184 & 0.206 & 0.186 & 0.258 & 0.264 & 0.175 & 0.637 & 0.651 \\
\hspace{1em}Non-peak  & 0.021 & 0.031 & 0.020 & 0.030 & 0.019 & 0.021 & 0.032 & 0.041 \\
\hspace{1em}COVID-Peak      & 1.342 & 1.456 & 1.367 & 1.884 & 2.005 & 1.272 & 4.947 & 4.994 \\

\addlinespace[0.5em]
\multicolumn{9}{l}{\textit{Full Set}} \\
\hspace{1em}Total     & 0.174 & 0.202 & 0.225 & 0.222 & 0.232 & 0.200 & 0.637 & 0.651 \\
\hspace{1em}Non-peak  & 0.032 & 0.042 & 0.030 & 0.039 & 0.025 & 0.029 & 0.032 & 0.041 \\
\hspace{1em}COVID-Peak      & 1.183 & 1.338 & 1.614 & 1.526 & 1.706 & 1.418 & 4.947 & 4.994 \\

\bottomrule
\end{tabular}
\begin{tablenotes}
\item \textit{Note: MSFE values are multiplied by  $1000$ due to the small scale of GDP growth.}
\end{tablenotes}
\end{threeparttable}
\label{tab:macro_w66}
\end{table}

\begin{table}[!h]
\centering
\caption{\label{tab:mcs_macro_w105_total} Model Confidence Set $p$-values: $T_{max}$ (Total, $W=105$)}
\begin{threeparttable}
{\footnotesize\setlength{\tabcolsep}{2pt}
\begin{tabular}[t]{lccccccc}
\toprule
\multicolumn{3}{c}{Configuration} & \multicolumn{5}{c}{Models} \\
\cmidrule(l{3pt}r{3pt}){1-3} \cmidrule(l{3pt}r{3pt}){4-8}
Reduction & Window & Horizon & \shortstack{Post\\StarTime} & \shortstack{Simple\\StarTime} & MIDAS-ML & AR(1) & \shortstack{Random\\Walk}\\
\midrule
\addlinespace[0.3em]
\multicolumn{8}{l}{}\\
\hspace{1em}Full & 105 & 0 & 1.000 & 0.183 & 0.729 & & \\
\hspace{1em}Full & 105 & 1 & 1.000 & 0.642 & 1.000 & & \\
\hspace{1em}Reduced & 105 & 0 & 1.000 & 1.000 & 1.000 & & \\
\hspace{1em}Reduced & 105 & 1 & 1.000 & 1.000 & 1.000 & & \\
\multicolumn{2}{l}{\hspace{1em}Univariate Benchmark} & & & & & 0.531 & 0.620 \\
\bottomrule
\end{tabular}}
\begin{tablenotes}
\footnotesize
\item \textit{Note: Values represent $p$-values for the Model Confidence Set (MCS) using the $T_{max}$ statistic. $p$-values $\ge 0.05$ indicate inclusion in the MCS.}
\end{tablenotes}
\end{threeparttable}
\end{table}

\begin{table}[!h]
\centering
\caption{\label{tab:mcs_macro_w105_offpeak}Model Confidence Set $p$-values: $T_{max}$ (Non-Peak, $W=105$)}
\begin{threeparttable}
{\footnotesize\setlength{\tabcolsep}{4pt}
\begin{tabular}[t]{lccccccc}
\toprule
\multicolumn{3}{c}{Configuration} & \multicolumn{5}{c}{Models} \\
\cmidrule(l{3pt}r{3pt}){1-3} \cmidrule(l{3pt}r{3pt}){4-8}
Reduction & Window & Horizon & \shortstack{Post\\StarTime} & \shortstack{Simple\\StarTime} & MIDAS-ML & AR(1) & \shortstack{Random\\Walk}\\
\midrule
\addlinespace[0.3em]
\multicolumn{8}{l}{}\\
\hspace{1em}Full & 105 & 0 & 0.550 & 0.537 & 1.000 & & \\
\hspace{1em}Full & 105 & 1 & 0.272 & 1.000 & 1.000 & & \\
\hspace{1em}Reduced & 105 & 0 & 1.000 & 1.000 & 1.000 & & \\
\hspace{1em}Reduced & 105 & 1 & 1.000 & 1.000 & 1.000 & & \\
\multicolumn{2}{l}{\hspace{1em}Univariate Benchmark} & & & & & 0.999 & 0.942 \\
\bottomrule
\end{tabular}}
\begin{tablenotes}
\item \textit{Note: Values represent $p$-values for the Model Confidence Set (MCS) using the $T_{max}$ statistic. $p$-values $\ge 0.05$ indicate inclusion in the MCS.}
\end{tablenotes}
\end{threeparttable}
\end{table}

\begin{table}[h!]
\centering
\caption{\label{tab:mcs_macro_w66_total}Model Confidence Set P-Values: $T_{max}$ (Total, $W=66$)}
\begin{threeparttable}
{\footnotesize\setlength{\tabcolsep}{4pt}
\begin{tabular}[t]{lccccccc}
\toprule
\multicolumn{3}{c}{Configuration} & \multicolumn{5}{c}{Models} \\
\cmidrule(l{3pt}r{3pt}){1-3} \cmidrule(l{3pt}r{3pt}){4-8}
Reduction & Window & Horizon & \shortstack{Post\\StarTime} & \shortstack{Simple\\StarTime} & MIDAS-ML & AR(1) & \shortstack{Random\\Walk}\\
\midrule
\addlinespace[0.3em]
\multicolumn{8}{l}{}\\
\hspace{1em}Full & 66 & 0 & 1.000 & 1.000 & 1.000 & & \\
\hspace{1em}Full & 66 & 1 & 1.000 & 1.000 & 1.000 & & \\
\hspace{1em}Reduced & 66 & 0 & 1.000 & 1.000 & 1.000 & & \\
\hspace{1em}Reduced & 66 & 1 & 1.000 & 1.000 & 1.000 & & \\
\multicolumn{2}{l}{\hspace{1em}Univariate Benchmark} & & & & & 0.529 & 0.546 \\
\bottomrule
\end{tabular}}
\begin{tablenotes}
\item \textit{Note: Values represent $p$-values for the Model Confidence Set (MCS) using the $T_{max}$ statistic. $p$-values $\ge 0.05$ indicate inclusion in the MCS.}
\end{tablenotes}
\end{threeparttable}
\end{table}

\begin{table}[!ht]
\centering
\caption{\label{tab:mcs_macro_w66_offpeak}Model Confidence Set P-Values: $T_{max}$ (Non-Peak, $W=66$)}
\begin{threeparttable}
{\footnotesize\setlength{\tabcolsep}{4pt}
\begin{tabular}[t]{lccccccc}
\toprule
\multicolumn{3}{c}{Configuration} & \multicolumn{5}{c}{Models} \\
\cmidrule(l{3pt}r{3pt}){1-3} \cmidrule(l{3pt}r{3pt}){4-8}
Reduction & Window & Horizon & \shortstack{Post\\StarTime} & \shortstack{Simple\\StarTime} & MIDAS-ML & AR(1) & \shortstack{Random\\Walk}\\
\midrule
\addlinespace[0.3em]
\multicolumn{8}{l}{}\\
\hspace{1em}Full & 66 & 0 & 1.000 & 1.000 & 1.000 & & \\
\hspace{1em}Full & 66 & 1 & 0.821 & 0.939 & 1.000 & & \\
\hspace{1em}Reduced & 66 & 0 & 1.000 & 1.000 & 1.000 & & \\
\hspace{1em}Reduced & 66 & 1 & 1.000 & 1.000 & 1.000 & & \\
\multicolumn{2}{l}{\hspace{1em}Univariate Benchmark} & & & & & 1.000 & 0.599 \\
\bottomrule
\end{tabular}}
\begin{tablenotes}
\item \textit{Note: Values represent $p$-values for the Model Confidence Set (MCS) using the $T_{max}$ statistic. $p$-values $\ge 0.05$ indicate inclusion in the MCS.}
\end{tablenotes}
\end{threeparttable}
\end{table}

\subsection{Lag Aggregation and Variable Selection}
Only variables and lags reaching the $35\%$ selection threshold are shown. The \emph{aggregation rate} is the share of estimation windows in which all lags of the variable are grouped into a single aggregated term; ``--'' marks variables below the threshold on this measure. The \emph{lag retention} is the share of windows in which the individual lag is selected with a non-zero coefficient. 

\begin{table}[H]
\centering
\caption{\label{tab:macro_coef_nowcast_full_w66}Aggregation and Selection: Full set, Nowcast, $W = 66$}
\begin{threeparttable}
\footnotesize
\setlength{\tabcolsep}{6pt}
\begin{tabular}[t]{lccc}
\toprule
& Aggregation & Lag & Lag \\[-1ex]
Variable & rate &  & retention \\
\midrule
CMRMTSPL & -- & 2 & 43.1\% \\
\addlinespace
BUSINV & -- & 1 & 43.1\% \\
\addlinespace
RSAFS & -- & 2 & 35.4\% \\
\bottomrule
\end{tabular}
\end{threeparttable}
\end{table}

\begin{table}[H]
\centering
\caption{\label{tab:macro_coef_nowcast_full_w105_pt1}Aggregation and Selection: Full set, Nowcast, $W = 105$ (Part 1 of 2)}
\begin{threeparttable}
\footnotesize
\setlength{\tabcolsep}{6pt}
\begin{tabular}[t]{lccc}
\toprule
& Aggregation & Lag & Lag \\[-1ex]
Variable & rate &  & retention \\
\midrule
BUSINV & 80.8\% & 0 & 88.5\% \\
 & & 1 & 96.2\% \\
 & & 2 & 80.8\% \\
\addlinespace
NFCI & 88.5\% & 0--3 & 88.5\% \\
 & & 4--11 & 84.6\% \\
\addlinespace
GDP & -- & 1 & 88.5\% \\
\addlinespace
EPU & 69.2\% & 0--35 & 65.4\% \\
 & & 36 & 69.2\% \\
 & & 37--39 & 65.4\% \\
 & & 40--42 & 69.2\% \\
 & & 43 & 76.9\% \\
 & & 44 & 69.2\% \\
 & & 45--49 & 76.9\% \\
 & & 50--59 & 69.2\% \\
\addlinespace
CCSA & 38.5\% & 0--6 & 61.5\% \\
 & & 7 & 76.9\% \\
 & & 8 & 61.5\% \\
 & & 9 & 76.9\% \\
 & & 10 & 61.5\% \\
 & & 11 & 73.1\% \\
\addlinespace
HOUST & 73.1\% & 0--2 & 73.1\% \\
\bottomrule
\end{tabular}
\end{threeparttable}
\end{table}

\begin{table}[H]
\centering
\caption{\label{tab:macro_coef_nowcast_full_w105_pt2}Aggregation and Selection: Full set, Nowcast, $W = 105$ (Part 2 of 2)}
\begin{threeparttable}
\footnotesize
\setlength{\tabcolsep}{6pt}
\begin{tabular}[t]{lccc}
\toprule
& Aggregation & Lag & Lag \\[-1ex]
Variable & rate &  & retention \\
\midrule
PERMIT & 61.5\% & 0 & 73.1\% \\
 & & 1--2 & 61.5\% \\
\addlinespace
INDPRO & 38.5\% & 0--1 & 38.5\% \\
 & & 2 & 73.1\% \\
\addlinespace
HSN1F & 69.2\% & 0--2 & 69.2\% \\
\addlinespace
ICSA & 57.7\% & 0--4 & 57.7\% \\
 & & 5 & 65.4\% \\
 & & 6--11 & 57.7\% \\
\addlinespace
ULCNFB & -- & 0 & 53.8\% \\
\addlinespace
RSAFS & 50.0\% & 0--2 & 50.0\% \\
\addlinespace
SP500 & -- & 10 & 46.2\% \\
 & & 41--42 & 38.5\% \\
 & & 59 & 38.5\% \\
\addlinespace
CMRMTSPL & -- & 2 & 46.2\% \\
\addlinespace
PPIACO & 42.3\% & 0--2 & 42.3\% \\
\addlinespace
PCEPI & 42.3\% & 0--2 & 42.3\% \\
\addlinespace
UNRATE & -- & 1 & 38.5\% \\
\addlinespace
NASDAQCOM & -- & 59 & 34.6\% \\
\addlinespace
VIXCLS & 34.6\% & 0--59 & 34.6\% \\
\addlinespace
AMTMNO & 34.6\% & 0--2 & 34.6\% \\
\addlinespace
CPILFESL & -- & 2 & 34.6\% \\
\addlinespace
PAYEMS & -- & 2 & 34.6\% \\
\bottomrule
\end{tabular}
\end{threeparttable}
\end{table}

\begin{table}[H]
\centering
\caption{\label{tab:macro_coef_nowcast_reduced_w66}Aggregation and Selection: Reduced set, Nowcast, $W = 66$}
\begin{threeparttable}
\footnotesize
\setlength{\tabcolsep}{6pt}
\begin{tabular}[t]{lccc}
\toprule
& Aggregation & Lag & Lag \\[-1ex]
Variable & rate &  & retention \\
\midrule
NFCI & 90.8\% & 0--11 & 90.8\% \\
\addlinespace
UNRATE & 73.8\% & 0--1 & 80.0\% \\
 & & 2 & 83.1\% \\
\addlinespace
AMTMNO & 72.3\% & 0--2 & 76.9\% \\
\addlinespace
PAYEMS & 46.2\% & 0--1 & 46.2\% \\
 & & 2 & 69.2\% \\
\addlinespace
PERMIT & 58.5\% & 0--2 & 58.5\% \\
\addlinespace
INDPRO & -- & 2 & 49.2\% \\
\addlinespace
ICSA & -- & 5 & 46.2\% \\
 & & 7 & 44.6\% \\
\addlinespace
RSAFS & -- & 2 & 43.1\% \\
\addlinespace
HOUST & 41.5\% & 0--2 & 41.5\% \\
\bottomrule
\end{tabular}
\end{threeparttable}
\end{table}

\begin{table}[H]
\centering
\caption{\label{tab:macro_coef_nowcast_reduced_w105}Aggregation and Selection: Reduced set, Nowcast, $W = 105$}
\begin{threeparttable}
\footnotesize
\setlength{\tabcolsep}{6pt}
\begin{tabular}[t]{lccc}
\toprule
& Aggregation & Lag & Lag \\[-1ex]
Variable & rate &  & retention \\
\midrule
AMTMNO & 96.2\% & 0--2 & 100.0\% \\
\addlinespace
UNRATE & 38.5\% & 0--1 & 96.2\% \\
 & & 2 & 100.0\% \\
\addlinespace
PAYEMS & -- & 0--1 & 84.6\% \\
 & & 2 & 100.0\% \\
\addlinespace
INDPRO & -- & 0--1 & 34.6\% \\
 & & 2 & 100.0\% \\
\addlinespace
ICSA & -- & 0 & 61.5\% \\
 & & 4 & 34.6\% \\
 & & 5 & 80.8\% \\
 & & 6 & 38.5\% \\
 & & 7 & 61.5\% \\
\addlinespace
GDP & -- & 1 & 69.2\% \\
\addlinespace
RSAFS & -- & 2 & 65.4\% \\
\addlinespace
PERMIT & 65.4\% & 0--2 & 65.4\% \\
\addlinespace
NFCI & 46.2\% & 0--11 & 46.2\% \\
\addlinespace
HOUST & 34.6\% & 0--2 & 34.6\% \\
\bottomrule
\end{tabular}
\end{threeparttable}
\end{table}

\begin{table}[H]
\centering
\caption{\label{tab:macro_coef_forecast_full_w66}Aggregation and Selection: Full set, Forecast, $W = 66$}
\begin{threeparttable}
\footnotesize
\setlength{\tabcolsep}{6pt}
\begin{tabular}[t]{lccc}
\toprule
& Aggregation & Lag & Lag \\[-1ex]
Variable & rate &  & retention \\
\midrule
\multicolumn{4}{c}{\textit{No variable reached the 35\% threshold.}} \\
\bottomrule
\end{tabular}
\end{threeparttable}
\end{table}

\begin{table}[H]
\centering
\caption{\label{tab:macro_coef_forecast_full_w105}Aggregation and Selection: Full set, Forecast, $W = 105$}
\begin{threeparttable}
\footnotesize
\setlength{\tabcolsep}{6pt}
\begin{tabular}[t]{lccc}
\toprule
& Aggregation & Lag & Lag \\[-1ex] Variable & rate &  & retention \\
\midrule
CMRMTSPL & -- & 1 & 100.0\% \\
\addlinespace
INDPRO & -- & 1 & 100.0\% \\
\addlinespace
HSN1F & -- & 1 & 96.2\% \\
\addlinespace
RSAFS & -- & 1 & 92.3\% \\
\addlinespace
PAYEMS & -- & 1 & 92.3\% \\
\addlinespace
CCSA & -- & 1 & 73.1\% \\
\addlinespace
DFXARC1M027SBEA & -- & 1 & 69.2\% \\
\addlinespace
SP500 & -- & 12 & 46.2\% \\
\addlinespace
PERMIT & -- & 2 & 34.6\% \\
\bottomrule
\end{tabular}
\end{threeparttable}
\end{table}

\begin{table}[H]
\centering
\caption{\label{tab:macro_coef_forecast_reduced_w66}Aggregation and Selection: Reduced set, Forecast, $W = 66$}
\begin{threeparttable}
\footnotesize
\setlength{\tabcolsep}{6pt}
\begin{tabular}[t]{lccc}
\toprule
& Aggregation & Lag & Lag \\[-1ex] Variable & rate &  & retention \\
\midrule
RSAFS & 44.6\% & 1 & 100.0\% \\
 & & 2--3 & 47.7\% \\
\addlinespace
INDPRO & 58.5\% & 1 & 100.0\% \\
 & & 2--3 & 72.3\% \\
\addlinespace
NFCI & 69.2\% & 1--4 & 72.3\% \\
 & & 5--12 & 50.8\% \\
\addlinespace
PERMIT & 64.6\% & 1 & 64.6\% \\
 & & 2 & 66.2\% \\
 & & 3 & 64.6\% \\
\bottomrule
\end{tabular}
\end{threeparttable}
\end{table}

\begin{table}[H]
\centering
\caption{\label{tab:macro_coef_forecast_reduced_w105}Aggregation and Selection: Reduced set, Forecast, $W = 105$}
\begin{threeparttable}
\footnotesize
\setlength{\tabcolsep}{6pt}
\begin{tabular}[t]{lccc}
\toprule
& Aggregation & Lag & Lag \\[-1ex] Variable & rate &  & retention \\
\midrule
RSAFS & -- & 1 & 100.0\% \\
\addlinespace
PERMIT & 96.2\% & 1 & 96.2\% \\
 & & 2 & 100.0\% \\
 & & 3 & 96.2\% \\
\addlinespace
INDPRO & -- & 1 & 100.0\% \\
 & & 2 & 53.8\% \\
\addlinespace
PAYEMS & -- & 1 & 84.6\% \\
\addlinespace
ICSA & -- & 1 & 34.6\% \\
 & & 2 & 73.1\% \\
 & & 3--4 & 34.6\% \\
 & & 9 & 73.1\% \\
 & & 10 & 42.3\% \\
 & & 11 & 53.8\% \\
 & & 12 & 38.5\% \\
\addlinespace
AMTMNO & -- & 1 & 61.5\% \\
\bottomrule
\end{tabular}
\end{threeparttable}
\end{table}

\end{document}